\documentclass[11pt,reqno]{amsart}
\usepackage[margin=1in]{geometry}
\usepackage{comment}    
\usepackage{lmodern}
\usepackage{bm}         
\usepackage{dsfont}     
\usepackage{braket}     
\usepackage{lipsum}     
\usepackage{amsmath,amssymb,amsthm}  
\usepackage{mathtools}               
\usepackage{esvect}                  

\usepackage[ruled,vlined]{algorithm2e}
\DontPrintSemicolon
\usepackage{xcolor, graphicx}
\usepackage{subcaption}
\usepackage{tikz}
\usetikzlibrary{shapes.geometric, arrows.meta, positioning, quantikz2}

\definecolor{darkgreen}{rgb}{0.0,0.72,0.06}
\definecolor{cerisepink}{rgb}{0.93,0.23,0.51}
\definecolor{red-brown}{rgb}{0.65,0.16,0.16}
\definecolor{blue}{rgb}{0,0,1}

\usepackage{enumitem}
\setlist[enumerate]{leftmargin=*}
\setlist[itemize]{leftmargin=*}
\usepackage[colorlinks]{hyperref}
\usepackage[capitalize,noabbrev]{cleveref}
\hypersetup{
    colorlinks=true,
    citecolor=darkgreen,
    linkcolor=blue,
    filecolor=magenta,
    urlcolor=cerisepink
}

\usepackage[backend=biber,style=numeric-comp,sorting=none,url=false,maxbibnames=99,doi=false]{biblatex}
\newcounter{mainthm} 

\theoremstyle{plain}
\newtheorem{thm}[mainthm]{Theorem}
\newtheorem{lem}[mainthm]{Lemma}
\newtheorem{prop}[mainthm]{Proposition}
\newtheorem{cor}[mainthm]{Corollary}

\theoremstyle{definition}
\newtheorem{defn}[mainthm]{Definition}
\newtheorem{assumption}[mainthm]{Assumption}

\newtheorem{example}[mainthm]{Example}
\newtheorem{remark}[mainthm]{Remark}
\newtheorem{result}[mainthm]{Result}

\crefname{thm}{Theorem}{Theorems}
\crefname{lem}{Lemma}{Lemmas}
\crefname{prop}{Proposition}{Propositions}
\crefname{cor}{Corollary}{Corollaries}
\crefname{defn}{Definition}{Definitions}
\crefname{assumption}{Assumption}{Assumptions}
\crefname{notation}{Notation}{Notations}
\crefname{remark}{Remark}{Remarks}
\crefname{example}{Example}{Examples}
\crefname{fact}{Fact}{Facts}
\crefname{result}{Result}{Results}
\crefname{ques}{Question}{Questions}

\newcommand{\C}{\mathbb{C}}

\begin{document}

\title{Non-Uniform Quantum Fourier Transform}

\author[J. Aftab]{Junaid Aftab}
\address{Department of Mathematics, University of Maryland, College Park, MD 20742, USA}
\email{junaida@umd.edu}

\author[Y. Khoo]{Yuehaw Khoo}
\address{Department of Statistics, University of Chicago, IL 60637, USA}
\email{yuehaw.khoo@uchicago.edu}

\author[H. Yang]{Haizhao Yang}
\address{Departments of Mathematics and Computer Science, University of Maryland, College Park, MD 20742, USA}
\email{haizhao.yang@umd.edu}

\date{\today}

\begin{abstract}
The Discrete Fourier Transform (DFT) is fundamental to the analysis of uniformly sampled signals; however, many practical settings involve non-uniform sampling, necessitating the Non-Uniform Discrete Fourier Transform (NUDFT). While efficient quantum algorithms for the standard DFT are well established, a systematic treatment of the non-uniform case remains limited.
We present a quantum algorithm for the Non-Uniform Quantum Fourier Transform (NUQFT) based on a low-rank approximation. This factorization is translated into an explicit quantum algorithm, yielding an $\epsilon$-accurate block encoding with rigorously controlled errors arising from both classical truncation and quantum implementation.
The resulting complexity scales polylogarithmically in the target precision, quadratically in the number of qubits, and logarithmically in a geometry-dependent conditioning parameter induced by the non-uniform sampling grid.
Numerical experiments validate the predicted low-rank truncation behavior, confirm the mild precision dependence on the geometry parameter, and demonstrate small-scale implementations of the key unitary primitives.
\end{abstract}

\maketitle 


\section{Introduction}
The Discrete Fourier Transform (DFT), introduced by Gauss and later popularized by Cooley and Tukey~\cite{cooley1965algorithm}, provides a framework for computing Fourier coefficients from uniformly sampled data and underpins applications in signal processing, data compression, and numerical solutions of differential equations (see, e.g.,~\cite{tan2018digital,shen2011spectral,briggs1995dft}). In practice, however, uniform sampling is often infeasible due to physical constraints, adaptive sampling, or irregular geometries. The Non-Uniform Discrete Fourier Transform (NUDFT)~\cite{bagchi2012nonuniform} generalizes the DFT to arbitrary sampling grids, enabling spectral analysis on non-uniform data and irregular computational domains, with important applications in imaging and numerical methods (see, e.g.,~\cite{qin2013robust,liu2000applications}). Consequently, a broad class of efficient classical algorithms for computing the NUDFT has been developed (see, e.g.,~\cite{dutt1993fast,ware1998fast,potts2001fast,fessler2003nonuniform,greengard2004accelerating,lee2005type,ruiz2018nonuniform,wilber2025superfast,li2025superfastdirectsolvertypeiii}).

The Quantum Fourier Transform (QFT) (see, e.g.,~\cite{coppersmith2002approximatefouriertransformuseful,nielsen2002quantum,lin2022lecture}) is the quantum analogue of the DFT and a central primitive in quantum algorithm design. It appears as a core subroutine in many quantum algorithms, including Shor’s factoring algorithm~\cite{shor1999polynomial}, Kitaev’s quantum phase estimation~\cite{kitaev1995quantummeasurementsabelianstabilizer}, quantum linear systems algorithms~\cite{harrow2009quantum}, ground-state energy estimation~\cite{abrams1999quantum}, and quantum arithmetic~\cite{draper2000additionquantumcomputer,ruiz2017quantum}. Recent work has explored several extensions of the QFT, such as multi-dimensional extensions, approximate implementations with reduced gate complexity, topology-aware variants adapted to specific hardware architectures, and fault-tolerant constructions for error-corrected quantum computation (see, e.g.,~\cite{pfeffer2023multidimensional,nam2020approximate,bäumer2025approximatequantumfouriertransform,jin2024optimizing,klaver2024swaplessimplementationquantumalgorithms,goto2014resource}). In contrast, quantum algorithms for Fourier transforms on non-uniform grids, particularly in a fully quantum setting, remain relatively underexplored.

In this work, we propose a quantum algorithm for the NUDFT, termed the Non-Uniform Quantum Fourier Transform (NUQFT). The construction builds on the low-rank approximation framework of Antol\'{i}n and Townsend~\cite{ruiz2018nonuniform} and translates this decomposition into an explicit quantum algorithm. The factorization is implemented using Quantum Signal Processing (QSP)~\cite{low2017qsp,gilyen2019qsvt} together with the Linear Combination of Unitaries (LCU)~\cite{ChildsWiebe2012} algorithms. This yields an explicit block encoding~\cite{gilyen2019qsvt} of the NUDFT matrix to arbitrary precision, with controlled approximation errors stemming from both the classical low-rank truncation and its quantum realization. We also provide explicit gate-level resource estimates for the resulting circuits under standard oracle access assumptions for non-uniform sampling points.

\subsection{Problem Statement} 
Let $\{x_j\}_{j=0}^{N-1} \subseteq \mathbb{C}$ denote samples taken at non-uniform locations $\{t_j\}_{j=0}^{N-1} \subseteq \mathbb{T}$, viewed as a discrete representation of an underlying function $f:\mathbb{T}\to\mathbb{C}$. The Non-Uniform Discrete Fourier Transform (NUDFT) of $\{x_j\}_{j=0}^{N-1}$, evaluated at non-uniform frequencies $\{\omega_k\}_{k=0}^{N-1} \subseteq \mathbb{R}$, is defined by
\begin{equation}\label{nudftIII-matrix}
\begin{pmatrix}
X_0 \\
X_1 \\
\vdots \\
X_{N-1}
\end{pmatrix}
=
\begin{pmatrix}
e^{-i \omega_0   t_0} & e^{-i \omega_0   t_1} & \cdots & e^{-i \omega_0   t_{N-1}} \\
e^{-i \omega_1   t_0} & e^{-i \omega_1   t_1} & \cdots & e^{-i \omega_1   t_{N-1}} \\
\vdots & \vdots & \ddots & \vdots \\
e^{-i \omega_{N-1} t_0} & e^{-i \omega_{N-1}   t_1} & \cdots & e^{-i \omega_{N-1}   t_{N-1}} \\
\end{pmatrix}
\begin{pmatrix}
x_0 \\
x_1 \\
\vdots \\
x_{N-1}
\end{pmatrix}.
\end{equation}
In the literature, \cref{nudftIII-matrix} is known as the Type--III NUDFT, and the matrix is denoted by $F_{\operatorname{III}}$. Two special cases arise depending on whether the sampling locations and frequency points are uniform or non-uniform:
\begin{enumerate}
    \item[(I)]\textbf{Type--I NUDFT}.  
    The sequence \( \{t_j\}_{j=0}^{N-1} \) is sampled uniformly at points \( t_j = 2\pi j/N \), but the
    sequence \( \{ \omega_k \}_{k=0}^{N-1} \) is sampled at potentially non-uniform points.  We denote the corresponding matrix as $F_{\operatorname{I}}$.

    \item[(II)] \textbf{Type--II NUDFT}. 
    The sequence \( \{t_j\}_{j=0}^{N-1} \) is sampled at potentially non-uniform points \( \{t_j\}_{j=0}^{N-1} \), but the sequence \( \{\omega_k\}_{k=0}^{N-1} \) is uniformly space with $\omega_k=k$.  We denote the corresponding matrix as $F_{\operatorname{II}}$.
\end{enumerate}

We argue in \cref{townsend-alg} that it suffices to develop a quantum algorithm for the Type--II case.
Hence, we focus on developing a quantum algorithm for the Type--II NUDFT in this work.

\subsection{Summary of Results} 
Let \(N = 2^n\) for \(n \ge 1\). Antol\'{i}n and Townsend~\cite{ruiz2018nonuniform} approximate \(F_{\operatorname{II}}\) by \(F_{\operatorname{II}} \approx \sum_{r=0}^{K-1} D_{\vec u_r} F D_{\vec v_r}\). (See later sections for the more general form of \(F_{\operatorname{II}}\).) Here, \(F\) is the standard discrete Fourier transform matrix, and \(D_{\vec u_r}\), \(D_{\vec v_r}\) are diagonal matrices. Assuming oracle access to the sampling points \(\{t_j\}_{j=0}^{2^n-1}\), our algorithm constructs quantum circuits \(U_{\vec u_r}\) and \(U_{\vec v_r}\) that provide block encodings of \(D_{\vec u_r}\) and \(D_{\vec v_r}\). Combining these with the standard block encoding of \(F\) and applying the Linear Combination of Unitaries (LCU) technique implements \(\sum_{r=0}^{K-1} D_{\vec u_r} F D_{\vec v_r}\). We provide non-asymptotic gate-count estimates (\cref{nuqftII-exact}), resource bounds for achieving \(\epsilon\)-accurate block encodings (\cref{nuqft-error}), and simplified complexity estimates (\cref{final-estimate-cor}), including explicit gate counts sufficient to obtain an \(\epsilon\)-accurate block encoding.

\begin{result}
[Asymptotic Version of \cref{nuqftII-exact} based on \cref{final-estimate-cor}]
\label{res:nuqft-informal}
Let $\epsilon > 0$ and let $n \in \mathbb{N}$ with $N = 2^n$. There exists a quantum circuit, $V_{\operatorname{II}}$, that implements an $\epsilon$-accurate block-encoding of $F_{\operatorname{II}}$.
Let 
\begin{align}
L_{n,\epsilon} := n + \log(1/\epsilon), \quad
\kappa := \max_j (1 - (y^*_j)^2)^{-1/2},
\end{align} 
where $y^*_j \in (-1,1)$ are points determined by the non-uniform sampling points. The circuit uses
$\mathcal{O}(L_{n,\epsilon} + \log(1 + \kappa L_{n,\epsilon}))$ qubits,
$\mathcal{O}(n + \log L_{n,\epsilon})$ Hadamard gates,
$\mathcal{O}(L^2_{n,\epsilon} + \log(1 + \kappa L_{n,\epsilon}))$ $\mathrm{CNOT}$ gates,
$\mathcal{O}(L_{n,\epsilon} + \log(1 + \kappa L_{n,\epsilon}))$ Toffoli and $\sigma_X$ gates, and
$\mathcal{O}(n^2 + L^2_{n,\epsilon} + \log(1 + \kappa L_{n,\epsilon}))$ single-qubit controlled rotation gates.
The circuit depth is $\mathcal{O}(L_{n,\epsilon} \log L_{n,\epsilon} + \log(1 + \kappa L_{n,\epsilon} ))$.
\end{result}

The resource bounds in \cref{res:nuqft-informal} depend on three parameters: the number of qubits $n$, the target accuracy $\epsilon$, and the geometry-dependent quantity $\kappa$. The accuracy parameter enters only through the logarithmic factor $\log(1/\epsilon)$, appearing additively with $n$ in the composite parameter $L_{n,\epsilon}$. Hence increasing precision affects gate counts and circuit depth only polylogarithmically.
The quadratic dependence on $n$ originates from the implementation of the quantum Fourier transform. All remaining gate counts scale as $L_{n,\epsilon}\log L_{n,\epsilon}$, except for a single $L_{n,\epsilon}^2$ term, up to additional logarithmic factors depending on $\kappa$.
The parameter $\kappa$, capturing conditioning from the non-uniform sampling points, contributes only through logarithmic factors of the form $\log(1+\kappa L_{n,\epsilon})$. In particular, no resource depends polynomially on $\kappa$, and its contribution remains subdominant whenever $\kappa$ grows at most exponentially in $n$ or $\log(1/\epsilon)$.

\subsection{Related Works} 
In contrast to the extensive literature on the QFT, little work has addressed quantum implementations of Fourier transforms on non-uniform grids. To the best of our knowledge, the closest related proposal is due to Gyongyosi and Imre~\cite{gyongyosi2012improvementquantumfouriertransform}. That work introduces a method termed ``Quantum-SVD,'' in which a standard QFT on a uniform grid is followed by a classical singular value decomposition--based interpolation step to approximate Fourier coefficients at non-uniform sampling points. As a result, the procedure is hybrid: the Fourier transform is quantum, while the handling of non-uniformity is performed classically. By contrast, the approach developed here provides a fully quantum implementation based on a block-encoding framework. Our work can also be placed in the broader context of efforts to develop circuit-level quantum analogues of classical integral transforms, including the quantum wavelet transform~\cite{fijany1998quantum,bagherimehrab2024efficient}, quantum Laplace transform~\cite{zylberman2024fastlaplacetransformsquantum,singh2025polylogarithmic}, quantum Hilbert transform~\cite{jha2025quantum,zhang2026efficient}, quantum cosine transform~\cite{klappenecker2001discrete,grigoryan2025new}, quantum Chebyshev transform~\cite{williams2023quantum}, and quantum Hermite transform~\cite{jain2025efficient}.

\subsection{Organization}
The remainder of this paper is organized as follows. \cref{prelims} reviews the necessary preliminaries, including the discrete and quantum Fourier transforms, as well as other quantum algorithmic primitives relevant to this study. \cref{townsend-alg} reviews the low-rank approximation approach to the non-uniform discrete Fourier transform. \cref{cheb-nuqfft} presents our quantum algorithm, and \cref{cheb-nuqfft-error} provides the corresponding error analysis. 
\cref{numerics-res} presents some numerical results.
Finally, \cref{conclusion} offers concluding remarks and discusses potential directions for future work.

\subsection*{Acknowledgments}
J. Aftab was supported by NSF DMS-2231533.
Y. Khoo was partially funded by NSF DMS-2339439, DOE DE-SC0022232, DARPA The Right Space HR0011-25-9-0031, and a Sloan research fellowship.
H. Yang was partially funded by NSF under awards IIS-2520978, GEO/RISE-5239902, the Office of Naval Research Award N00014-23-1-2007, DOE (ASCR) Award DE-SC0026052, and the DARPA D24AP00325-00.

\section{Preliminaries}\label{prelims}
This section presents the necessary preliminary material. \cref{not} introduces the notation used throughout the paper. \cref{qft} reviews the Quantum Fourier Transform (QFT). Finally, \cref{quant-algs-prims} summarizes the quantum algorithmic primitives that form the foundation of the proposed method.

\subsection{Notation}\label{not}
We briefly discuss the main notation used throughout this work.
\subsubsection{Standard Notation} 
Let $\mathbb{R}$, $\mathbb{C}$, and $\mathbb{N}$ denote the sets of real, complex and natural numbers, respectively. Let $\mathbb T$ denote the $1$-dimensional torus (i.e., the unit circle) identified with the interval~$[0,1]$ with its endpoints identified. If $f, g : \mathbb{N} \to \mathbb{R}^+$ are non-negative functions, we use the following standard notation from complexity theory:
\begin{enumerate}
  \item  $f(n) = \mathcal{O}(g(n))$ if and only if there exists a constant $C > 0$ and an integer $N \in \mathbb{N}$ such that $f(n) \leq C g(n)$ for all $n \geq N$.
  \item  $f(n) = \Omega(g(n))$ if and only if there exists a constant $C > 0$ and an integer $N \in \mathbb{N}$ such that $f(n) \geq C g(n)$ for all $n \geq N$.
  \item  $f(n) = \Theta(g(n))$ if and only if $f(n) = \mathcal O (g(n))$ and $f(n) = \Omega(g(n))$.
\end{enumerate}
For $r \in \mathbb N\cup \{0\}, T_r(x)$ denotes the degree-$r$ Chebyshev polynomial of the first-kind on $[-1,1]$ defined recursively by $T_0(x)=1, T_1(x)= x$ and
\begin{equation}
T_r(x)  = 2xT_{r-1}(x) - T_{r-2}(x), \quad r \geq 2.
\end{equation}
Similarly,  $U_r(x)$ denotes the degree-$r$ Chebyshev polynomial of the second-kind on $[-1,1]$ defined recursively by $U_0(x)=1, U_1(x)= 2x$ and
\begin{equation}
U_r(x)  = 2xU_{r-1}(x) - U_{r-2}(x), \quad r \geq 2.
\end{equation}

\subsubsection{Linear Algebra Notation}  
Vectors in $\mathbb R^n$ or $\mathbb C^n$ are denoted by lowercase Roman letters with an overhead arrow. The linear algebra notation used throughout this paper is summarized below:

\begin{enumerate}
    \item $\|\vec{x}\|_{1}$ denotes the $1$-norm of a vector $\vec{x}$. Similarly, $\|\vec{x}\|_{\infty}$ denotes the $\infty$-norm of $\vec{x}$.
    \item $\|A\|$ denotes the spectral norm of a matrix $A$, defined as its largest singular value.
    \item $\|A\|_{\operatorname{max}}$ denotes the maximum absolute value among the entries of $A$.
    \item The symbol $\circ$ denotes the elementwise (Hadamard) product, acting entry-wise on vectors or matrices of the same dimensions.
    \item If $\vec{x}$ is a vector and $f$ is a function, then $f(\vec{x})$ denotes the componentwise application of $f$ to $\vec{x}$.
\end{enumerate}

\subsubsection{Quantum Computing Notation}
A single-qubit quantum state is represented by a unit vector in \(\mathbb{C}^2\). In Dirac's notation, such a state is written as \(\ket{v}\) and referred to as a \emph{ket}, while its conjugate transpose, denoted \(\bra{v}\), is called a \emph{bra}. A single-qubit quantum gate is a unitary operator in \(\mathbb{C}^{2 \times 2}\). Standard examples include the Pauli matrices:
\begin{equation}
    \sigma_X = \begin{pmatrix} 0 & 1 \\ 1 & 0 \end{pmatrix}, \quad
    \sigma_Y = \begin{pmatrix} 0 & -i \\ i & 0 \end{pmatrix}, \quad
    \sigma_Z = \begin{pmatrix} 1 & 0 \\ 0 & -1 \end{pmatrix}.
\end{equation}
Single-qubit rotation gates about the \(X\), \(Y\), and \(Z\) axes are defined as follows:
\begin{equation}
    e^{-i \theta \sigma_X} = \begin{pmatrix} \cos\theta & -i \sin\theta \\ -i \sin\theta & \cos\theta \end{pmatrix}, \quad
    e^{-i \theta \sigma_Y} = \begin{pmatrix} \cos\theta & -\sin\theta \\ \sin\theta & \cos\theta \end{pmatrix}, \quad
    e^{-i \theta \sigma_Z} = \begin{pmatrix} e^{-i\theta} & 0 \\ 0 & e^{i\theta} \end{pmatrix}.
\end{equation}
We also denote the quantum gates  $e^{-i \theta \sigma_X}$, $e^{-i \theta \sigma_Y}$, and $e^{-i \theta \sigma_Z}$ as $R_X(\theta)$, $R_Y(\theta)$, and $R_Z(\theta)$, respectively. 
Another important example of a single-qubit gate is the Hadamard gate:
\begin{equation}
H =
\frac{1}{\sqrt{2}}
\begin{pmatrix}
1 & 1 \\
1 & -1
\end{pmatrix}
=
\begin{quantikz}
     & \gate{H} & \qw
    \end{quantikz}
\end{equation}
A $k$-qubit quantum state is a unit vector $\mathbb{C}^{2^k}$ and a $k$-qubit quantum gate is a unitary operator in \(\mathbb{C}^{2^k \times 2^k}\). An important example of a $k$-qubit gate is the $k$-fold tensor product of Hadamard gates:
\begin{equation}
H^{\otimes k} =
\frac{1}{\sqrt{2}}
\begin{pmatrix}
1 & 1 \\
1 & -1
\end{pmatrix}^{\otimes k}
=
\begin{quantikz}[wire types={b},classical gap=0.07cm]
&   \qwbundle{k} & \gate{H^{\otimes k}} & \qw
    \end{quantikz}
\end{equation}
It is often necessary to construct multi-qubit controlled gates. 
A generic $k$-qubit gate, \(U\), controlled on $l$ qubits, is denoted as follows:
\begin{equation}
C^l U
=
\begin{quantikz}[wire types={b,b}, classical gap=0.07cm]
&   \qwbundle{l} & \ctrl{1} & \qw \\
&   \qwbundle{k} & \gate{U} & \qw
\end{quantikz}
\end{equation}
An important example of a two-qubit gate is the CNOT gate, while a three-qubit example is the Toffoli (CCNOT) gate:
\begin{equation}
    \operatorname{CNOT} =
    \begin{pmatrix}
        1 & 0 & 0 & 0 \\
        0 & 1 & 0 & 0 \\
        0 & 0 & 0 & 1 \\
        0 & 0 & 1 & 0
    \end{pmatrix}
    =
    \begin{quantikz}
     & \ctrl{1} & \qw \\
     & \targ{} & \qw
    \end{quantikz}
    \;
    \operatorname{CCNOT} =
    \begin{pmatrix}
        1 & 0 & 0 & 0 & 0 & 0 & 0 & 0 \\
        0 & 1 & 0 & 0 & 0 & 0 & 0 & 0 \\
        0 & 0 & 1 & 0 & 0 & 0 & 0 & 0 \\
        0 & 0 & 0 & 1 & 0 & 0 & 0 & 0 \\
        0 & 0 & 0 & 0 & 1 & 0 & 0 & 0 \\
        0 & 0 & 0 & 0 & 0 & 1 & 0 & 0 \\
        0 & 0 & 0 & 0 & 0 & 0 & 0 & 1 \\
        0 & 0 & 0 & 0 & 0 & 0 & 1 & 0
    \end{pmatrix}
    =
    \begin{quantikz}
     & \ctrl{2} & \qw \\
     & \ctrl{1} & \qw \\
     & \targ{} & \qw
    \end{quantikz}
\end{equation}
Note that $\operatorname{CNOT}=C \sigma_X$ and 
$\operatorname{CCNOT}=C^2 \sigma_X$
Single-qubit controlled rotations also play an important role in quantum computation. In particular, the following controlled rotation operation arises in the implementation of the Quantum Fourier Transform.
\begin{equation}\label{rm-rot}
CR_m
=
\begin{pmatrix}
1 & 0 & 0 & 0 \\
0 & 1 & 0 & 0 \\
0 & 0 & 1 & 0 \\
0 & 0 & 0 & e^{-2\pi i / 2^{m}}
\end{pmatrix}
=
\begin{quantikz}[wire types={q,q},classical gap=0.07cm]
     & \ctrl{+1} & \qw \\
     & \gate{R_m} & \qw
    \end{quantikz} 
\end{equation}

\subsection{Quantum Fourier Transform}\label{qft}
Let $N \in \mathbb{N}$. The Discrete Fourier Transform (DFT) can be interpreted as the linear map
$\mathrm{DFT}_N : \mathbb{C}^N \to \mathbb{C}^N$
given by the matrix whose \((j,k)\)-entry is  
\begin{equation}
(\mathrm{DFT}_N)_{jk} =\frac{1}{\sqrt{N}} \omega_N^{jk},
\end{equation}
where \( \omega_N = e^{-2\pi i / N} \) is a primitive \( N \)-th root of unity and $0\leq j,k \leq N-1$. 
It is straightforward to verify that $\mathrm{DFT}_N$ is a unitary operator. The Quantum Fourier Transform (QFT) is the quantum analogue of the Discrete Fourier Transform and acts on quantum states by implementing $\mathrm{DFT}_N$ as a quantum circuit, denoted by $U_{\operatorname{QFT}}$. The algorithm realizes $U_{\operatorname{QFT}}$ through a sequence of single-qubit and two-qubit gates based on the following expression:
\begin{align}\label{product-rep-qft}
U_{\operatorname{QFT}} \ket{x} 
= \frac{1}{\sqrt{N}} \sum_{y=0}^{N-1} e^{-2 \pi i xy / N} \ket{y}  
 = \frac{1}{\sqrt{N}} \bigotimes_{\ell=0}^{n-1} 
   \left( \ket{0} + e^{-2\pi i \sum_{m=1}^{n-\ell} x_{n-\ell-m}/2^{m}} \ket{1} \right).
\end{align}
The tensor-product form in \cref{product-rep-qft} is important because it reveals a hierarchical dependence: the last output qubit depends on all \(n\) input qubits, while earlier outputs depend on progressively fewer qubits. This structure enables a simplified and more efficient circuit implementation.
In each of the \(n\) iterations in \cref{alg:qft}, a Hadamard gate is applied followed by progressively fewer controlled phase rotations, giving a total of
$\mathcal{O}(n^2)$
quantum gates. For a vector of size \(N = 2^n\), the QFT has computational complexity \(\mathcal{O}(\log^2 N) = \mathcal{O}(n^2)\), offering an exponential speedup over the classical FFT, which requires \(\mathcal{O}(N \log N) = \mathcal{O}(n2^n)\) operations.
\begin{algorithm}[h]
\caption{Quantum Fourier Transform (QFT) on an $n$-qubit register}
\KwIn{$n$-qubit quantum state $\ket{x} = \ket{x_0 x_1 \cdots x_{n-1}}$}
\KwOut{QFT-transformed quantum state $U_{\operatorname{QFT}}\ket{x}$}
\For{$\ell \gets 0$ \KwTo $n-1$}{
    Apply Hadamard gate $H$ to qubit $n-\ell$\; 
    \For{$m \gets 2$ \KwTo $n-\ell$}{
        Apply $CR_m$ (\cref{rm-rot}) with qubit $n-\ell-m$ (control) and qubit $n-\ell$ (target)\;
    }
}
\Return{$\mathrm{QFT \; transformed \; quantum \; state \;}  U_{\operatorname{QFT}}\ket{x}$}
\label{alg:qft}
\end{algorithm}

\subsection{Quantum Algorithm Primitives}
\label{quant-algs-prims}
We review the fundamental algorithmic primitives in quantum computation that underpin this work. In particular, we cover quantum arithmetic, the block encoding framework, quantum signal processing (QSP), and the Linear Combination of Unitaries (LCU) algorithm.

\subsubsection{Quantum Arithmetic}\label{quant-arithm}
We assume that all quantum algorithms take as input a fixed-precision binary number stored in an \(n\)-qubit quantum register, denoted \(\lvert x \rangle\). The leftmost \(m\) qubits encode the integer part, while the remaining \(n-m\) qubits encode the fractional part. When necessary, an additional qubit may be introduced to represent the sign of the number. Formally,
\begin{equation}
\lvert x \rangle 
= \ket{x^{(m-1)}} \otimes \cdots \otimes \ket{x^{(0)}} \otimes \ket{x^{(-1)} } \otimes \cdots \otimes  \ket{x^{(m-n)} },
\end{equation}
where each $x^{(j)} \in \{0,1\}$ and the number is represented as \begin{equation}
x = 
\underbrace{\sum_{j=0}^{m-1} x^{(j)} 2^j}_{\text{integer part}} + \underbrace{\sum_{j=-1}^{-(n-m)} x^{(j)} 2^{j}}_{\text{fractional part}}.
\end{equation}


Quantum algorithms for addition (see, e.g.,~\cite{draper2000additionquantumcomputer,Cuccaro2004quantumadd,draper2006logarithmic}) and multiplication (see, e.g.,~\cite{parent2017improvedreversiblequantumcircuits,gidney2019asymptoticallyefficientquantumkaratsuba,ramezani2023quantum}) have been extensively studied, with circuit complexities of \(\mathcal{O}(n)\) for addition and \(\mathcal{O}(n^2)\) for multiplication, measured in Toffoli and CNOT gates. Since the NUQFT algorithm considered here has overall complexity \(\Omega(n^2)\), it is sufficient to treat addition and multiplication as elementary primitives. See \cref{fig:mac} for a quantum oracle \(U_{\mathrm{AM}}\) that implements both addition and multiplication within a single quantum circuit.
\begin{figure}[h]
\centering
\begin{quantikz}[wire types={b,b,b}, classical gap=0.07cm]
\lstick{$\lvert x \rangle$} 
  & \gate[wires=3]{U_{\mathrm{AM}}} &  \rstick{$\ket{x}$} \\
\lstick{$\lvert y \rangle$} 
  & \qw                               & \rstick{$\ket{y}$} \\
\lstick{$\lvert z \rangle$} 
  & \qw                               & \rstick{$\ket{xy+z}$}
\end{quantikz}
\caption{Quantum oracle for multiplication and addition implementing  $\lvert x \rangle \lvert y \rangle \lvert z \rangle  \mapsto \lvert x \rangle \lvert y \rangle \lvert xy+z \rangle$.}
\label{fig:mac}
\end{figure}

Recent work by Burge et al.~\cite{burge2024quantumcordicarcsin} introduces a reversible quantum algorithm for computing the arcsine function using a quantum adaptation of the CORDIC (COordinate Rotation DIgital Computer) method. The algorithm iteratively approximates \(\arcsin(x)\) through a sequence of reversible rotations, enabling efficient implementation on quantum hardware. For a \(p\)-bit output, it requires \(\mathcal{O}(p)\) qubits and \(\mathcal{O}(p^2)\) CNOT gates, with circuit depth scaling as \(\mathcal{O}(p \log p)\), and achieves a worst-case approximation error of \(\mathcal{O}(2^{-p})\).

\begin{remark}
The implementation of Burge et al. employs $\mathcal{O}(p)$ additional qubit registers as scratch space for intermediate variables. As this overhead is absorbed into the asymptotic complexity, we use $p$ to denote the ancilla register that stores the final $p$-bit approximation whenever only the principal $p$ ancillas are being referenced.
\end{remark}

\subsubsection{Block Encoding}\label{block-encoding}
Quantum algorithms for matrix manipulation typically assume that a possibly non-unitary matrix can be embedded as a block of a unitary operator. The block-encoding framework formalizes this approach, enabling linear algebraic operations on a quantum computer by representing a subnormalized matrix \(A/\alpha\) as the upper-left block of a unitary. Since such embeddings may be approximate, we recall the following formal definition:

\begin{defn}\label{block-encoding-def}
Let $\alpha, \epsilon \in \mathbb{R}^+$. Given \(A \in \mathbb{C}^{r \times c}\), we say \(U \in \mathbb{C}^{d \times d}\) for $ d \geq  \max \{r,c\}$ is a $(\alpha, \epsilon)$ block encoding of \(A\) 
\begin{equation}
\| A - \alpha(B_L^\dagger \, U \, B_R) \|
\leq 
\epsilon,
\end{equation}
where \(B_L \in \mathbb{C}^{d \times r}\) and \(B_R \in \mathbb{C}^{d \times c}\) are the first \(r\) and \(c\) columns of the identity matrix, respectively. 
\end{defn}

We will often consider \(d\), \(r\), and \(c\) as powers of two. In this case, if $A \in \C^{2^m \times 2^n}$, we will say that  a unitary $U \in \C^{2^{m+a} \times 2^{n+a}}$ is a $(\alpha, a, \epsilon)$-block-encoding of $A$ if
\begin{equation}
\| A - \alpha( \bra{0}^{a} \otimes I_{2^n} )U( \ket{0}^{a} \otimes I_{2^m}) \|
\leq 
\epsilon
\end{equation}
In most cases, we have \(m = n\), and this is the case that will be considered below. For ease of notation, we write:
\begin{equation}
    U =
    \begin{pmatrix}
    A/\alpha & \ast\\
    \ast& \ast\\
\end{pmatrix}
\end{equation}
Let us consider an explicit computation to clarify the notation.
Let \(A \in \mathbb{C}^{2 \times 2}\) and \(\alpha = 1\). In this case, we have \begin{equation} (\langle 0| \otimes I_2) \, U \, (|0\rangle \otimes I_2)=\begin{pmatrix} I_2 & 0 \end{pmatrix}\begin{pmatrix} A & C \\ B & D \end{pmatrix}\begin{pmatrix} I_2 \\ 0 \end{pmatrix}=\begin{pmatrix} I_2 & 0 \end{pmatrix}\begin{pmatrix} A \\ B \end{pmatrix}= A. \end{equation} The case when \(A \in \mathbb{C}^{2^s \times 2^s}\) for $s > 1$ and \(a > 1\) is similar to the above. We will use the following lemma concerning block-encoding products of matrices in our analysis.

\begin{lem}
[Lemma 53 in \cite{gilyen2019qsvt}]
\label{block-lem}
If $U$ is an $(\alpha, a, \delta)$-block-encoding of an $s$-qubit operator $A$, and $V$ is a $(\beta, b, \epsilon)$-block-encoding of an $s$-qubit operator $B$, then $(I_b \otimes U)(I_a \otimes V)$ is an $(\alpha\beta, a+b, \alpha\epsilon + \beta\delta)$-block-encoding of $AB$, where $I_a$ and $I_b$ act on the ancilla qubits of $U$ and $V$ respectively.
\end{lem}

We will also use a result on block-encoding sparse matrices, which can be found in \cite[Lemma~48]{gilyen2019qsvt}.
For our purposes, it suffices to restrict attention to the case in which \(A\) is a sparse \(0\!-\!1\) matrix.
The corresponding specialization of this result is stated below.

\begin{lem}
\label{sparse-block-lem}
Let \(A \in \mathbb{C}^{2^n \times 2^n}\) be a \(d_r\)-row-sparse and \(d_c\)-column-sparse $0$-$1$ matrix. Suppose that we have access to the following \(2(n + 1)\)-qubit oracles:
\begin{align}
O_r: \lvert i \rangle \lvert k \rangle &\mapsto \lvert i \rangle \lvert r_{ik} \rangle, \quad 0 \le i < 2^n, \ 0 \le k < d_r, \\
O_c: \lvert l \rangle \lvert j \rangle &\mapsto \lvert c_{lj} \rangle \lvert j \rangle, \quad 0 \le j < 2^n, \ 0 \le l < d_c,
\end{align}
where \(r_{ik}\) is the index of the \(k\)-th non-zero entry of the \(i\)-th row of \(A\), or if there are fewer than \(k\) non-zero entries, then \(r_{ik} = k + 2^n\), and similarly \(c_{lj}\) is the index of the \(l\)-th non-zero entry of the \(j\)-th column of \(A\), or if there are fewer than \(l\) non-zero entries, then \(c_{lj} = l + 2^n\).  Additionally, assume we have access to a one-bit oracle
\begin{equation}
O_A: \lvert i \rangle \lvert j \rangle \lvert 0 \rangle \mapsto \lvert i \rangle \lvert j \rangle \lvert A_{ij} \rangle, \quad 0 \le i,j < 2^n,
\end{equation}
where \(A_{ij} \in \{0,1\}\).
Then one can implement a \((\sqrt{d_r d_c}, n + 3, 0)\)-block-encoding of \(A\) with a single use of \(O_r, O_c\), two uses of \(O_A\) and  additionally \(\mathcal O(n))\) one-qubit and two-qubit gates, while using \(\mathcal O(1)\)
ancilla qubits (which are discarded before the \(\epsilon\)-post-selection step).
\end{lem}

The proof of \cref{sparse-block-lem} follows the same argument as that of the general result in \cite[Lemma~48]{gilyen2019qsvt} and is therefore omitted.

\subsubsection{Quantum Signal Processing}\label{qsp}
Quantum signal processing (QSP) is a framework for implementing polynomial transformations of scalars that are block-encoded in a unitary operator.
In particular, QSP constructs \( 2 \times 2 \) unitary matrices whose entries are complex polynomials of a real scalar \( x \in [-1,1] \) by composing single-qubit rotations and phase gates.
It considers a single-qubit unitary operator that block-encodes \( x \in [-1,1] \) as follows:
\begin{equation}\label{x-encoding}
e^{i \arccos(x) \sigma_X}
=
\begin{pmatrix}
x & i\sqrt{1 - x^2} \\
i\sqrt{1 - x^2} & x
\end{pmatrix} 
, \quad \theta \in [0, \pi].
\end{equation}
For $\boldsymbol\varphi = (\varphi_0, \dots, \varphi_\ell)$, QSP implements the following quantum circuit:
\begin{equation}\label{qsp-circuit}
V_{\boldsymbol\varphi}(x) = e^{i\varphi_0 \sigma_Z} e^{i \arccos(x) \sigma_X} e^{i\varphi_1 \sigma_Z} e^{i \arccos(x) \sigma_X} \cdots e^{i \arccos(x) \sigma_X} e^{i\varphi_\ell \sigma_Z}.
\end{equation}
Polynomial functions realizable by \cref{qsp-circuit}  are completely characterized by the following result:
	\begin{prop}[Theorem 3 in \cite{gilyen2019qsvt}]
		\label{qsp-complex-poly}
		Let $\ell \in \mathbb{N}$ and $x \in [-1, 1]$.
		There exists $\boldsymbol \varphi = (\varphi_0, \varphi_1, \ldots, \varphi_\ell) \in \mathbb{R}^{\ell+1}$ such that
		\begin{equation}\label{equation:qsp-result}
			V_{\boldsymbol{\varphi}}(x) = e^{i\varphi_0\sigma_Z} \prod_{j=1}^{\ell} \left( e^{i \arccos(x) \sigma_X} e^{i\varphi_j\sigma_Z} \right)
			=
			\begin{pmatrix}
				p(x) & iq(x)\sqrt{1-x^2} \\
				iq^*(x)\sqrt{1-x^2} & p^*(x)
			\end{pmatrix}
		\end{equation}
		if and only if $p, q \in \mathbb{C}[x]$ such that:
		\begin{enumerate}
			\item[(i)] $\mathrm{deg}(p(x)) \leq \ell$ and $\mathrm{deg}(q(x)) \leq \ell - 1,$
			\item[(ii)] $p(x)$ has parity $\ell \mod 2$ and $q(x)$ has parity $\ell - 1 \mod 2$\footnote{
            A polynomial is said to have parity zero (one) if it is
            an even (odd) polynomial.
            },
			\item[(iii)] For all  $x \in [-1, 1]$, we have $|p(x)|^2 + (1 - x^2) |q(x)|^2 = 1$.
		\end{enumerate}
	\end{prop}
	


    \begin{example}\label{qsp-cheb}
        QSP is an important quantum circuit because it enables the implementation of various polynomial transformations. A simple choice of phase factors $(0, \dots, 0) \in \mathbb{R}^{\ell+1}$ ensures that $p(x) = T_{\ell}(x)$, where $T_{\ell}(x)$ is the degree-$\ell$ Chebyshev polynomial of the first kind. 
        We will use this result below.
        This observation has also recently been used to advocate QSP as a framework for function approximation (see, e.g., \cite{aftab2024korobov, aftab2025gpcpreprint}).
    \end{example}

\subsubsection{Linear Combination of Unitaries}\label{lcu} 
The Linear Combination of Unitaries (LCU) algorithm~\cite{Childs2012lcu} provides a framework for implementing a linear combination of unitary operators on a quantum computer.
Let \(M \in \mathbb{N}\), and let \(\alpha_1, \ldots, \alpha_M \in \mathbb{R}\) and \(U_1, \ldots, U_M\) be unitary operators.
The LCU method aims to implement the operator $U = \sum_{j=1}^{M} \alpha_j U_j$.
The algorithm assumes that the following two operators can be implemented. First, it assumes the existence of a state preparation (\(\operatorname{PREP}\)) oracle:
\begin{align}\label{F-LCU}
    \operatorname{PREP} \ket{0}  = \frac{1}{\sqrt{\| \vec \alpha \|_1 }} \sum_{j=1}^{M} \sqrt{\alpha_{j}} \ket{j},
\end{align}
where $\sqrt{\cdot}$ denotes the principal branch of the square root.
Second, it assumes that a multi-qubit controlled gate, also called a select (\(\operatorname{SEL}\)) oracle in the literature, can be implemented:
\begin{equation}\label{controlled-LCU}	
\operatorname{SEL} = \sum_{j=1}^{M} U_j \otimes \ket{j}\bra{j}.\end{equation}
Assuming access to the required operators, the LCU algorithm implements the unitary operator $U_{\text{LCU}} = \operatorname{PREP}^{\dagger} \cdot \operatorname{SEL} \cdot \operatorname{PREP}$.
See \cref{figure:lcu} for the corresponding quantum circuit. 
If \(\ket{\psi}\) denotes an input state, \(\ket{0}^{\log M}\) represents \(\mathcal{O}(\log M)\) ancilla qubits, we have
\begin{equation} U_{\text{LCU}} \ket{\psi} \ket{0}^{\log M} = \frac{1}{\|\vec \alpha\|_1} \left(\sum_{j=1}^{M} \alpha_j U_j\right) \ket{\psi} \ket{0}^{\log M} + \ket{\perp}, \end{equation}
where \(\ket{\perp}\) denotes a potentially non-normalized state satisfying
\((I_{2^k} \otimes \ket{0}^{\log M}\bra{0}^{\log M}) \ket{\perp} = 0\).
Note that the LCU algorithm implements a \((\|\vec \alpha\|_1, \log M, 0)\)-block-encoding of \(\sum_{j=1}^{M} \alpha_j U_j\).
\begin{figure}[h]
\centering
\begin{quantikz}[wire types={b,b},classical gap=0.07cm]
			\lstick{$\ket{\psi} \; \; $} &   \qwbundle{k} & &  & \gate[2]{\operatorname{SEL}} &  & \\
			\lstick{$\ket{0} \; \; $} & \qwbundle{\log M} &   & \gate{\operatorname{PREP}} &  & \gate{\operatorname{PREP}^\dagger} &  
		\end{quantikz}
	\caption{Quantum circuit $U_{\text{LCU}}$ implementing the LCU algorithm.
	}
	\label{figure:lcu}
\end{figure}

\section{Classical Algorithm for Non-Uniform Discrete Fourier Transform}\label{townsend-alg}
In this section, we review a classical algorithm for the nonuniform discrete Fourier transform based on a low-rank matrix factorization approach, as presented in~\cite{ruiz2018nonuniform}.

\subsection{Type--II NUDFT}
We first present the algorithm for the Type--II NUDFT, which assumes uniform frequency points given by $\omega_k = k$ for $k = 0, \ldots, N-1$.

\subsubsection{Perturbed Equispaced Grid Setting}
We begin by describing the algorithm in the special case where the nonuniform grid points are obtained by perturbing an equispaced grid. 
WLOG, let $\{ t_j \}_{j=0}^{N-1} \subseteq [0,1)$ such that\footnote{We can compensate for this by introducing a factor of $2 \pi$ in the complex exponential.}
\begin{equation}
\left| t_j - \frac{j}{N} \right| \leq \frac{\gamma}{N}, 
\end{equation}
for some $\gamma \in [0,1/2]$. The simple but important observation is that the entries of ${F}_{\operatorname{II}}$ can be written as
\begin{equation}
({F}_{\operatorname{II}})_{jk} = e^{-2\pi i t_j k} = e^{-2\pi i (t_j - \frac{j}{N}) k}  e^{-2\pi i \frac{j k}{N}},
\end{equation}
for $0 \leq j, k \leq N - 1$.
This observation leads to the matrix decomposition 
\begin{equation}
F_{\operatorname{II}} = A \circ F,
\end{equation}
where \(F\) is the DFT matrix, and the entries of \(A\) are given by \(A_{jk} = e^{-2\pi i (t_j - j/N ) k}\). 
The matrix $A$ is defined by sampling the bivariate function $(x,y) \mapsto e^{-2 \pi i x y}$ over the domain $[-\gamma/N, \gamma/N] \times [0, N]$. 
As shown in \cite{ruiz2018nonuniform}, the construction of a low-rank factorization of $A$ is based on approximating this function by a bivariate polynomial of degree $(K-1)$, which can then be extended to obtain a rank-$K$ approximation of $A$. 
This polynomial approximation is derived from a truncated Chebyshev expansion of $e^{-i x y}$ \cite[Appendix A]{ruiz2018nonuniform}. 
As a result, the resulting rank-$K$ approximation of $A$ takes the form
\begin{equation}\label{cheb-approx}
A_K = \sum_{r=0}^{K-1}{}^{'} \underbrace{\left[ \sum_{q=0}^{K-1}{}^{'} \alpha_{qr} \left( \exp(-i \pi N (\vec t - \vec e)) \circ T_p\left( \frac{N( \vec t - \vec e)}{\gamma} \right) \right) \right]}_{\vec u_r}  \underbrace{T_r \left(  \frac{ 2 \vec{\omega}^\top}{N} - \vec{1}^\top \right)}_{\begin{cases}
2 \vec v^\top_r & r=0 \\
\vec v^\top_r & r \geq 1.
\end{cases}} 
\end{equation}
Here \(\vec{e} = (0, \tfrac{1}{N}, \ldots, \tfrac{N-1}{N})^\top\), \(\vec{\omega} = (0, 1, \ldots, N-1)^\top\), and the primes on the summands indicate that the first term should be halved. 
The coefficients \( \alpha_{qr} \) are given by
\begin{equation}\label{a-coeffs}
\alpha_{qr} =
\begin{cases}
4 i^r  J_{\frac{q + r}{2}}\left(-\gamma \pi/2 \right)  J_{\frac{r - q}{2}}\left(- \gamma \pi/2 \right), & \text{if } \operatorname{mod}(|q - r|, 2) = 0, \\
0, & \text{otherwise}
\end{cases}
\end{equation}
where \( J_\nu(\cdot) \) denotes the Bessel function of the first kind of order \( \nu \). In order to eliminate the primes in the summand, we define a new quantity $\alpha'_{qr}$, which will replace $\alpha_{qr}$ in subsequent expressions.
\begin{equation}\label{alpha-prime-coeffs}
\alpha'_{qr} = 
\begin{cases}
   \frac{1}{4} \alpha_{qr} & \text{if } q=r=0, \\
   \frac{1}{2} \alpha_{qr} & 
   \text{if } q=0, r \neq 0
   \text{ or } r=0, q \neq 0,
   \\
   \alpha_{qr} & \text{otherwise}
\end{cases}
\end{equation}
The analysis in \cite{ruiz2018nonuniform} shows that it suffices to choose $K = \mathcal{O}\Big( \frac{\log(1/\epsilon'_{\operatorname{trunc}})}{\log\log(1/\epsilon'_{\operatorname{trunc}})} \Big)$ to ensure that the approximation satisfies $\|A - A_K\|_{\operatorname{max}} \leq \epsilon'_{\operatorname{trunc}}$ for some $\epsilon'_{\operatorname{trunc}} > 0$. Consequently, the matrix $F_{\operatorname{II}}$ can then be approximated as
\begin{equation}\label{approx-FII}
{F}_{\operatorname{II}} = (A \circ F) \approx ( A_K \circ  F) = \sum_{r=0}^{K-1} D_{\vec u_r} F D_{\vec v_r},
\end{equation}
where \( D_{\vec u_r} = \mathrm{diag}(\vec{u}_r) \) and \( D_{\vec v_r} = \mathrm{diag}(\vec{v}_r) \). 
The last equality follows from the identity
$(\vec u_r \vec v_r^{T}) \circ F = D_{\vec u_r}  F  D_{\vec v_r}$.
Since the complexity of applying \(F\) is \(\mathcal{O}(N \log N)\), the overall complexity of implementing the Type--II NUDFT in the special case where the non-uniform grid points \( \{ t_j \}_{j=0}^{N-1} \subseteq [0,1) \) are generated by perturbing an equispaced grid is
\begin{equation}\label{nudft-complexity}
\mathcal{O}\left(  \frac{N \log N \log(1/\epsilon'_{\operatorname{trunc}})}{ \log \log(1/\epsilon'_{\operatorname{trunc}})}\right).
\end{equation}

\subsubsection{General Nonuniform Node Distribution}
We now discuss how to implement the Type--II NUDFT when the nodes $\{ t_j \}_{j=0}^{N-1} \subseteq [0,1)$ are distributed arbitrarily.
The key idea is to reduce the problem to the equispaced case considered previously by identifying a sequence
$\{ s_j \}_{j=0}^{N-1} \subseteq \{0, \ldots, N-1\}$ such that
\begin{equation}\label{nearest-point-function}
s_j 
=
\begin{cases}
    m & \text{if } | t_j - \frac{m}{N} | \leq 1/2N \text{ and } m \neq N, \\
    0 & \text{if } | t_j - 1 | \leq 1/2N.
\end{cases}
\end{equation}
Note that \( s_n/N \) is the closest point to $t_j$ on the equidistant grid \(  \{ 0, 1/N, \ldots, (N-1)/N\} \subseteq [0,1) \) We have
\begin{equation}
(F_{\operatorname{II}})_{j,k} = e^{-2\pi i t_j k} = e^{-2\pi i (t_j - \frac{s_j}{N})k} e^{-2\pi i \frac{s_j k}{N}}.
\end{equation}
We obtain the matrix decomposition 
\begin{equation}
F_{\operatorname{II}} 
= B \circ F_{\vec{s}}
,
\end{equation}
where \( (F_{\vec{s}})_{j,k} = e^{-2\pi i s_j k/N} \), and \( B_{jk} = e^{-2\pi i (t_j - s_j/N) k} \). 
Note that $(F_{\vec{s}})_{j,k}$  is the \((k,s_j)\)-th entry of the DFT matrix, $F$. As before, the matrix $B$ admits a rank-$K$ approximation via a truncated bivariate Chebyshev approximation applied to each entry. 
The resulting rank-$K$ approximation of $B$ is then given by
\begin{equation}\label{cheb-approx-general}
 B_K = \sum_{r=0}^{K-1}
 \underbrace{\left[ \sum_{q=0}^{K-1} \alpha'_{qr} \left( \exp(-i \pi N (\vec t - \vec s/N)) \circ T_q\left( 2 N( \vec t - \vec s / N ) \right) \right) \right]}_{\vec u_r}  \underbrace{T_r \left(  \frac{ 2 \vec{\omega}^\top}{N} - \vec{1}^\top \right)}_{\begin{cases}
2 \vec v^\top_r & r=0 \\
\vec v^\top_r & r \geq 1
\end{cases}}.
\end{equation}
Here $\vec{s} = (s_0, s_1, \dots, s_{N-1})$ and $\vec{\omega} = (0, 1, \dots, N-1)$. The coefficient $\alpha'_{qr} \propto \alpha_{qr}$ in \cref{a-coeffs} is defined with the choice $\gamma = \frac{1}{2}$.
Note that $F_{\vec{s}} = F M_{\vec s}$, 
where 
\begin{equation}\label{coeff-s-matrix}
(M_{\vec s})_{ij} =
\begin{cases}
1, & \text{if } i = s_j, \\
0, & \text{otherwise},
\end{cases}
\end{equation}
for $0 \leq i,j \leq N-1$.
Consequently, $F_{\operatorname{II}}$ in the general case can then be approximated as
\begin{equation}\label{approx-FII-general}
{F}_{\operatorname{II}} = (B \circ F) \approx ( B_K \circ  F) = \sum_{r=0}^{K-1} D_{\vec u_r} F_{\vec s} D_{\vec v_r}
= \sum_{r=0}^{K-1} D_{\vec u_r} F M_{\vec s} D_{\vec v_r}.
\end{equation}
Once again, it suffices to choose $K = \mathcal{O}\left( \frac{\log(1/\epsilon'_{\operatorname{trunc}})}{\log\log(1/\epsilon'_{\operatorname{trunc}})} \right)$ to guarantee that $\|B - B_K\|_{\operatorname{max}} \le \epsilon'_{\operatorname{trunc}}$. The overall computational complexity of the algorithm remains as stated in~\cref{nudft-complexity}. 

\subsection{Reduction to Type--II NUDFT}
Algorithms for Type--I and Type--III NUDFTs can now be derived from those of the Type--II NUDFT.
\begin{enumerate}
\item[(I)] \textbf{Type--I NUDFT}.  
The Type-I NUDFT can be efficiently implemented using the Type-II NUDFT. In particular, the entries of the matrix $F_{\operatorname{I}}$ satisfy
\begin{equation}
(F_{\operatorname{I}})_{jk} = e^{-2\pi i \frac{j}{N} \omega_k} = e^{-2\pi i \frac{\omega_k}{N} j} = (F_{\operatorname{II}})_{kj},
\end{equation}
which shows that \( F_{\operatorname{I}} \) is the transpose of \( F_{\operatorname{II}} \), where now $\frac{\omega_0}{N}, \ldots, \frac{\omega_{N-1}}{N}$ assume the role of non-equispaced sample points in a Type-I NUDFT. Hence, $F_{\operatorname{I}}$ can be approximated as
\begin{equation}
{F}_{\operatorname{I}} 
=
{F}^T_{\operatorname{II}} 
\approx
\sum_{r=0}^{K-1} (D_{\vec u_r} F M_{\vec s} D_{\vec v_r})^{T}
=
\sum_{r=0}^{K-1} D_{\vec v_r} M^{T}_{\vec s} F^T D_{\vec u_r}
.
\end{equation}
Note that 
$F^{T}$ can be efficiently computed using the inverse FFT.
Therefore, the Type-I NUDFT may be realized via the Type-II NUDFT, inheriting the same computational complexity.

\item[(III)] \textbf{Type--III NUDFT}. 
The Type-III NUDFT can be also efficiently implemented using the Type-II NUDFT. It is shown in~\cite{ruiz2018nonuniform} the entries of the matrix $F_{\operatorname{III}}$ satisfy
\begin{equation}
(F_{\operatorname{III}})
= A \circ B \circ (F M_{\vec{s}})
:=A \circ B \circ (F_{\operatorname{I}})_{\vec{s}},
\end{equation}
where $A_{jk}= e^{-2\pi i (t_j - s_j/N)\omega_k}$, $B$ is a matrix of rank at most $2$ and $(F_{\operatorname{I}})_{\vec{s}}$ is obtained
extracting the columns indexed by $(s_0, \cdots , s_{N-1})$ from $F_{\operatorname{I}}$.
Therefore, the Type-III NUDFT may be realized via the Type-II NUDFT, inheriting the same computational complexity.
\end{enumerate}

\section{Quantum Algorithm for Non-Uniform Discrete Fourier Transform}\label{cheb-nuqfft} 
We now present a quantum algorithm for the Non-Uniform Quantum Fourier Transform (NUQFT) based on the low-rank factorization framework reviewed in \cref{townsend-alg}. Throughout, we assume that $N = 2^n$ for some $n \geq 1$ and restrict attention to the Type--II NUQFT. The overall approach may be summarized as follows:
\begin{enumerate}
    \item \textbf{Construction of $U_{\vec{v}_r}$.}
    We first construct quantum circuits that encode the coefficient vectors $\vec{v}_r$,
    defined as follows:
    \begin{equation}
    (\vec{v}_r)_j =
    \begin{cases}
        2^{-1} T_0(2 j / N - 1), & r = 0, \\
        T_r(2 j / N - 1), & r \ge 1.
    \end{cases}
    \end{equation}
    For each $r$, we construct a unitary operator $U_{\vec{v}_r}$ that encodes the entries of $\vec{v}_r$ as amplitudes of a quantum state.

    \item \textbf{Construction of $U_{\vec{u}_r}$.}
    We next construct quantum circuits that encode the coefficient vectors $\vec{u}_r$,
    defined as follows:
    \begin{equation}\label{ur-construct}
    (\vec{u}_r)_j = \sum_{q=0}^{K-1} \alpha'_{qr}\,
    e^{-i \pi N (t_j - s_j / N)}\,
    T_q (2 N (t_j - s_j / N) ).
    \end{equation}
    For each $r$, we construct a unitary operator $U_{\vec{u}_r}$ that encodes the entries of $\vec{u}_r$ as amplitudes of a quantum state.

    \item \textbf{Block-encoding Operators.}
    Using $U_{\vec{v}_r}$ and $U_{\vec{u}_r}$, we implement block encodings of the diagonal matrices
    $D_{\vec{v}_r}$ and $D_{\vec{u}_r}$. These are combined with a standard block encoding of the Fourier
    transform and the matrix $M_{\vec{s}}$ defined in \cref{coeff-s-matrix}.

    \item \textbf{Assembly via LCU.}
    Finally, we employ the Linear Combination of Unitaries (LCU) algorithm to combine the block-encoded
    matrix products, thereby implementing a block encoding of the Type--II non-uniform discrete Fourier transform.
\end{enumerate}

\begin{remark}
In what follows, a quantum register denotes a collection of qubits prepared in superposition over computational basis states. Ancilla qubits serve as auxiliary qubits for intermediate operations within a circuit, while a computational qubit is the qubit in which the final result is stored. Any deviation from this convention will be stated explicitly.
\end{remark}


\subsection{Construction of $U_{v_r}$}\label{4.1}
We first construct \(U_{\vec{v}_r}\). The case \(r=0\) is particularly simple, since \((\vec{v}_0)_{j,j} = \tfrac{1}{2}\) because \(T_0(x) \equiv 1\).
We prepare a quantum register encoding \(j\) in uniform superposition, together with an ancilla qubit whose \(\ket{0}\) amplitude represents the constant value \(\tfrac{1}{2}\) for all \(j\) in the superposition. 
The corresponding quantum circuit, denoted \(U_{\vec{v}_0}\), can be implemented directly.

\begin{prop}
Let $n \in \mathbb{N}$.  
The unitary operator $U_{\vec v_0}$ can be implemented using
$\mathcal{O}(n)$ Hadamard gates and $\mathcal{O}(1)$ single-qubit rotation gates.
\end{prop}

\begin{proof}
Prepare an $n$-qubit register in uniform superposition over all computational basis states using $n$ Hadamard gates.  
Next, prepare a single computational qubit in the state
$\frac{1}{2}\ket{0} + \frac{i \sqrt{3}}{2}\ket{1}$
using a single-qubit rotation of the form $e^{i \pi/3  \sigma_X}$.  
This defines the unitary $U_{\vec v_0}$. The corresponding pseudocode is given in~\cref{alg:u-v0}.
\end{proof}

\begin{algorithm}
\caption{Preparation of $U_{\vec{v}_0}$}
\KwIn{Integer $n \geq 1$}
\KwOut{Quantum state with $U_{\vec{v}_0}$ applied in superposition}
Initialize $(n+1)$ qubits $\ket{0} \ket{0}^n$\;
\For{$i \gets 2$ \KwTo $n+1$}{
    Apply Hadamard gate $H$ to qubit $i$ in $\ket{0}^n$\;
}
Apply single-qubit rotation $e^{i \pi/3 \sigma_X}$ to the computational qubit\;
\Return{$
\mathrm{Quantum \; state \; } 
U_{\vec{v}_0} \ket{0} \ket{0}^n = \frac{1}{\sqrt{2^{n}}} \sum_j ( \frac{1}{2}\ket{0} + \frac{i\sqrt{3}}{2}\ket{1}) \ket{j}
$}
\label{alg:u-v0}
\end{algorithm}

Before proceeding, we first prove a lemma on how to compute $1-x$ on a quantum computer, where $x \in (0,1)$.

\begin{lem}\label{lem:u-neg}
Let $k \in \mathbb{N}$ and let $x \in (0,1)$ be stored in binary in a $k$-qubit computational register.  
There exists a quantum circuit $U_{\mathrm{neg}}$ acting on the $k$-qubit register such that
$U_{\mathrm{neg}}\ket{x}  = \ket{1-x}$.
Moreover, $U_{\mathrm{neg}}$ can be implemented using $\mathcal{O}(k)\; \sigma_X$ gates, $\mathcal{O}(k)$ CNOT gates,
and $\mathcal{O}(k)$ Toffoli gates. 
\end{lem}

\begin{proof}
Let $x = 0.x_1 x_2 \dots x_k$ be a $k$-bit fractional number expressed in binary, so that $x = \sum_{i=1}^{k} x_i  2^{-i}$. We define the bitwise complement of $x$ as $\bar{x} = 0.\bar{x}_1 \bar{x}_2 \dots \bar{x}_k$, where $\bar{x}_i = 1 - x_i$ for each $i = 1, \dots, k$. Then
\begin{equation}
\bar{x} = \sum_{i=1}^{k} \bar{x_i} \cdot 2^{-i} = \sum_{i=1}^{k} (1 - x_i) 2^{-i} = \sum_{i=1}^{k} 2^{-i} - \sum_{i=1}^{k} x_i 2^{-i} = 1 - 2^{-k} - \sum_{i=1}^{k} x_i 2^{-i}
= 1 - x - 2^{-k}.
\end{equation}
$\bar{x}$ can be computed by applying the $\sigma_X$ gate to each of the $k$ qubits. It remains to add $2^{-k}$. This is a standard application of binary arithmetic and can be implemented by invoking the oracle $U_{\operatorname{AM}}$ using $\mathcal{O}(k)$ CNOT and Toffoli gates. The resulting quantum circuit is denoted by $U_{\operatorname{neg}}$.
\end{proof}

We now construct \(U_{\vec{v}_r}\) for \(r \geq 1\). 
In \cref{ur1}, a parameter \(p \in \mathbb{N}\) is selected to facilitate the approximation of \(\arccos(\cdot)\) on a quantum computer to arbitrary precision. 
This parameter is chosen freely at this stage, and will be fixed to a specific value later when required.
\begin{prop}\label{ur1}
Let $p, r, n \in \mathbb{N}$. $U_{\vec{v}_r}$ can be implemented using 
$\mathcal{O}(n)$ Hadamard gates, $\mathcal{O}(n+p^2)$ $\mathrm{CNOT}$ gates, $\mathcal{O}(n) \; \sigma_X$ gates and $\mathcal{O}(p)$ single-qubit controlled rotation gates generated by $\sigma_X$.
The corresponding quantum circuit has depth $\mathcal{O}(p\log p)$.
\end{prop}

\begin{proof}
Initialize two quantum registers of $n$ qubits, $\mathcal{O}(p)$ ancilla qubits, and one computational qubit, all initialized to $\ket{0}$.
We index the first quantum register of $n$ qubits from $0$ to $n-1$. 
Apply Hadamard gates to the first quantum register to construct the state $\frac{1}{\sqrt{2^n}} \sum_{j=0}^{2^n-1} \ket 0 \ket{j} \ket{0}^{n} \ket{0}^{p}$. Let
\begin{equation}\label{xj-comp}
x_j 
:= \frac{2j}{N} - 1 
=\frac{\sum_{k=0}^{n-1} j_k 2^{k}}{2^{n-1}} - 1
=\sum_{k=0}^{n-1} j_k 2^{k-n+1} - 1
=\sum_{k=0}^{n-2} j_k 2^{k-(n-1)} + j_{n-1} -1 \in [-1,1).
\end{equation} 
We first represent \(x_j\) in a fixed-point binary format. 
In what follows, we index the qubits in the second quantum register from \(1\) to \(n\). 
The construction proceeds as follows:
\begin{enumerate}
    \item For $k = 0$ to $n-2$, apply a CNOT gate with the control on the $k$-th qubit of the first quantum register and the target on the $m$-th qubit of the second quantum register, where $m + k = n - 1$. This step ensures that the first $n-1$ qubits in the second quantum register encode the fraction $\sum_{k=0}^{n-2} j_k 2^{k-n+1}$.

    \item If $j_{n-1} = 1$, the construction is complete. If $j_{n-1} = 0$, then we must compute $\sum_{k=0}^{n-2} j_k 2^{k-n+1} - 1 \in [-1,0)$.
This is achieved by applying a controlled unitary operator, with the control on the last qubit of the first quantum register and the target on the second quantum register. The operator acts only if the control qubit is in the state $\ket{0}$. The sign is stored in the last qubit of the second quantum register, with $\ket{1}$ denoting a negative sign, while the first $n-1$ qubits encode the absolute value
\begin{equation}
\left|\sum_{k=0}^{n-2} j_k 2^{k-(n-1)} - 1 \right|
= 1 - \sum_{k=0}^{n-2} j_k 2^{k-(n-1)}.
\end{equation}
This can be implemented using \cref{lem:u-neg}.
\end{enumerate}

We now construct the entry $(\vec{v}_r)_{j}$ in superposition using the quantum signal processing (QSP) algorithm. We implement $e^{i r \arccos(x_j) \sigma_X}$. This requires computing $\theta_j = \arccos(x_j)$, which can be obtained by first computing $\arcsin(x_j)$ using the algorithm of Burge et al.~\cite{burge2024quantumcordicarcsin}
as in \cref{quant-arithm}, and then applying $\arccos(x_j) = \pi/2 - \arcsin(x_j)$. Using the $p$ ancilla qubits to store the approximation $\widehat{\theta}_j = \widehat{\arccos}(x_j)$\footnote{
The $p$-bit binary representation of $\pi/2$ can be computed classically and used to obtain a $p$-bit approximation of $\pi/2 - \arcsin(x_j)$ in the $p$-qubit register containing $\widehat{\theta}_j = \widehat{\arcsin}(x_j)$. Since this operation does not change the overall circuit complexity, we assume that it is incorporated within the quantum circuit, $U_{\operatorname{arccos}}$.
}, we have
\begin{equation}\label{p-bit-approx}
| \widehat{\theta}_j - \theta_j | \leq 2^{-p+1}.
\end{equation}
We label this unitary operator $U_{\text{arccos}}$, which can be implemented using $\mathcal{O}(p^2)$ CNOT gates and a circuit of depth $\mathcal{O}(p \log p)$. Note that we can write
\begin{equation}
\widehat{\theta}_j 
= b^{(\theta_j)}_1 2^1 + b^{(\theta_j)}_0 2^0  + \sum_{k=-1}^{-(p-2)} b^{(\theta_j)}_k 2^k
:=\sum_{k=1}^{-(p-2)} b^{(\theta_j)}_k 2^k, 
\quad b^{(\theta_j)}_k \in \{0,1\},
\end{equation}
since $\theta \in [0,\pi]$. We can implement $e^{i r\widehat{\theta}_j\sigma_X}:=R_X(-r\widehat{\theta}_j)$ via the following formula:
\begin{equation}
R_X(-r\widehat{\theta}_j) = R_X\left(-r\sum_{k=1}^{-(p-2)} b^{(\theta_j)}_k 2^k\right) 
= 
\prod_{k=1}^{-(p-2)} R_X(-r 2^k)^{b^{(\theta_j)}_k}.
\end{equation}
If we adopt the convention of numbering the qubits in the $p$-qubit register from $1$ down to $-(p-2)$, we have
\begin{equation}\label{rotation-approximate}
R_X(-r\widehat{\theta}_j)=\prod_{k=1}^{-(p-2)} \left( \lvert0\rangle\langle0\rvert_{k}\otimes I+ \lvert1\rangle\langle1\rvert_{k}\otimes R_X(-r2^k)\right).
\end{equation}
This step requires $\mathcal{O}(p)$ applications of single-qubit controlled rotation gates generated by $\sigma_X$. Using \cite[Equation~(45)]{aftab2024korobov}, which provides the explicit expression for $e^{i r \arccos(x) \sigma_X}$, the resulting quantum state is given by
\begin{equation}
U_{\vec v_r} \ket 0 \ket 0^{2n} \ket 0^{p} =
\frac{1}{\sqrt{2^n}}\sum_{j=0}^{2^n-1} 
\left(T_r(x_j) \ket{0} + i\sqrt{1-x^2_j} S_{r-1}(x_j) \ket{1}\right)
\ket{j} \ket{x_j}^n \ket{\widehat{\arccos}(x_j)}^p.
\end{equation}
Here we have used that $T_r(x)=\cos(r \arccos(x))$.
See~\cref{alg:ur} and \cref{fig:11qubit_circuit} for the corresponding pseudocode and quantum circuit.
\end{proof}

\begin{algorithm}[h]
\caption{Construction of $U_{\vec{v}_r}$ for $r \geq 1$} 
\KwIn{$p,n, r \in \mathbb N$ and $U_{\text{neg}}$ from~\cref{lem:u-neg}} 
\KwOut{Quantum state with $U_{\vec{v}_r}$ applied in superposition} 
Initialize $1+2n+p$ qubits to $\ket{0}$\; 
Apply Hadamard gates to first register of $n$ qubits\; 
\For{$k \gets 0$ \KwTo $n-2$}{ 
    Apply CNOT with control $k$ in 1st register and target $m=n-k-1$ in 2nd register
} 
Apply C-$U_{\text{neg}}$ with control last qubit of 1st register and target first $n-1$ qubits on 2nd register\;
Apply CNOT with control last qubit of 1st register and target last qubit on 2nd register\;
Compute (an approximation to) $\theta_j = \arccos(x_j)$ using $U_{\text{arccos}}$ on $p$-qubit ancilla\; 

\For{$k \gets 1 \mathrm{\; down \; to \;} -(p-2)$}{ 
    \If{$b^{(\theta_j)}_k = 1 \mathrm{\; in \;} \theta_j \mathrm{\; register} $}{ 
        Apply controlled rotation $e^{ir 2^k \sigma_X}$ with control on ancilla qubit $k$\; 
    } 
} 
\Return{$
\mathrm{Quantum \; state \; } 
U_{\vec{v}_r} \ket 0 \ket 0^{2n} \ket 0^p$}
\label{alg:ur}
\end{algorithm}

The second quantum register of $n$ qubits, together with the $p$ ancilla qubits, can be uncomputed, yielding a unitary $U_{\vec v_r}'$ such that  
\begin{equation}
U_{\vec v_r}' \ket{0}\ket{0}^{2n}\ket{0}^{p} = \frac{1}{\sqrt{2^n}} \sum_{j=0}^{2^n-1} \left( T_r(x_j)\ket{0} + i\sqrt{1 - x_j^2} U_{r-1}(x_j)\ket{1} \right) \ket{j}\ket{0}^n\ket{0}^p.
\end{equation}
This step does not alter the computational complexity of the quantum circuit. For notational simplicity, we continue to denote \( U_{\vec v_r}' \) by \( U_{\vec v_r} \).

\subsection{Construction of $U_{u_r}$}\label{4.2}
We now construct $U_{\vec{u}_r}$. The expression for $(\vec{u}_r)_j$ is determined by $\vec{t} \in [0,1)^{2^n}$. Since $\vec{t}$ is arbitrary, we assume access to the following oracle as part of the algorithm.

\begin{assumption}\label{assumption-t}
For \(m \geq n\), we assume access to the oracle $O^{(m)}_{\vec t}$ such that
\begin{equation}
O^{(m)}_{\vec t} : \ket{i} \ket{0}^{m}   \mapsto  \ket{i} \ket{t_i^{(m)}}^m  = \bigotimes_{j=1}^{m} \ket{i} \ket{t_i^{(j)}} ,
\end{equation}
where \(t_i^{(m)} = \sum_{j=1}^{m} t_i^{(j)} 2^{-j} = 0.t_i^{(1)} \cdots t_i^{(m)}\) is the \(m\)-bit binary representation of \(t_i\). 
The parameter \(m\) is left arbitrary for now and will be fixed to a specific value as required.
\end{assumption}


\begin{figure}[t]
\begin{quantikz}[row sep=0.3cm,column sep=0.3cm, wire types={q,q,q,q,q,q,q,q,q,q,q,b},classical gap=0.07cm]
\lstick{$\ket{0}$} & \qw & \qw & \qw & \qw & \qw & \qw & \qw & \qw & \gate{R_X(\widehat{-r \arccos}(\vec x))} & \qw \\
\lstick[5]{
$\begin{array}{c}
\text{First} \\
\text{Register}
\end{array}
$
} 
\midstick{$\ket{0}$} & \gate{H} & \ctrl{+8} & \qw & \qw & \qw  & \qw  & \qw & \qw & \qw & \qw \\
\midstick{$\ket{0}$} & \gate{H} & \qw & \ctrl{+6} & \qw  & \qw & \qw & \qw & \qw & \qw & \qw  \\
\midstick{$\ket{0}$} & \gate{H} & \qw & \qw  & \ctrl{+4} & \qw & \qw & \qw & \qw & \qw & \qw  \\
\midstick{$\ket{0}$} & \gate{H} & \qw & \qw & \qw & \ctrl{+2} & \qw & \qw & \qw & \qw & \qw \\
\midstick{$\ket{0}$} & \gate{H} & \qw & \qw & \qw & \qw & \octrl{+1} & \octrl{+5}  & \qw & \qw & \qw \\
\lstick[5]{
$\begin{array}{c}
\text{Second} \\
\text{Register}
\end{array}
$
} 
\midstick{$\ket{0}$} & \qw & \qw  & \qw & \qw & \targ{} & \gate[wires=4]{U_{\text{neg}}} & \qw & \gate[wires=6]{U_{\text{arccos}}} & \qw & \qw \\
\midstick{$\ket{0}$} & \qw & \qw & \qw  & \targ{} & \qw & \qw & \qw & \qw & \qw & \qw  \\
\midstick{$\ket{0}$} & \qw & \qw & \targ{} & \qw  & \qw & \qw & \qw & \qw & \qw & \qw  \\
\midstick{$\ket{0}$} & \qw & \targ{} & \qw & \qw & \qw  & \qw & \qw & \qw & \qw & \qw   \\
\midstick{$\ket{0}$} & \qw & \qw & \qw & \qw & \qw & \qw  & \targ{}  & \qw & \qw & \qw  \\
\lstick{$\ket{0}^{p}$} 
& \qwbundle{p} & \qw & \qw & \qw & \qw & \qw & \qw & \qw & \ctrl{-11} & \qw  \\
\end{quantikz}
\caption{Quantum circuit implementing \(U_{\vec{v}_r}\) for \(r \geq 1\) in the special case \(n = 5\). 
}
\label{fig:11qubit_circuit}
\end{figure}

\begin{remark}
In practice, such oracles correspond to data structures that allow efficient superposition queries to classical datasets (e.g., QRAM). While building these structures efficiently remains an active research topic, this assumption is standard in the design of quantum linear-algebra algorithms.
\end{remark}

The first step is to compute $\vec{s} = (s_0, s_1, \dots, s_{2^n-1})$, as defined in \cref{nearest-point-function}. The vector $\vec s$ can be computed as
\begin{equation}
    s_i = 
    \begin{cases}
    \mathrm{round}(N t_i) & \text{if } \mathrm{round}(N t_i) < N   \\
    0 & \text{if } \mathrm{round}(N t_i) = N,
    \end{cases}
\end{equation}
where \(N = 2^n\) and \(\mathrm{round}(\cdot)\) denotes rounding to the nearest integer. In fact, we now explicitly construct a quantum circuit that implements the oracle
\begin{equation}\label{os-imp}
O^{(m)}_{\vec s} : \ket{i}  \ket{0}^{m}   \mapsto  \ket{i} \ket{s_i^{(m)}}^m =\bigotimes_{j=1}^{m} \ket{i} \ket{s_i^{(j)}},
\end{equation}
where $s_i^{(m)}$ is the $m$-bit representation of $s_i/N$. We first claim that the operation $N\vec{t}$ can be implemented efficiently.

\begin{lem}\label{implement-mul}
Let $n \in \mathbb N, m \geq n$ and $\vec{t} \in [0,1)^{2^n}$.
$N \vec{t}$ can be computed using $\mathcal{O}(1)$ queries to $O^{(m)}_{\vec{t}}$ and $\mathcal{O}(n)$ Hadamard gates on an $\mathcal{O}(n+m)$ qubit circuit.
\end{lem}

\begin{proof}
Initialize an \((n+m)\)-qubit quantum register in $\ket 0$.
Apply Hadamard gates to the first \(n\) qubits to prepare the uniform index superposition $\frac{1}{\sqrt{2^n}}\sum_{i=0}^{2^n-1} \ket{i}  \ket{0}^{m}  $. Querying \(O^{(m)}_{\vec{t}}\) yields  
\begin{equation}
O^{(m)}_{\vec{t}} \ket{0}^{n+m} 
= 
\frac{1}{\sqrt{2^n}}\sum_{i=0}^{2^n-1} \bigotimes_{j=1}^{m} \ket{i} \ket{t_i^{(j)}}.
\end{equation}
Multiplication by $N$ yields 
$
N t_i^{(m)} = t_i^{(1)} \cdots t_i^{(n)} . t_i^{(n+1)} \cdots t_i^{(m)},
$
where the first $n$ bits correspond to $\lfloor N t_i \rfloor$, 
and the remaining $m-n$ bits correspond to the fractional part.  No additional quantum gates are required.
\end{proof}


The next step is to round each $N t_i$ to the nearest integer, a nontrivial operation only when $\ket{t_i^{(n+1)}} = 1$. This rounding can be implemented as a controlled unitary, denoted $U_{\mathrm{round}}$. In this setting, rounding the binary string $t_i^{(1)} \cdots t_i^{(n)}$ amounts to computing the binary sum $t_i^{(1)} \cdots t_i^{(n)} \oplus 1$. Equivalently, this corresponds to the map $\ket{x} \mapsto \ket{x \oplus y}$, where $\ket{x} = \ket{t_i^{(1)} \cdots t_i^{(n)}}$ and $\ket{y} = \ket{t_i^{(n+1)}}$. This operation can be realized using the oracle $U_{\mathrm{AM}}$.

\begin{prop}\label{u-round}
Let $n \in \mathbb{N}$ and $m \geq n$. The oracle $O^{(m)}_{\vec{s}}$ can be implemented using a 
$\mathcal{O}(n+m)$-qubit quantum circuit, 
$\mathcal{O}(1)$ queries to the oracle $O^{(m)}_{\vec{t}}$, 
$\mathcal{O}(n)$ Hadamard gates and
$\mathcal{O}(m)$
Toffoli, $\mathrm{CNOT}$ and $\sigma_X$ gates.
\end{prop}

\begin{proof}
This follows directly from \cref{implement-mul} together with the use of $U_{\operatorname{AM}}$, noting that the $m$-bit representation of $s_i / N$ is stored in the $m$-qubit register upon application of circuit in \cref{fig:Us-oracle}.
\end{proof}


\begin{figure}[h]
\centering
\[
\begin{quantikz}[row sep=0.3cm, column sep = 0.3cm, wire types={b,b,q,b}, classical gap=0.07cm]
\lstick{$\ket{0}^{n}$} & \gate{H^{\otimes n}} & \gate[4]{O^{(m)}_{\vec t}} & \qw & \qw \\
\lstick{$\ket{0}^{n}$} & \qw & \qw & \gate[2]{U_{\text{AM}}} & \qw \\ 
\lstick{$\ket{0}$} & \qw & \qw & \qw & \qw \\
\lstick{$\ket{0}^{m-n-1}$} & \qw & \qw & \qw & \qw \\
\end{quantikz}
\]
\caption{Quantum circuit implementing $O^{(m)}_{\vec s}$.}
\label{fig:Us-oracle}
\end{figure}

Having constructed the oracle $O^{(m)}_{\vec{s}}$, the next step is to construct $U_{\vec{u}_r}$ for $r \ge 0$, which implements $(\vec{u}_r)_j$. We first establish a lemma analogous to \cref{lem:u-neg}, which enables the computation of $x - y$ for $x, y \in [0,1)$ using the quantum circuit $U_{\text{subtract}}$.

\begin{remark}
In what follows, when we state that a quantum register is ``multiplied by \(2^k\)'' for some \(k \ge 1\), we mean a reinterpretation of the same qubit register under a fixed-point encoding in which the binary point is shifted \(k\) positions to the right. No quantum gates are applied; this is purely a change in the numerical interpretation of the register contents. Similarly, ``division by \(2^k\)'' corresponds to shifting the binary point \(k\) positions to the left.
\end{remark}

\begin{lem}\label{frac-sub}
Let $m \in \mathbb{N}$ and let $x,y \in [0,1)$ be stored in $k$-qubit registers.
There exists a unitary operator $U_{\mathrm{subtract}}$ that computes the signed difference $x-y \in (-1,1)$
using $\mathcal{O}(k)$ $\mathrm{CNOT}$, Toffoli, and $\sigma_X$ gates.
\end{lem}

\begin{proof}
Let $x = 0.x_1 x_2 \dots x_k$ and $y = 0.y_1 y_2 \dots y_k$
be the $k$-bit fractional binary representations of $x$ and $y$, so that 
$x = \sum_{i=1}^{k} x_i 2^{-i}$ and $y = \sum_{i=1}^{k} y_i 2^{-i}$. Multiply $x$ and $y$ by $2^k$ to obtain integers
\begin{equation}
X = \sum_{i=1}^{k} x_i 2^{k-i} = x \cdot 2^k, \quad 
Y = \sum_{i=1}^{k} y_i 2^{k-i} = y \cdot 2^k .
\end{equation}
Note that $X, Y \in \{0,1,\dots,2^k-1\}$. Embed $X$ and $Y$ into $(k+1)$-qubit registers by adding a leading sign bit set to $0$,
thereby interpreting them as signed integers. We have
\begin{equation}
X = 0\cdot 2^k + \sum_{i=1}^k x_i 2^{k-i}, \quad
Y = 0\cdot 2^k + \sum_{i=1}^k y_i 2^{k-i}.
\end{equation}
Compute the bitwise complement of $Y$:
\begin{equation}
\overline{Y} 
= 1 \cdot 2^k + \sum_{i=1}^{k} \bar{y_i} 2^{k-i}
= 1 \cdot 2^k + \sum_{i=1}^{k} (1 - y_i) 2^{k-i}
= 1 \cdot 2^k + \sum_{i=1}^{k} 2^{k-i} - \sum_{i=1}^{k} y_i 2^{k-i}
= 2^{k+1} - Y - 1.
\end{equation}
\(\overline{Y}\) is obtained by applying a \(\sigma_X\) gate to each of the \(k+1\) qubits encoding \(Y\). Adding \(1\) yields
\(-Y := \overline{Y} + 1 = 2^{k+1} - Y\), ignoring overflow. In binary arithmetic, \(\overline{Y}+1\) is the two's complement of \(Y\) and represents its negative when the \((k+1)\)-st qubit serves as the sign bit.\footnote{The identity \(-Y := \overline{Y}+1\) is the standard two's complement representation.}
This justifies the notation \(-Y\).
Applying \(U_{\mathrm{AM}}\) to compute \(X + (-Y)\) produces the signed integer \(Z = X - Y\). Interpreting \(Z\) as a fixed-point number by dividing by \(2^k\) gives \(\frac{Z}{2^k} = x - y \in (-1,1)\). Since \(x,y \in [0,1)\), no overflow occurs. The circuit requires \(\mathcal{O}(k)\) CNOT, Toffoli, and \(\sigma_X\) gates.
\end{proof}

\begin{algorithm}
\caption{Computation of $x-y$ using signed fixed-point arithmetic ($U_{\mathrm{subtract}}$)}
\label{alg:subtraction}
\KwIn{
$k$-qubit registers $\ket{x}=\ket{x_1\cdots x_k}$ and $\ket{y}=\ket{y_1\cdots y_k}$
encoding $x,y\in[0,1)$
}
\KwOut{
$(k+1)$-qubit register encoding the signed fixed-point value $x-y$
}

Scale $\ket{x}$ and $\ket{y}$ by $2^k$ to obtain integer registers $\ket{X}$ and $\ket{Y}$\;
Append a leading qubit $\ket{0}$ to each register to form signed $(k+1)$-qubit registers\;
\For{$i \gets 1$ \KwTo $k+1$}{
    Apply $\sigma_X$ to qubit $i$ of $\ket{Y}$
}
Add $1$ to $\ket{Y}$ to obtain the two’s-complement encoding of $-Y$\;
Compute $\ket{X-Y}$ using  $U_{\mathrm{AM}}$\;
Interpret the result as a signed fixed-point number by dividing by $2^k$\;
\Return{$ \mathrm{Computational \; register \; representing \;} \ket{x-y}$}
\end{algorithm}

We construct \(U_{\vec{u}_r}\) for \(r \ge 0\) by first computing \(\vec{t} - \vec{s}/N\) using the preceding routines, and then implementing \(T_q\!\bigl(2N(\vec{t} - \vec{s}/N)\bigr)\) via quantum signal processing (QSP). Complex exponential factors are realized through \(\sigma_Z\)-generated rotations, and the weighted Chebyshev terms are combined using the linear combination of unitaries (LCU) technique. 
We further assume access to a suitable state-preparation oracle.

\begin{assumption}\label{PREP-Oracle}
We assume access to a state-preparation oracle
\begin{align}\label{lcu-needed-1}
\operatorname{PREP}_r \ket{0} = \frac{1}{\sqrt{\| \vec{\alpha}'_r \|_1}}
\sum_{q=0}^{K-1} \sqrt{\alpha'_{qr}}  \ket{q},
\end{align}
where $\alpha'_{qr}$ are defined in \cref{a-coeffs} and \cref{alpha-prime-coeffs} with $\gamma = 1/2$.
\end{assumption}


\begin{prop}\label{uu0}
Let $m,p,K \in \mathbb{N}$ and $r \in \mathbb N \cup \{ 0 \}$. $U_{\vec u_r}$ can be implemented using 
$\mathcal{O}(m+n+p+\log K)$ qubits,
$\mathcal{O}(1)$ queries to $O^{(m)}_{\vec t}$,
$\mathcal{O}(n)$ Hadamard gates,
$\mathcal{O}(m+p^2) \; \mathrm{CNOT}$ gates,
$\mathcal{O}(m)$ Toffoli gates,
$\mathcal{O}(m)$ $\sigma_X$ gates, and 
$\mathcal{O}(p+m)$ single-qubit controlled rotation gates.
The corresponding quantum circuit has depth $\mathcal{O}(m + p\log p + \log K)$.
\end{prop}

\begin{proof}
Consider a $(1+n+2m+\mathcal{O}(p))$-qubit system, with all qubits initialized to $\lvert 0 \rangle$. First, apply Hadamard gates to the quantum register of $n$ qubits, and then apply the oracles $O^{(m)}_{\vec{t}}$ and $O^{(m)}_{\vec{s}}$ on the two $m$-qubit ancilla registers to construct the state
\begin{equation}
\frac{1}{\sqrt{2^n}}
\sum_{i=0}^{2^n-1} 
\lvert 0 \rangle
\lvert i \rangle 
\lvert t^{(m)}_i \rangle^{m}
\lvert s^{(m)}_i/N \rangle^{m} 
\lvert 0 \rangle^{p} .
\end{equation}
Invoking \cref{frac-sub}, we compute $\vec{t} - \vec{s}/N$ in superposition within the quantum register containing $\vec{s}/N$.\footnote{A single ancilla qubit can be appended to encode the sign of $t_j - s_j / N$ without affecting the overall complexity. This ancilla qubit is used in \cref{frac-sub}.} We then compute an approximation of $\arccos (2N(\vec{t} - \vec{s}/N) )$. The multiplication by $2N$ is implemented reversibly as a constant-factor scaling on the fixed-point representation of $\vec{t} - \vec{s}/N$. Since $2N$ is a power of two, this operation reduces to a bit shift and therefore does not affect the overall complexity of the quantum circuit. By invoking $U_{\arccos}$ as in \cref{ur1}, we obtain an approximation of $\arccos (2N(\vec{t} - \vec{s}/N) )$ in superposition within the $p$-qubit register. We denote the quantum circuit implementing this scaling step followed by $U_{\arccos}$ by $U'_{\arccos}$. As in \cref{ur1}, we implement a controlled rotation $R_X\!\bigl(-q \arccos(2N(\vec{t} - \vec{s}/N))\bigr)$ to realize the Chebyshev polynomial $T_q(x)$. Similarly, we implement a controlled rotation $R_Z\!\bigl(\pi N (\vec{t} - \vec{s}/N)\bigr)$ to realize the phase factor $e^{-i \pi N (\vec{t} - \vec{s}/N)}$. Denote the unitary operator constructed above by $U_{q}$. Finally, using the state-preparation oracle in \cref{lcu-needed-1}, we apply the LCU algorithm to implement the desired expression via the operator $U_{\vec u_r} = \sum_{q=0}^{K-1} \alpha'_{qr} U_{q}$.
We need to append an $\mathcal{O}(\log K)$ ancilla qubits to invoke the LCU algorithm.
The complexity estimates follow from previous complexity estimates and the observation that
a $\mathcal{O}(\log K)$-qubit controlled gate acting can be implemented via a quantum circuit using only CNOT gates and single-qubit gates such that the circuit has depth \(\mathcal{O}(\log K)\) \cite{daSilva2022LinearDepthMultiQubit}.
If $y_i^{(m)}=t^{(m)}_i-s^{(m)}_i/N$, the resulting quantum state $U_{\vec u_r} \ket{0}^{(n+2m+p+1+\mathcal{O}(\log K))}$ is given by
\begin{equation}
\frac{1}{\sqrt{2^n}} \sum_{i=0}^{2^n-1}
\left ( \frac{1}{\|\vec \alpha'_r\|_1} \left(\sum_{q=0}^{K-1} \alpha'_{qr} U_{q} \ket{0}\right) \ket{0}^{\mathcal{O}(\log K)} + \ket \perp \right )
\ket{i}
\ket{t_i^{(m)}}^m  \ket{y_i^{(m)}}^m \ket{\widehat{\arccos}(y_i^{(m)})}^p ,
\end{equation}
where $\ket \perp$ is a state that is orthogonal to $(\sum_{q=0}^{K-1} \alpha'_{qr} U_{q} \ket{0}) \ket{0}^{\mathcal{O}(\log K)}$.
See \cref{fig:uur} for the corresponding quantum circuit and \cref{alg:u-ur} for the corresponding pseudocode. 
\end{proof}

Once again, the ancilla qubits can be uncomputed, yielding a unitary $U_{\vec u_r}'$ such that the output quantum state is
\begin{equation}
\frac{1}{\sqrt{2^n}} \sum_{i=0}^{2^n-1}
\left ( \frac{1}{\|\vec \alpha'_r\|_1} \left(\sum_{q=0}^{K-1} \alpha'_{qr} U_{q} \ket{0}\right) \ket{0}^{\mathcal{O}(\log K)} + \ket \perp \right )
 \ket{i} \ket{0}^{2m} \ket{0}^p.
\end{equation}
This step does not change the computational complexity of the quantum circuit. For notational simplicity, we continue to denote \( U_{\vec v_r}' \) by \( U_{\vec v_r} \).
Note that, in \cref{fig:uur}, we have departed from the usual convention of placing the ancilla qubits $\ket 0^{\mathcal{O}(\log K)}$ below the $n$-qubit register. 
This deviation was necessary to avoid awkward placement of the $U_{\mathrm{LCU}}$ gate in the circuit diagram.

\begin{algorithm}[t]
\caption{Implementation of $U_{\vec{u}_r}$}
\KwIn{Integers $n,m,p \geq 1$, system size $n+2m+p+1$ qubits}
\KwOut{Operator $U_{\vec{u}_r} = \sum_{q=0}^{K-1} \alpha'_{qr} U_{q}$ applied in superposition}

Initialize all $n+2m+p+1$ qubits to $\ket{0}$\;
Apply Hadamard gates $H^{\otimes n}$ to the $n$-qubit register\;
Apply oracles $O^{(m)}_{\vec{t}}$ and $O^{(m)}_{\vec{s}}$ to the two $m$-qubit registers\;
    Compute $\vec{t} - \vec{s}/N$ in the $\vec{t}$-register using 
    $U_{\text{subtract}}$ (see \cref{alg:subtraction})\;
    Apply $U_{\arccos}$ to approximate $\arccos(\vec{t} - \vec{s}/N)$ in the $p$-qubit register\;
    Apply controlled rotation $R_X(-q \arccos(2N(\vec{t} - \vec{s}/N)))$ to realize $T_q(x)$\;
    Apply controlled rotation $R_Z(\pi N (\vec{t} - \vec{s}/N))$ to realize $e^{-i \pi N (\vec{t} - \vec{s}/N)}$\;
    Append $\mathcal{O}(\log K)$ ancilla qubits\;
    Denote the resulting unitary as $U_{q}$\;
    Use LCU method to implement $U_{\vec u_r} = \sum_{q=0}^{K-1} \alpha'_{qr} U_{q}$ via $U_{\text{LCU}}$\;
    \Return{$U_{\vec{u}_r} \ket{0}^{ (n+2m+p+1+\mathcal{O}(\log K))}$}
\label{alg:u-ur}
\end{algorithm}

\begin{figure}[b]
\centering
\[
\begin{quantikz}[row sep=0.3cm, column sep = 0.3cm, wire types={q,b,b,b,b,b}, classical gap=0.07cm]
\lstick{$\ket{0}$} & \qw & \qw & \qw & \qw & \qw & \gate{R_X(-q\widehat{\arccos}(2 N \vec y))} & \gate{R_Z(\pi N ( \vec y))} & \gate[wires=2]{U_{\text{LCU}}} \\ 
\lstick{$\ket{0}^{\log K}$} & \qw & \qw & \qw & \qw & \qw & \qw & \qw & \qw \\
\lstick{$\ket{0}^{n}$} & \gate{H^{\otimes n}} & \gate[wires=2]{O^{(m)}_{\vec{t}}} & \gate[wires=3]{O^{(m)}_{\vec{s}}} & \qw & \qw & \qw & \qw & \qw \\
\lstick{$\ket{0}^{m}$} & \qw & \qw & \qw & \gate[wires=2]{U_{\text{subtract}}} & \qw & \qw & \qw & \qw \\      
\lstick{$\ket{0}^{m}$} & \qw & \qw &  & \qw & \gate[wires=2]{U'_{\arccos}} & \qw & \ctrl{-4} & \qw \\ 
\lstick{$\ket{0}^{p}$} & \qw & \qw &  & \qw & \qw & \ctrl{-5} & \qw & \qw \\
\end{quantikz}
\]
\caption{Quantum circuit implementing the operator $U_{\vec u_r} = \sum_{q=0}^{K-1} \alpha'_{qr} U_{q}$. 
We have abbreviated $\vec{t}-\vec{s}/N$ by $\vec{y}$.}
\label{fig:uur}
\end{figure}

\subsection{Block Encoding Operators}
We now use $U_{\vec{u}_r}$ and $U_{\vec{v}_r}$ within a block-encoding framework to implement the diagonal matrices $D_{\vec{u}_r}$ and $D_{\vec{v}_r}$. We assume that the ancilla qubits used to implement $U_{\vec{v}_r}$ and $U_{\vec{u}_r}$ have been uncomputed, as discussed above. These qubits will not be referenced in any equations or quantum circuits below; however, their usage is accounted for in the overall complexity analysis of the algorithm in \cref{nuqftII-exact}. Recall that $U_{\vec{v}_r}$ acts on $\ket{j}\ket{0}$ as
\begin{equation}
U_{\vec{v}_r} \ket{j}\ket{0} = \ket{j} (v_r(j)\ket{0} + \sqrt{1 - v_r(j)^2}\ket{1} ).
\end{equation}
The top-left block of $U_{\vec{v}_r}$ corresponds exactly to the desired diagonal matrix:
\begin{equation}
(\bra{0}  \otimes I) U_{\vec{v}_r} (\ket{0} \otimes I) = D_{\vec{v}_r}.
\end{equation}
This yields a $(1,1,0)$ block encoding of $D_{\vec v_r}$.

\begin{example}
Let $N = 4$, so that $D_{\vec{v}_r} = \mathrm{diag}\big(v_r(0), v_r(1), v_r(2), v_r(3)\big)$. We abbreviate $i \sqrt{1-v_r(0)^2}$ as $v'_r(0)$.
After uncomputing and ignoring any auxiliary registers, the matrix representation of $U_{\vec{v}_r}$ block encodes $D_{\vec{v}_r}$ as follows:
\begin{equation}
U_{\vec{v}_r} =
\left(
\begin{array}{cccc|cccc}
v_r(0) & & & & * & & & \\
& v_r(1) & & & & * & & \\
& & v_r(2) & & & & * & \\
& & & v_r(3) & & & & * \\ \hline
v'_r(0) & & & & * & & & \\
& v'_r(1) & & & & * & & \\
& & v'_r(2) & & & & * & \\
& & & v'_r(3) & & & & *
\end{array}
\right)
\end{equation}
We have adopted the basis convention $\{\ket{000}, \ket{010}, \ket{100}, \ket{110}, \ket{001}, \ket{011}, \ket{101}, \ket{111}\}$ in the matrix representation of $U_{\vec{v}_r}$ above, where the first qubit corresponds to the computational qubit. Projecting onto the computational state $\ket{0}$ extracts the top-left submatrix:
\begin{equation}
(\bra{0} \otimes I_4) U_{\vec{v}_r} (\ket{0} \otimes I_4)
=
\begin{pmatrix}
v_r(0) &  &  &  \\
 & v_r(1)  &  &  \\
  &   &  v_r(2)  &  \\
  &   &    & v_r(3)  \\
\end{pmatrix}
= D_{\vec{v}_r}.
\end{equation}
\end{example}

A similar argument yields a block encoding of $D_{\vec{u}_r}$. The only difference is that $D_{\vec{u}_r} / \lVert \vec{\alpha}'_{r} \rVert_{1}$ is block encoded, since $\vec{u}_r$ is implemented via the LCU method. Consequently, we obtain a $(\lVert \vec{\alpha}'_{r} \rVert_{1}, 1, 0)$ block encoding of $D_{\vec{u}_r}$. 
Furthermore, recall that $F_{\vec{s}} = F M_{\vec{s}}$, where $M_{\vec{s}}$ is defined in \cref{coeff-s-matrix}. Since $F$ is unitary, a block encoding of $F$ is immediate. 
Since $M_{\vec{s}}$ is a binary ($0$--$1$) matrix, we can apply \cref{sparse-block-lem} to obtain a block encoding of $M_{\vec{s}}$. Observe that $M_{\vec{s}}$ is a $1$-column-sparse matrix and that the oracle $O^{(m)}_{\vec{s}}$ effectively realizes the column-access oracle $O_c$. In this setting, it suffices to choose $m = n$. We further note that the row-access oracle $O_r$ corresponds to access to the inverse mapping of $\vec{t} \mapsto \vec{s}$. For generic $\vec{t}$ and $\vec{s}$, constructing $O_r$ from $O_{\vec{t}}$ and $O_{\vec{s}}$ may require additional pre-processing or data-structure assumptions; accordingly, we treat $O_r$ as an independent primitive.

\begin{assumption}\label{assumption-Or}
We assume black-box access to the row-access oracle \( O_r \) required by \cref{sparse-block-lem}. 
\end{assumption}

The oracle $O_A$ can be implemented efficiently. Specifically, we construct $O_A$ by appending an additional $n$-qubit register and applying $O^{(n)}_{\vec{s}}$ to obtain
\begin{equation}
\ket{0}^n \ket{i} \ket{j} \ket{0} \mapsto \ket{s_j} \ket{i} \ket{j} \ket{0}.
\end{equation}
Noting that $A_{i,j} = \langle s_j | i \rangle$, we compute the inner product of the $n$-qubit registers $\ket{s_j}$ and $\ket{i}$, which are computational basis states, and store the resulting binary outcome in the last qubit, thereby encoding $A_{i,j}$. More formally, the entry $A_{i,j}$ is computed as follows.

\begin{lem}\label{OA-oracle}
Let $\ket{i}$ and $\ket{s_j}$ be $n$-qubit computational basis states, and define
$A_{i,j} = \langle s_j | i \rangle \in \{0,1\}$.
Then there exists a unitary $U_{i,j}$ acting on an $2n+1$-qubit register 
such that
\[
U_{i,j} (\ket{i} \ket{s_j} \ket 0 ) = \ket{i}   \ket{s_j}   \ket{A_{i,j}}.
\]
Moreover, $U_{i,j}$ can be implemented using $\mathcal O(n) \; \operatorname{CNOT}$  gates, $\mathcal O(n)$ $\sigma_X$ gates, $\mathcal O(n)$ Toffoli gates and possibly $\mathcal O(n)$ ancilla qubits.
\end{lem}

\begin{proof}
For $r = 1, \dots, n$, apply a CNOT gate with the $r$-th qubit of $\ket{i}$ as the control and the $r$-th qubit of $\ket{s_j}$ as the target. After this operation, the $r$-th qubit of $\ket{s_j}$ stores $i_r \oplus s_{j,r}$. Next, apply a $\sigma_X$ gate to each qubit of $\ket{s_j}$ to obtain $i_r \,\text{XNOR}\, s_{j,r}$. Each qubit of $\ket{s_j}$ then stores $1$ if $i_r = s_{j,r}$ and $0$ otherwise. Finally, apply a $C^n$-$\operatorname{AND}$ gate with the $n$ qubits of $\ket{s_j}$ as controls and the last qubit as the target. This sets the last qubit to $1$ if and only if all transformed qubits of $\ket{s_j}$ are equal to $1$. Hence, the last qubit correctly encodes $A_{i,j}$. The procedure uses $\mathcal{O}(n)$ CNOT gates, $\mathcal{O}(n)$ $\sigma_X$ gates, and a single $C^n$-$\operatorname{AND}$ gate, which can be implemented using $\mathcal{O}(n)$ Toffoli gates and, if necessary, $\mathcal{O}(n)$ ancilla qubits.
\end{proof}

Consider the vector
\begin{equation}
\vec{c}_s =  (c^{s_{0}}_{0}, c^{s_{1}}_{1}, \dots, c^{s_{2^{n}-1}}_{2^{n}-1} )
\end{equation}
where \(c^{s_i}_i\) denotes the number of occurrences of \(s_{i}\) in \(\vec{s}\).
Noting that $M_{\vec s}$ is a $1$-column sparse matrix and a $\| \vec{c}_s \|_\infty$-row sparse matrix, we can construct a $\mathcal{O}(\sqrt{\| \vec{c}_s \|_\infty}, n+3, 0)$ block-encoding of $M_{\vec s}$ by invoking \cref{sparse-block-lem}.

\begin{remark}
The computational cost of the oracle for $M_{\vec{s}}$ discussed above is evaluated under the assumption that $m = n$. In what follows, $m$ is treated as a free parameter with $m \ge n$. Under this assumption, the computational complexity of implementing the block encoding described in \cref{OA-oracle} does not affect the overall computational complexity of the algorithm derived in \cref{nuqftII-exact}. Accordingly, it is omitted from the subsequent complexity analysis.
\end{remark}

\subsection{Assembly via LCU}
The final step is to implement $\sum_{r=0}^{K-1} D_{\vec{u}_r} F_{\vec{s}} D_{\vec{v}_r}$
via a quantum circuit. 
This is achieved by combining the block encodings constructed above with the LCU algorithm. 
In \cref{nuqftII-exact}, we present this procedure and derive the corresponding non-asymptotic complexity of the algorithm.

\begin{thm}\label{nuqftII-exact}
Let $n,m,p,K \in \mathbb{N}$ with $m \ge n$. 
There exists a quantum circuit, $V_{\operatorname{II}}$, implementing a block encoding of
$\sum_{r=0}^{K-1} D_{\vec u_r} F M_{\vec s} D_{\vec v_r}$
with normalization factor
\begin{equation}\label{norm-factor}
\alpha
:= K \sqrt{\|\vec{c}_s\|_\infty} \sum_{r=0}^{K-1} \|\vec{\alpha}'_r\|_1.
\end{equation}
and $\mathcal{O}(n + \log K)$ ancilla qubits. In other words, we have
\begin{equation}
(\bra 0^{\mathcal{O}(n + \log K)} \otimes I_{2^n}) V_{\operatorname{II}} (\ket 0^{\mathcal{O}(n + \log K)} \otimes I_{2^n})
=
\frac{1}{\alpha}
\sum_{r=0}^{K-1} D_{\vec u_r} F M_{\vec s} D_{\vec v_r}.
\end{equation}
The circuit uses
$\mathcal{O}(m+n+p+\log K)$ qubits,
$\mathcal{O}(n+\log K)$ Hadamard gates,
$\mathcal{O}(m+p^2)$ $\mathrm{CNOT}$ gates,
$\mathcal{O}(m)$ Toffoli gates,
$\mathcal{O}(m)$ $\sigma_X$ gates, and
$\mathcal{O}(p+m+n^2)$ single-qubit controlled rotation gates.
The circuit depth is $\mathcal{O}(m + p \log p + \log K)$.
\end{thm}

\begin{proof}
Fix $r \in \{0,\dots,K-1\}$.  
Prepare two single-qubit ancilla registers $a_r$ and $b_r$ together with
an $n$-qubit system register.
Observe that
$I_{a_r} \otimes I_{b_r} \otimes U_{\mathrm{QFT}}$
is a $(1,2,0)$ block encoding of $U_{\mathrm{QFT}}$.
By repeated application of \cref{block-lem}, the unitary
\begin{equation}
U_r
:= (I_{b_r} \otimes U_{\vec u_r})
(I_{a_r} \otimes I_{b_r} \otimes U_{\mathrm{QFT}})
(I_{a_r} \otimes U_{\vec v_r})
\end{equation}
is a $(\|\vec{\alpha}'_r\|_1,2,0)$ block encoding of
$D_{\vec u_r} F D_{\vec v_r}$. Using the construction described above for $M_{\vec s}$ and again invoking
\cref{block-lem}, we obtain a
\begin{equation}
\bigl(
\|\vec{\alpha}'_r\|_1
\sqrt{\|\vec{c}_s\|_\infty},
\; n+5,
\; 0
\bigr)
\end{equation}
block encoding of
$D_{\vec u_r} F M_{\vec s} D_{\vec v_r}$, which we continue to denote by $U_r$. To implement the sum
$\sum_{r=0}^{K-1} D_{\vec u_r} F M_{\vec s} D_{\vec v_r}$,
we apply the LCU algorithm.
Introduce $\mathcal{O}(\log K)$ ancilla qubits and define the select oracle
\begin{equation}
\operatorname{SEL} = \sum_{r=0}^{K-1} U_r \otimes \ket r \bra r.
\end{equation}
To implement the linear combination of unitaries (LCU), we assume access to a separate state-preparation oracle $\operatorname{PREP} := H^{\otimes \mathcal{O}(\log K)}$ 
Conjugating $\operatorname{SEL}$ with $\operatorname{PREP}$ yields the standard LCU unitary.  
The resulting construction is therefore a
\begin{equation}
\left(
\sum_{r=0}^{K-1} K \sqrt{\|\vec{c}_s\|_\infty} \|\vec{\alpha}'_r\|_1 ,\;
n+5+ \log K,
0
\right)
\end{equation}
block encoding of $\sum_{r=0}^{K-1} D_{\vec u_r} F M_{\vec s} D_{\vec v_r}$
~\cite[Lemma 52]{gilyen2019qsvt}.
The stated gate counts, qubit complexity, and circuit depth follow from the previous complexity estimates, taking into account the costs of implementing each $U_r$ as well as the overhead introduced by the LCU construction.
\end{proof}

Note that the $\mathcal{O}(n^2)$ term in the expression $\mathcal{O}(p + n^2)$ in \cref{nuqftII-exact} arises from the requirement to apply $\mathcal{O}(n^2)$ controlled-rotation gates in the standard quantum circuit implementation of the Quantum Fourier Transform. We therefore observe that this quadratic dependence constitutes the dominant contribution to the overall gate-complexity estimates.

\section{Error Analysis}\label{cheb-nuqfft-error}
We now present an error analysis of the quantum algorithm implementing the Type-II NUDFT.  

\subsection{Error Sources}
We first identify the sources of error arising in the quantum implementation of the algorithm.

\begin{enumerate}
    \item \textbf{Truncation Error}. 
    Antolín and Townsend \cite{ruiz2018nonuniform} approximate the NUDFT matrix by a rank-$K$ matrix from a Chebyshev expansion. To achieve $\|\cdot\|_{\max} \le \epsilon_{\mathrm{trunc}}$, choose
    $K = \mathcal{O}\!\left( \frac{\log(1/\epsilon'_{\mathrm{trunc}})}{\log \log (1/\epsilon'_{\mathrm{trunc}})} \right)$,
    for some $\epsilon'_{\mathrm{trunc}} > 0$.

    \item \textbf{Finite-Precision Oracle Error}.  
    The sampling points $t_j$ and $s_j$ are given with $m$-bit precision, introducing errors in $t_j - s_j/N$ that propagate to the complex exponentials and Chebyshev evaluations.

    \item \textbf{Arccosine Evaluation Error}. 
    Implementing QSP requires rotations about the $X$ axis with angles of the form $\arccos(x)$. Computing $\arccos(x)$ to $p$-bit precision introduces an approximation error of order $\mathcal{O}(2^{-p})$.

    \item \textbf{Block-Encoding Error}.  
    These errors accumulate in the block encoding of the constructed operators, reducing the overall accuracy of the resulting unitaries.
\end{enumerate}

\subsection{Error Bounds}
We now bound the total error of the quantum algorithm and specify the parameters $(K, m, p)$ so that the total error remains below the prescribed tolerance $\epsilon$. Let $\widetilde{V}_{\mathrm{II}}$ denote the approximately constructed version of the quantum circuit introduced in \cref{nuqftII-exact}. Moreover, recall that we have approximated 
$F_{\operatorname{II}} = A \circ (FM_{\vec s})$ on a quantum computer by implementing
\begin{equation}
\widetilde{F}_{\operatorname{II}} := A_K \circ (FM_{\vec s}) = \sum_{r=0}^{K-1} D_{\vec u_r} F M_{\vec s} D_{\vec v_r}.
\end{equation}
To bound the total error, we apply the triangle inequality
\begin{equation}\label{triangle-ineq}
\|\widetilde{V}_{\operatorname{II}} - F_{\operatorname{II}}\| 
\le  \|\widetilde{F}_{\operatorname{II}} - F_{\operatorname{II}}\| +
\|\widetilde{V}_{\operatorname{II}} - \widetilde{F}_{\operatorname{II}}\|
\end{equation}
The first term is determined by the truncation error. It has already been analyzed in \cite{ruiz2018nonuniform} in the $\| \cdot \|_{\operatorname{max}}$ norm. We now derive an estimate for it in the $\| \cdot \|$ norm based on the following lemma.
\begin{lem}\label{max-2norm-lem}
Let $B, C \in \mathbb{C}^{N \times N}$. We have
\begin{equation}
\| B \circ  C \| \le \sqrt{N} \| B \|_{\operatorname{max}} \|C\|.
\end{equation}
\end{lem}

\begin{proof}
Let $e_1, \dots, e_N$ be the standard basis vectors in $\mathbb{C}^N$, and let $v = \sum_{i=1}^N v_i e_i$ be any vector. For each $i$, we have
\begin{equation}
\|(B \circ C) e_i\|_2 \le \|B\|_{\operatorname{max}} \|C e_i\|_2 \le \|B\|_{\operatorname{max}}  \|C\|
\end{equation}
because $(B \circ C) e_i$ is obtained from $C e_i$ by multiplying each entry by a scalar of magnitude at most $\|B\|_{\operatorname{max}}$. By linearity and the triangle inequality, we have
\begin{equation}
\|(B \circ C)v\|_2 = \left\| \sum_{i=1}^N v_i (B \circ C) e_i \right\|_2 \le \sum_{i=1}^N |v_i|  \|(B \circ C) e_i\|_2 
\le \|B\|_{\operatorname{max}} \|C\| \|v\|_1.
\end{equation}
Since $\|v\|_1 \leq \sqrt{N} \|v\|_2$,  we obtain
\begin{equation}
\|(B \circ C)v\|_2 \le  \sqrt{N} \|B\|_{\operatorname{max}} \|C\| \|v\|_2
\end{equation}
for all $v \in \C^N$. The claim follows.
\end{proof}

We can use \cref{max-2norm-lem} to compute an upper bound on $K$. This is carried out in \cref{nuqft-error} below. 
Next, we bound the second term in \cref{triangle-ineq}, which will be used to determine the values of $m$ and $p$ in \cref{nuqft-error}. We assume that the block-encoding of $M_{\vec{s}}$ is exact. This assumption is justified because we have black-box access to the oracle $O_r$, and the oracles $O_c$ and $O_A$ (\cref{OA-oracle}) are constructed exactly. 
Furthermore, we assume that the quantum Fourier transform is implemented exactly, as we analyze the algorithm within the fault-tolerant regime. efine the normalization factors as
$
\alpha_{\vec u_r} = \|\vec \alpha'_r\|_1$
and 
$
\alpha_{\vec v_r} = 1$,
where $\vec \alpha'_r$ are the coefficients appearing in \cref{ur-construct}. The analysis in the previous section implies that $U_{\vec u_r}$ and $U_{\vec v_r}$ are exact block-encodings.
\begin{equation}
U_{\vec u_r} =
\begin{pmatrix}
D_{\vec u_r}/\alpha_{\vec u_r} & * \\
* & *
\end{pmatrix}, 
\quad
U_{\vec v_r} =
\begin{pmatrix}
D_{\vec v_r}/\alpha_{\vec v_r} & * \\
* & *
\end{pmatrix}.
\end{equation}
Since the errors introduced above render the block-encodings inexact in the construction of $\widetilde{V}_{\operatorname{II}}$, let $\widetilde U_{u_r}$ and $\widetilde U_{v_r}$ be the implemented block-encodings such that 
\begin{equation}
\widetilde U_{\vec u_r} =
\begin{pmatrix}
\widetilde D_{\vec u_r}/\alpha_{\vec u_r} & * \\
* & *
\end{pmatrix}, 
\quad
\widetilde U_{\vec v_r} =
\begin{pmatrix}
\widetilde D_{\vec v_r}/\alpha_{\vec v_r} & * \\
* & *
\end{pmatrix},
\end{equation}
with $\widetilde D_{\vec u_r}$ and $\widetilde D_{\vec v_r}$ obtained via finite-precision oracle access and approximate function evaluation. Let $\widetilde{F}$ and $\widetilde{M}_{\vec s}$ denote the corresponding exact block-encodings of $F$ and $M_{\vec s}$, respectively.
Note that $\widetilde{V}_{\operatorname{II}}$ can be written as
\begin{equation}
\widetilde{V}_{\operatorname{II}} = \sum_{r=0}^{K-1} \widetilde{D}_{\vec u_r}  \widetilde{F}  \widetilde{M}_{\vec s}  \widetilde{D}_{\vec v_r}.
\end{equation}
The normalization factor $1/\alpha$, defined in \cref{norm-factor}, is omitted in the following analysis, as it affects only the post-selection success probability. We bound the unitary circuit synthesis error of the unnormalized top-left block in order to determine the required parameters $(K,m,p)$ as functions of the target precision.
We have
\begin{align}
\|\widetilde{V}_{\operatorname{II}} - \widetilde{F}_{\operatorname{II}}\|
& = \left\| \sum_{r=0}^{K-1} \widetilde{D}_{\vec u_r} \widetilde{F} \widetilde{M}_{\vec s}  \widetilde{D}_{\vec v_r} - \sum_{r=0}^{K-1} D_{\vec u_r} F M_{\vec s} D_{\vec v_r} \right\|, \\
&\le \sum_{r=0}^{K-1} 
\|  (\widetilde{D}_{\vec u_r} - D_{\vec u_r}) \widetilde{F} \widetilde{M}_{\vec s} \widetilde{D}_{\vec v_r} \| 
+ \sum_{r=0}^{K-1} 
 \| D_{\vec u_r} F M_{\vec s} ( \widetilde{D}_{\vec v_r} - D_{\vec v_r} ) \|.
\end{align}
Here, the terms involving $\widetilde{F} - F$ and $\widetilde{M}_{\vec s} - M_{\vec s}$ vanish because we assume that $\|\widetilde{M}_{\vec s} - M_{\vec s}\| = 0 = \|\widetilde{F} - F\|$.
Note that $\|F\| = \|\widetilde{F}\| = 1, \| M_{\vec s} \| = \| \widetilde M_{\vec s} \| = \sqrt{\| \vec c_s \|_\infty}$, 
$\| D_{\vec u_r} \| \leq \| \vec \alpha'_r \|_1$ and 
$\| D_{\vec v_r} \| \leq 1$. Hence, we have
\begin{align}\label{second-term-to-bound}
\|\widetilde{V}_{\operatorname{II}} - \widetilde{F}_{\operatorname{II}}\|
\le \sqrt{\|\vec c_s\|_\infty} \sum_{r=0}^{K-1} 
\left( \| \widetilde{D}_{\vec u_r} - D_{\vec u_r} \| + \|\vec \alpha'_r\|_1  \| \widetilde{D}_{\vec v_r} - D_{\vec v_r} \| \right).
\end{align}
Since we have
\begin{align}
\| D_{\vec v_r} - \widetilde{D}_{\vec v_r}  \|
= \max_{j}  | (\vec v_r)_j - (\widetilde{\vec v}_r)_j   |, 
\quad
\| D_{\vec u_r} - \widetilde{D}_{\vec u_r} \|
= \max_{j} | (\vec u_r)_j - (\widetilde{\vec u}_r)_j  |, 
\end{align}
it suffices to bound the worst-case scalar error in a single diagonal entry. We first compute the error in the evaluation of the complex exponential and the Chebyshev polynomials, which will allow us to determine the worst-case scalar error in a single diagonal entry. We first bound the complex exponential term. In what follows, we assume that $m \geq n+1$. In this case, the oracle constructing $s_j$ is exact; hence, the only source of error is in $t_j$.

\begin{lem}\label{exp-error}
Let $m,n \in \mathbb N$ such that $m \geq n+1$ and $N=2^n$.
Let $y_j = t_j - s_j/N$ for $j = 0, \dots, N-1$, where $t_j$ and $s_j$ are provided by $m$-bit precision oracles. Let $\widetilde{y}_j$ denote the $m$-bit approximation of $y_j$.
The errors in the complex exponential satisfy
\begin{equation}
| e^{-i \pi N y_j} - e^{-i \pi N \widetilde{y}_j} | \le \pi N  2^{-m+1}
\end{equation}
uniformly in $j$.
\end{lem}

\begin{proof}
We have $|t_j - \widetilde{t}_j| \leq 2^{-m}$. Hence, we have
\begin{align}
 | e^{-i \pi N y_j} - e^{-i \pi N \widetilde{y}_j}  |
&= \left| e^{-i \pi N (t_j - s_j/N)} - e^{-i \pi N (\widetilde{t}_j - s_j/N)} \right| \\
&\le \pi N \bigl| (t_j - s_j/N) - (\widetilde{t}_j - s_j/N) \bigr| \le \pi N\, 2^{-m}
\end{align}
uniformly in $j$.  
\end{proof}

We now bound the error in the Chebyshev polynomials. Note that the Chebyshev polynomials are evaluated at two inputs,  $x_j := 2j/N - 1$ and $y_j := (t_j - s_j/N)$ Since $j = 0, \dots, N-1$, the values $2j/N - 1$ can be computed exactly (see \cref{ur1}), as we use at least $n$ qubits. On the other hand, there is an approximation error in computing $t_j - s_j/N$ via the oracles. Hence, we treat these cases separately below.  In \cref{cheb-error}, we continue to denote $y_j = 2N\left(t_j - \frac{s_j}{N}\right)$ for simplicity.

\begin{lem}\label{cheb-error}
Let $m,n \in \mathbb N$ such that $m \geq n+1$.
Let $q \le K$ and $j = 0, \dots, N-1$. Let $T_q(y_j) = \cos(q\arccos(y_j))$ denote the degree-$q$ Chebyshev polynomial of the first kind.
The errors in computing the Chebyshev polynomials can be bounded as follows:
\begin{enumerate}
    \item For $x_j = 2 j / N - 1$, the error in the QSP implementation due to $p$-bit finite-precision evaluation of $\arccos(x)$ satisfies
    \begin{equation}
    |T_q(x_j) - \widehat{T}_q(x_j)| \le q 2^{-p+1} \le K 2^{-p+1},
    \end{equation}
    uniformly in $j$.

    \item For $y_j = 2N(t_j - s_j/N)$, where $t_j$ and $s_j$ are provided by $m$-bit precision oracles, let $\widehat{y}_j$ denote the $m$-bit approximation. Let $y_j^* \in (\min(y_j, \widehat{y}_j), \max(y_j, \widehat{y}_j))$. The total Chebyshev polynomial error satisfies
    \begin{equation}
    |T_q(y_j) - \widehat{T}_q(y_j)| \le q \left( 2^{-p+1} + \frac{N 2^{-m+1}}{\sqrt{1 - (y_j^*)^2}} \right) \le K \left( 2^{-p+1} + \frac{N2^{-m+1}}{\sqrt{1 - (y_j^*)^2}} \right).
    \end{equation}
\end{enumerate}
\end{lem}

\begin{proof}
Let $\theta_j = \arccos(\cdot)$, where the argument is either $x_j$ or $y_j$, and define $\delta_{\theta_j} \coloneqq \lvert \theta_j - \widehat{\theta}_j \rvert$. 
\begin{enumerate}
\item Let $x_j = 2j/N - 1$.
The only source of error is the finite-precision evaluation of $\arccos(x_j)$ to $p$ bits, giving a rotation-angle error $\delta_{\theta_j} \le 2^{-p+1}$. By the mean-value theorem for the cosine function, the Chebyshev polynomial error is bounded by
\begin{equation}
|T_q(x_j) - \widehat{T}_q(x_j)| = |\cos(q \theta_j) - \cos(q (\theta_j + \delta_{\theta_j}))| \le q  \delta_{\theta_j} \le q  2^{-p+1} \le K  2^{-p+1}.
\end{equation}

\item 
Let $y_j = 2N(t_j - s_j/N)$, where $t_j$ and $s_j/N$ are provided by an $m$-bit precision oracle. We have $|y_j - \widehat{y}_j| \le 2N 2^{-m} = N 2^{-m+1}$. The rotation angle is $\theta_j = \arccos(y_j)$, while the computed rotation is $\widehat{\theta}_j$ with finite-precision $p$-bit evaluation. Applying the mean-value theorem to $\arccos$, there exists $y_j^* \in (\min(y_j, \widehat{y}_j), \max(y_j, \widehat{y}_j))$ such that
\begin{equation}
|\theta_j - \arccos(\widehat{y}_j)| \le \frac{|y_j - \widehat{y}_j|}{\sqrt{1 - (y_j^*)^2}} \le \frac{N 2^{-m+1}}{\sqrt{1 - (y_j^*)^2}}.
\end{equation}
Let $\widehat{\theta}_j = \arccos(\widehat{y}_j)$.
By the triangle inequality, we have
\begin{equation}
\delta_{\theta_j} = |\theta_j - \widehat{\theta}_j| 
\le |\theta_j - \arccos(\widehat{y}_j)| + |\arccos(\widehat{y}_j) - \widehat{\theta}_j|
\le \frac{N2^{-m+1}}{\sqrt{1 - (y_j^*)^2}} + 2^{-p+1},
\end{equation}
Applying the mean-value theorem to the cosine function,
\begin{equation}
|T_q(y_j) - \widehat{T}_q(y_j)| = |\cos(q\theta_j) - \cos(q \widehat{\theta}_j)| \le q \delta_{\theta_j}
= q \left( 2^{-p+1} + \frac{N2^{-m+1}}{\sqrt{1 - (y^*_j)^2}} \right).
\end{equation}
Since the maximum Chebyshev degree in the truncated expansion is $q \le K$, the error is bounded by
\begin{equation}
|T_q(y_j) - \widehat{T}_q(y_j)| \le K \left( 2^{-p+1} + \frac{N 2^{-m+1}}{\sqrt{1 - (y^*_j)^2}} \right).
\end{equation}
\end{enumerate}
This completes the proof.
\end{proof}


We can now compute $\| D_{\vec v_r} - \widetilde{D}_{\vec v_r} \|$ and $\| D_{\vec u_r} - \widetilde{D}_{\vec u_r} \|$.
Using \cref{cheb-error}, we have
\begin{equation}
\| D_{\vec v_r} - \widetilde{D}_{\vec v_r} \| \le 
\begin{cases}
0, & r = 0, \\[2mm]
r 2^{-p+1}, & r \ge 1,
\end{cases}.
\end{equation}
Similarly, using \cref{exp-error}, and \cref{cheb-error}, we have
\begin{equation}\label{non-uniform-bound}
\| D_{\vec u_r} - \widetilde{D}_{\vec u_r} \| \leq 
\|\vec \alpha'_r\|_1 \left( \pi N 2^{-m} + K \left( 2^{-p+1} + \max_j \frac{N 2^{-m+1}}{\sqrt{1 - (y^*_j)^2}} \right)
\right),
\end{equation}
Let $\alpha' = \sum_{r=0}^{K-1} \|\vec \alpha'_r\|_1$.
Substituting the bounds in \cref{second-term-to-bound}, we have
\begin{align}\label{block-encoding-error-final}
\| \widetilde{V}_{\rm II} - \widetilde{F}_{\rm II} \| 
&\le \sqrt{\|\vec c_s\|_\infty} \sum_{r=0}^{K-1} \|\vec \alpha'_r\|_1
\left(
\pi N 2^{-m} 
+ K 2^{-p+1} + \max_j \frac{N K 2^{-m+1}}{\sqrt{1 - (y^*_j)^2}}
+ K 2^{-p+1} 
\right), \\
&=
\alpha' \sqrt{\|\vec c_s\|_\infty} \left(
\pi N 2^{-m} + K 2^{-p+2} + \max_j \frac{NK 2^{-m+1}}{\sqrt{1 - (y^*_j)^2}}
\right).
\end{align}
We can now complete our error analysis by specifying the parameters $(K, m, p)$ so that the total error remains below the prescribed tolerance $\epsilon$.

\begin{thm}\label{nuqft-error}
Let $\epsilon > 0$ and let $n,m,p,K \in \mathbb{N}$ with $m \ge n+1$ and $N = 2^n$. Let $\widetilde{V}_{\mathrm{II}} = \sum_{r=0}^{K-1} \widetilde{D}_{\vec u_r}\,\widetilde{F}\,\widetilde{M}_{\vec s}\,\widetilde{D}_{\vec v_r}$ denote the quantum implementation of the Type-II NUQFT with Chebyshev truncation degree $K$, $m$-bit oracle precision, and $p$-bit finite-precision evaluation of $\arccos(x)$ in the QSP implementation. 
Let $\|\vec{\alpha}'_r\|_1 = \sum_{q=0}^{K-1} \lvert \alpha'_{qr} \rvert$, where $\alpha'_{qr}$ is defined in \cref{alpha-prime-coeffs}. Define
\begin{equation}\label{xj-def}
\alpha' = \sum_{r=0}^{K-1} \|\vec \alpha'_r\|_1, \quad 
y_j = 2N(t_j - s_j/N), \quad
y_j^* \in (\min(y_j, \widehat{y}_j), \max(y_j, \widehat{y}_j)),
\end{equation}
where $\widehat{y}_j$ is a $m$-bit approximation of $y_j$.
If $(K,m,p)$ are chosen such that
\begin{align}
K &= \mathcal{O}\left( \frac{\log(\sqrt{N \| \vec c_s \|_{\infty}}/\epsilon)}{\log \log (\sqrt{N \| \vec c_s \|_{\infty}}/\epsilon)} \right),
\label{big-o-K} \\
m & =\left \lceil \log_2 \frac{8  \alpha' \sqrt{\|\vec c_s\|_\infty}  (\pi N + K \max_j (1-(y_j^*))^{-1/2}}{\epsilon} \right \rceil, 
\\
p & = \left \lceil \log_2 \frac{16  \alpha' \sqrt{\|\vec c_s\|_\infty}  K}{\epsilon} \right \rceil,
\end{align}
then the total error of the quantum algorithm satisfies
\begin{equation}
\|\widetilde{V}_{\mathrm{II}} - F_{\mathrm{II}}\| \le \epsilon.
\end{equation}
In particular, it suffices to choose $m$ and $p$ such that
\begin{align}
m &= \mathcal{O} \left( \log \left( \frac{ K \sqrt{\|\vec c_s\|_\infty}  ( N + K \max_j (1-(y_j^*))^{-1/2}}{\epsilon} \right) \right) \label{big-O-m}, \\
p &= \mathcal{O} \left(\log \left( \frac{\sqrt{\|\vec c_s\|_\infty} K^2}{\epsilon} \right) \right) \label{big-O-p}.
\end{align}
\end{thm}

\begin{proof}
We have
\begin{equation}\label{eq:triangle}\|\widetilde{V}_{\mathrm{II}} - F_{\mathrm{II}}\| \le \underbrace{\|\widetilde{F}_{\mathrm{II}} - F_{\mathrm{II}}\|}_{\text{truncation error}} + \underbrace{\|\widetilde{V}_{\mathrm{II}} - \widetilde{F}_{\mathrm{II}}\|}_{\text{block-encoding error}},
\end{equation}
where the first term is the Chebyshev truncation error and the second term is the block-encoding error. Using \cref{max-2norm-lem}, the first term in \cref{triangle-ineq} can be expressed as
\begin{align}
\bigl\lVert \widetilde{F}_{\mathrm{II}} - F_{\mathrm{II}} \bigr\rVert
&= \bigl\lVert (A - A_K) \circ (F M_{\vec{s}}) \bigr\rVert \\
&\le \sqrt{N} \lVert A - A_K \rVert_{\max}  \lVert F M_{\vec{s}} \rVert \\
&\le \lVert \vec{c}_s \rVert_{\infty}  \sqrt{N} \lVert A - A_K \rVert_{\max}.
\end{align}
The last equality follows since $F$ is unitary, and we have 
\begin{equation}
\| F M_{\vec s} \| = \| M_{\vec s}\| = \sqrt{\| \vec c_s \|_{\infty}}.
\end{equation}
Hence, in order to ensure that $\|\widetilde{F}_{\operatorname{II}} - F_{\operatorname{II}}\| \leq \epsilon/2$ for some $\epsilon > 0$ , it suffices to choose
$\epsilon'_{\operatorname{trun}} = \epsilon/2 \sqrt{N \| \vec c_s \|_{\infty}}$. Hence, it suffices to choose $K$ such that
\begin{equation}\label{K-bigO}
K = \mathcal{O} \left( \frac{\log(\sqrt{N \| \vec c_s \|_{\infty}}/\epsilon}{\log \log (\sqrt{N \| \vec c_s \|_{\infty}}/\epsilon)} \right),
\end{equation}
From \eqref{block-encoding-error-final}, the block-encoding error satisfies
\begin{equation}\label{second-error}
\|\widetilde{V}_{\rm II} - \widetilde{F}_{\rm II}\| 
\le \alpha' \sqrt{\|\vec c_s\|_\infty} \left(
\pi N 2^{-m} + K 2^{-p+2} + \max_j \frac{NK 2^{-m+1} }{\sqrt{1-(y_j^*)^2}}
\right).
\end{equation}
Let $C := \alpha' \sqrt{\|\vec c_s\|_\infty}$ and $\kappa := \max_j \frac{1}{\sqrt{1-(y_j^*)^2}}$.
Using this notation, \cref{second-error} can be succinctly expressed as
\begin{equation}
\|\widetilde{V}_{\rm II} - \widetilde{F}_{\rm II}\| 
\le C  ( 2^{-m} (\pi N + 2N K \kappa) + K 2^{-p+2}  ).
\end{equation}
To ensure that the right-hand side is at most $\epsilon/2$, we can allocate the error budget equally between the oracle and QSP contributions, i.e.,
\begin{align}
C 2^{-m} (\pi N + 2N K \kappa) &\le \frac{\epsilon}{4}, \\
C K 2^{-p+2} &\le \frac{\epsilon}{4}.
\end{align}
Solving for $m$, we obtain
\begin{equation}
C2^{-m} (\pi N + 2N K \kappa) \le \frac{\epsilon}{4} 
\iff 
m \ge \log_2  \left( \frac{4 C (\pi N + 2N K \kappa)}{\epsilon}\right).
\end{equation}
Solving for $p$, we obtain
\begin{equation}
C K 2^{-p+2} \le \frac{\epsilon}{4}
\iff 
2^{-p} \le \frac{\epsilon}{16 C K}
\iff 
p \ge \log_2 \left( \frac{16 C K}{\epsilon} \right).
\end{equation}
Substituting the definitions of $C$ and $\kappa$ yields
\begin{align}
m &= \left \lceil \log_2 \left(\frac{4 \alpha' \sqrt{\|\vec c_s\|_\infty} (\pi N + 2N K \max_j (1-(y_j^*))^{-1/2})}{\epsilon} \right) \right \rceil \label{first-m}, \\
p &= \left \lceil \log_2 \left( \frac{16 \alpha' \sqrt{\|\vec c_s\|_\infty} K}{\epsilon} \right) \right \rceil \label{first-p}.
\end{align}
With these choices of $m$ and $p$, the error satisfies $\|\widetilde{V}_{\mathrm{II}} - \widetilde{F}_{\mathrm{II}}\| \le \epsilon$.  
Combined with the choice of $K$, we therefore have
\begin{equation}
\|\widetilde{V}_{\mathrm{II}} - F_{\mathrm{II}}\| \le \epsilon.
\end{equation}
We now further simplify the choices of $m$ and $p$. It is well known \cite{watson1922treatise,abramowitz1964handbook} that for $\nu \ge 0$ and $0 \le x \ll \nu$, we have
\begin{equation}\label{bessel-asymp}
J_\nu(x) \sim \frac{(x/2)^{\nu}}{\Gamma(\nu+1)},
\end{equation}
where $\Gamma$ is the Gamma function. Furthermore, recall that $\alpha'_{qr} = \mathcal{O}(\alpha_{qr})$ (see \cref{alpha-prime-coeffs}), where $\alpha_{qr}$ is defined in \cref{a-coeffs} with $\gamma = \tfrac{1}{2}$ by
\begin{equation}\label{cases}
\alpha_{qr}=
\begin{cases}
4 i^r  J_{\frac{q + r}{2}}\left(-\pi/4 \right)  J_{\frac{r - q}{2}}\left(-\pi/4\right), & \text{if } \operatorname{mod}(|q - r|, 2) = 0, \\
0, & \text{otherwise}
\end{cases}
\end{equation}
The condition $\operatorname{mod}(|q-r|,2)=0$ implies that $q$ and $r$ have the same parity. Hence, the order of the Bessel function $J$ is an integer in \cref{cases}. Since we are interested in $|\alpha_{qr}|$, it suffices to evaluate the Bessel functions at $\pi/4$ and to assume $(r-q)/2 \ge 0$, as Bessel functions of integer order are either even or odd.
\cref{bessel-asymp} now implies that $|\alpha_{qr}| = \mathcal{O}(1)$, since for the fixed value $x_0 = \pi/4$, the magnitude of $J_\nu(x_0)$ decreases rapidly as $\nu$ increases. Hence, $\alpha' = \mathcal{O}(K)$, since each of the $K$ summands $|\alpha'_{qr}|$ is $\mathcal{O}(1)$.
Based on \cref{first-m} and \cref{first-p}, it suffices to choose $m$ and $p$ such that
\begin{align}
m = \mathcal{O} \left( \log \left( \frac{ N K \sqrt{\|\vec c_s\|_\infty}  ( 1 + K \max_j (1-(y_j^*))^{-1/2}}{\epsilon} \right) \right), \quad
p = \mathcal{O} \left(\log \left( \frac{\sqrt{\|\vec c_s\|_\infty} K^2}{\epsilon} \right) \right).
\end{align}
This completes the proof.
\end{proof}

The expressions in \cref{nuqft-error} are rather involved.
Noting that $N = 2^n$ and $\|\vec c_s\|_\infty = \mathcal{O}(2^n)$, we can further simplify the $\mathcal{O}(\cdot)$ estimates in \cref{nuqft-error} in terms of $n$, $\epsilon$, and $\kappa := \max_j \frac{1}{\sqrt{1-(y_j^*)^2}}$. Noting that 
$\sqrt{N \| \vec c_s \|_{\infty}} = \mathcal{O}(2^n)$, we have
\begin{align}
\log\left( \frac{\sqrt{N \| \vec c_s \|_{\infty}}}{\epsilon} \right)
= \Theta(n) + \log\left( \frac{1}{\epsilon} \right) = \Theta \left( n + \log \left( \frac{1}{\epsilon} \right) \right).
\end{align}
Hence, we can simplify the complexity estimate for $K$ in \cref{big-o-K} to
\begin{equation}
K = \mathcal{O}\left( 
\frac{n + \log (1/\epsilon)}{\log(n + \log (1/\epsilon))}\right) 
=
 \mathcal{O}\left( n + \log\left(\frac{1}{\epsilon}\right) \right).
\end{equation}
Similarly, we obtain
\begin{align}
p &= \mathcal{O}\left( n + \log(1/\epsilon) + \log K \right), \\
&= \mathcal{O}\left( n + \log\left(\frac{1}{\epsilon}\right) + \log\left(n + \log\left(\frac{1}{\epsilon}\right) \right) \right)
= \mathcal{O}\left( n + \log\left(\frac{1}{\epsilon}\right) \right).
\end{align}
The second last equality follows since $\log K = \mathcal{O}(n + \log(1/\epsilon))$. In a similar manner, we obtain
\begin{align}
m &= \mathcal{O}\left(\log\left( \frac{NK \sqrt{\|\vec c_s\|_\infty} (1 + K \kappa)}{\epsilon} \right)\right),\\
&= \mathcal{O}\left(\log\left( \frac{ K 2^{3n/2} (1 + K \kappa)}{\epsilon} \right)\right), \\
&= \mathcal{O}\left(\log K + n + \log (1 + K \kappa ) + \log\left(\frac{1}{\epsilon}\right) \right) 
= \mathcal{O}\left(n + \log\left(\frac{1}{\epsilon}\right) + \log (1 + K \kappa ) \right).
\end{align}

\begin{cor}\label{final-estimate-cor}
Let $\epsilon > 0$ and let $n, m, p, K \in \mathbb{N}$ with $m \ge n+1$. Let $\widetilde{V}_{\mathrm{II}}$, $y_j$ and $y_j^*$ be defined as in \cref{nuqft-error}. If the parameters $(K,m,p)$ are chosen such that
\begin{align}
K, p =  \mathcal{O}\left( n + \log\left(\frac{1}{\epsilon}\right) \right), \quad
m = \mathcal{O}\left(n + \log\left(\frac{1}{\epsilon}\right) + \log (1 + K \kappa ) \right),
\end{align}
where $\kappa := \max_j \frac{1}{\sqrt{1-(y_j^*)^2}}$, then 
$\|\widetilde{V}_{\mathrm{II}} - F_{\mathrm{II}}\| \le \epsilon$.
\end{cor}

\begin{proof}
This follows from \cref{nuqft-error} and the discussion above.
\end{proof}

\section{Numerical Results}\label{numerics-res}
We present numerical results obtained from both classical simulations and hybrid classical–quantum implementations for small-scale instances.
The objective is to validate the new components introduced in the quantum algorithm rather than to re-establish the accuracy of the underlying classical method. The experiments therefore isolate the elements specific to our construction: the low-rank truncation used in the LCU representation, finite-precision discretization of node data, the reversible implementation of the $\arccos$ subroutine, the state-preparation routines $U_{v_{\vec r}}$ and $U_{u_{\vec r}}$, the small-$K$ LCU combination, and the block encoding of the matrix $M_{\vec s}$.
Accordingly, we run numerical tests for problem sizes accessible to classical simulation. The code used to generate the reported results is available in \cite{Junaid2026NUQFT}.

\subsection{Validation of Low-Rank Truncation Error Scaling}
The low-rank Chebyshev approximation underlying our construction was extensively analyzed and numerically validated in Antol\'{i}n and Townsend~\cite{ruiz2018nonuniform}. As our work focuses on translating this decomposition into a quantum block-encoding framework, we do not repeat those experiments and instead only validate the truncation behavior of the low-rank approximation for representative node configurations. 
The low-rank truncation error is characterized by
\begin{equation}
\bigl\lVert \widetilde{F}^{(K)}_{\mathrm{II}} - F_{\mathrm{II}} \bigr\rVert
\le \lVert \vec{c}_s \rVert_{\infty} \sqrt{N}
\lVert A - A_K \rVert_{\max}.
\end{equation}
Here, the superscript indicates that $\widetilde{F}^{(K)}_{\mathrm{II}}$ is defined with respect to the truncation parameter $K$.
While $\lVert A - A_K \rVert_{\max} = \mathcal{O}(\log(1/\epsilon))$,
a more precise estimate is
$5\gamma e^{W\!\left(\frac{\log(140/\epsilon)}{5\gamma}\right)}$,
as given in \cite{ruiz2017quantum}, where $W$ is Lambert's $W$ function.
Hence, we obtain the estimate
\begin{equation}
K =  
5\gamma e^{ \left(
\frac{\log \left(
\lVert \vec{c}_s \rVert_{\infty} \sqrt{N} \, 140/\epsilon
\right)}{5\gamma}
\right)}.
\end{equation}
For a range of transform sizes $N$, we numerically construct $\widetilde{F}_{\mathrm{II}}^{(K)}$ for increasing $K$ and measure the spectral-norm error $\|F_{\mathrm{II}} - \widetilde{F}_{\mathrm{II}}^{(K)}\|$. 
The experiments cover perturbed equispaced grids, clustered nodes, and randomly distributed nodes in $[0,1)$. \cref{fig:rank-error} confirms the predicted error baselines and shows that moderate $K$ suffices to achieve $\epsilon = 10^{-10}$.
This test also demonstrates that, for quantum implementations, the LCU qubit cost is modest. For instance, when $\epsilon = 10^{-5}$, approximately $4$ qubits are required.
\begin{figure}[htbp]
    \centering
    
    \begin{subfigure}{0.30\textwidth}
        \centering
        \includegraphics[width=\textwidth]{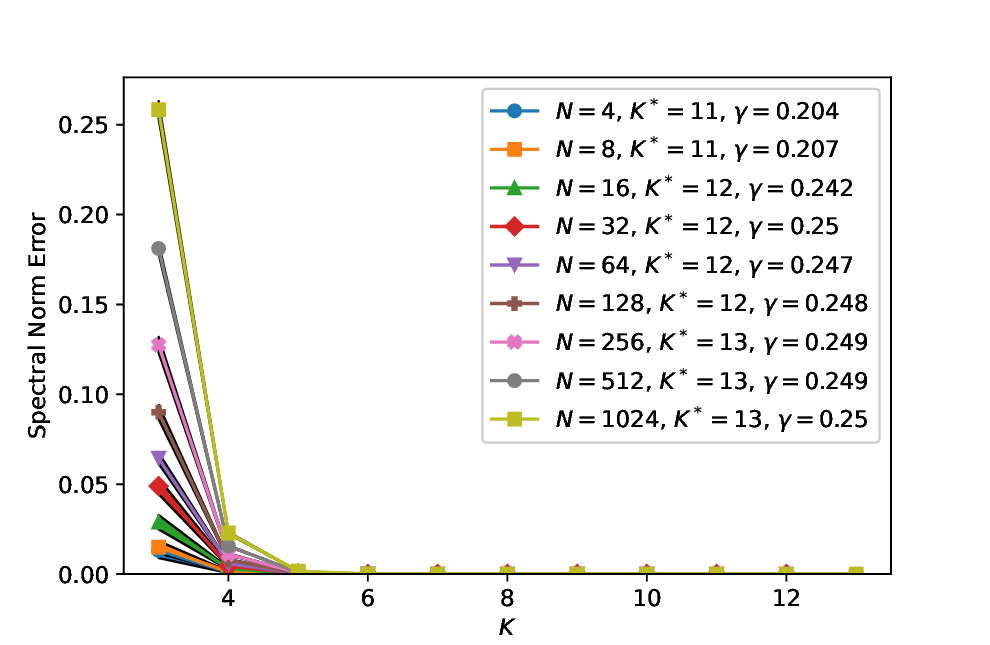}
        \caption{Perturbed equispaced nodes}
        \label{fig:rank-error-a}
    \end{subfigure}
    \hfill
    \begin{subfigure}{0.30\textwidth}
        \centering
        \includegraphics[width=\textwidth]{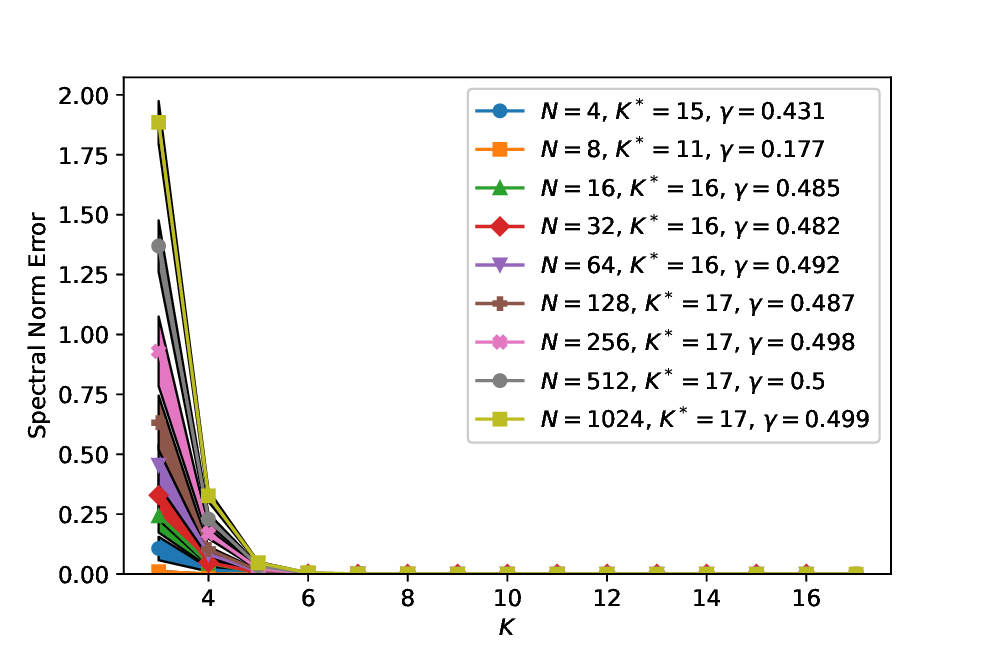}
        \caption{Clustered nodes}
        \label{fig:rank-error-b}
    \end{subfigure}
    \hfill
    \begin{subfigure}{0.30\textwidth}
        \centering
        \includegraphics[width=\textwidth]{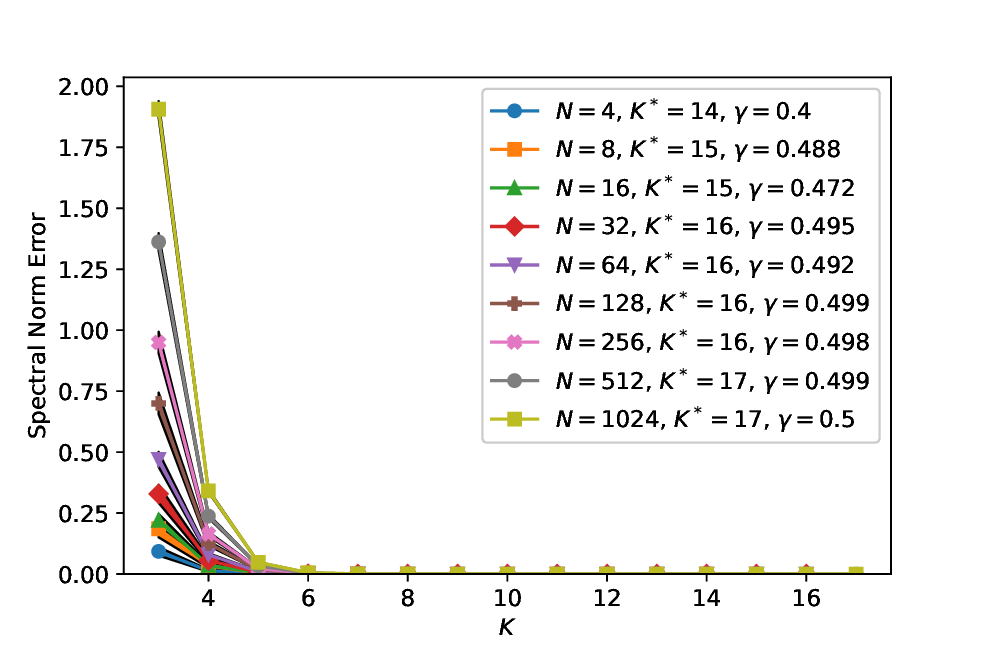}
        \caption{Uniformly random nodes}
        \label{fig:rank-error-c}
    \end{subfigure}
    \caption{
    $\| (F_{\mathrm{II}} - \widetilde{F}_{\mathrm{II}}^{(K)}) \|_2$
    as a function of $K$ for 
    $N = 2^n$, $n=2,\dots,10$, with accuracy $\epsilon = 10^{-10}$. 
    For each $N$ and node geometry, the error is averaged over $20$ independent 
    complex Gaussian test vectors. 
    $K^\ast$ denotes the minimal $K$ required to achieve the prescribed tolerance.
    }
    \label{fig:rank-error}
\end{figure}
 
The low-rank Chebyshev approximation underlying our construction was extensively analyzed and numerically validated in Antol\'{i}n and Townsend~\cite{ruiz2018nonuniform}. As our work focuses on translating this decomposition into a quantum block-encoding framework, we do not repeat those experiments and instead validate the newly introduced quantum components.

\subsection{Finite-Precision Effects}
We next investigate the effect of finite-precision node representation on the
accuracy of the transform. The parameter $m$ governs the precision with which
the non-uniform nodes $t_j$ and derived quantities such as
$s_j$ are represented. In the theoretical analysis presented in \cref{u-round}, the parameter $m$
scales logarithmically with the truncation parameter $K$ and the condition
parameter $\kappa$. The purpose of this experiment is to demonstrate that the
dependence of $m$ on both $\kappa$ and $K$ is sufficiently mild for practical
purposes. This is verified through classical emulation.
For a fixed transform size \(N\), we generate node sets using a function parameterized by \(\gamma\) such that
\begin{equation}
\left|t_j - \frac{s_j}{N}\right| = \gamma N, \qquad \gamma \in \left(0,\tfrac12\right),
\end{equation}
thereby producing a range of geometric configurations with varying \(\kappa\). Given a node set $\{x_j\} \subseteq [0,1)$, we can form the truncated representations
\begin{equation}
t_j^{(m)} = \frac{\mathrm{round}(2^m t_j)}{2^m},
\end{equation}
and construct the corresponding approximate transform
$\widetilde{F}_{\mathrm{II}}^{(m)}$ using these discretized nodes.
For each node set, we compute the parameter $\kappa$ and the minimal precision $m^*$ required to achieve spectral-norm error $\epsilon$.
The minimal precision $m^*$ is determined via a bracketing and binary
search procedure over $m$.   
As shown in \cref{tab:gamma_comparison}, $\kappa$ increases monotonically with $\gamma$, with particularly sharp growth once $\gamma \geq 0.49$. Despite this amplification, the required precision parameter $m^*$ remains essentially stable across the entire range.
This stability arises from compensating effects. Although $\kappa$ grows rapidly near the boundary, its overall scaling remains moderate (effectively sublinear), and the admissible $K$ also stays moderate. Most importantly, the precision depends on $\log(1+K\kappa)$, whose logarithmic structure mitigates the impact of large $\kappa$. Consequently, even when $\kappa$ becomes large, $\log(1+K\kappa)$ grows slowly and $m^*$ changes very little. Thus, the $\log(1+K\kappa)$ factor in \cref{u-round} is mild, and the parameters $m,n,p$ all scale as $\mathcal{O}\!\left(n + \log(1/\epsilon)\right)$ morally.
\begin{table}[h]
    \renewcommand{\arraystretch}{1.25}
    \centering
    \scalebox{0.80}{
    \begin{tabular}{c@{\hspace{1.5cm}}c}
    \begin{tabular}{c|ccccc}\hline\hline
        \# & $\gamma$ & $\kappa$ & $K$ & $\log(1+K\kappa)$ & $m^*$ \\\hline
        1  & 0        & 1        & 1  & 0.693147 & 7  \\
        2  & 0.0544444 & 1.00598 & 10 & 2.40332  & 15 \\
        3  & 0.108889  & 1.02459 & 11 & 2.5072   & 15 \\
        4  & 0.163333  & 1.05804 & 13 & 2.69155  & 15 \\
        5  & 0.217778  & 1.11091 & 14 & 2.80655  & 15 \\
        6  & 0.272222  & 1.19218 & 14 & 2.87303  & 15 \\
        7  & 0.326667  & 1.32088 & 15 & 3.03559  & 15 \\
        8  & 0.381111  & 1.54484 & 16 & 3.24717  & 15 \\
        9  & 0.435556  & 2.0363  & 17 & 3.57283  & 15 \\
        10 & 0.49      & 5.02519 & 17 & 4.45931  & 15 \\\hline\hline
    \end{tabular}
    &
    \begin{tabular}{c|ccccc}\hline\hline
        \# & $\gamma$ & $\kappa$ & $K$ & $\log(1+K\kappa)$ & $m^*$ \\\hline
        1  & 0.49   & 5.02519 & 17 & 4.45931 & 15 \\
        2  & 0.4911 & 5.32373 & 17 & 4.51638 & 13 \\
        3  & 0.4922 & 5.68359 & 17 & 4.58109 & 15 \\
        4  & 0.4933 & 6.12904 & 17 & 4.6558  & 13 \\
        5  & 0.4944 & 6.70032 & 17 & 4.74411 & 7  \\
        6  & 0.4955 & 7.47039 & 17 & 4.852   & 15 \\
        7  & 0.4966 & 8.58954 & 17 & 4.99058 & 12 \\
        8  & 0.4977 & 10.4377 & 17 & 5.18426 & 15 \\
        9  & 0.4988 & 14.4424 & 17 & 5.50745 & 15 \\
        10 & 0.4999 & 50.0025 & 17 & 6.74646 & 15 \\\hline\hline
    \end{tabular}
    \end{tabular}}
    \caption{Comparison of $\kappa$ and $\log(1+K\kappa)$ for two $\gamma$ regimes for $N = 2^6$, $\epsilon = 10^{-12}$.}
    \label{tab:gamma_comparison}
\end{table}

\subsection{Quantum Arccos Subroutine}\label{arccos-subroutine}
We examine the quantum implementation of the $\arccos$ subroutine using a reversible CORDIC-style routine with precision parameter $p$. The $\arcsin(x)$ algorithm of~\cite{burge2024quantumcordicarcsin} uses $5p-1$ qubits, while our $\arccos(x)$ implementation requires $7p+10$, preserving $\mathcal{O}(p)$ scaling but with a larger constant factor. Numerical tests with a \texttt{matrix-product-state simulator} for $p=1,\dots,6$ (with $m=6$) show that $p=7$ exceeds simulator capacity. Results for $p=1,3,5$ are shown in \cref{fig:odd-p-figs}. 
\begin{figure}[h]
    \centering
    
    \begin{subfigure}{0.30\textwidth}
        \centering
        \includegraphics[width=\textwidth]{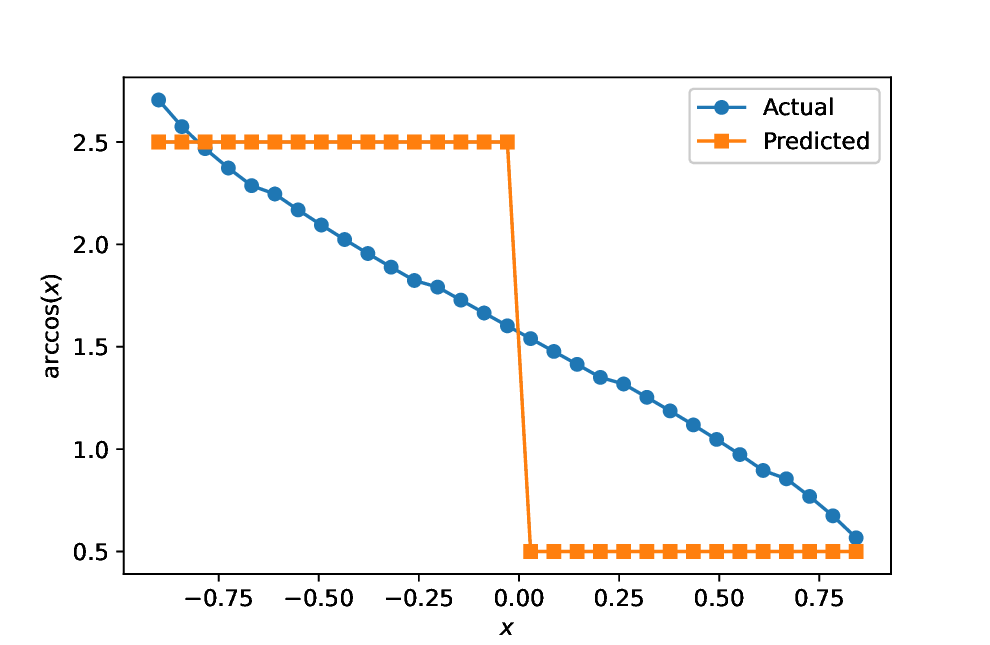}
        \caption{$p=1$}
    \end{subfigure}
    \hfill
    \begin{subfigure}{0.30\textwidth}
        \centering
        \includegraphics[width=\textwidth]{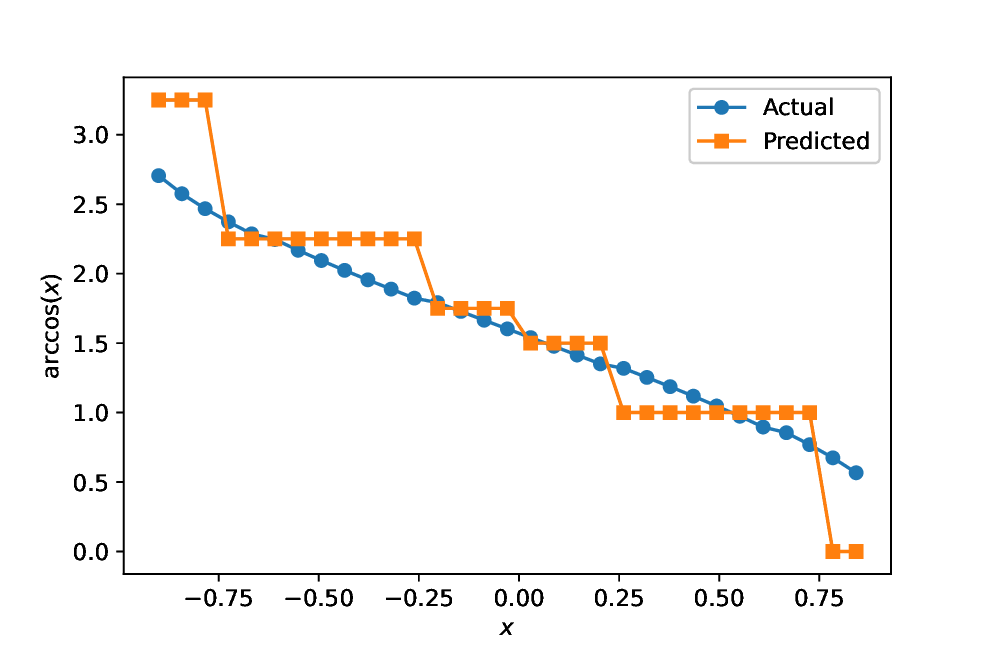}
        \caption{$p=3$}
    \end{subfigure}
    \hfill
    \begin{subfigure}{0.30\textwidth}
        \centering
        \includegraphics[width=\textwidth]{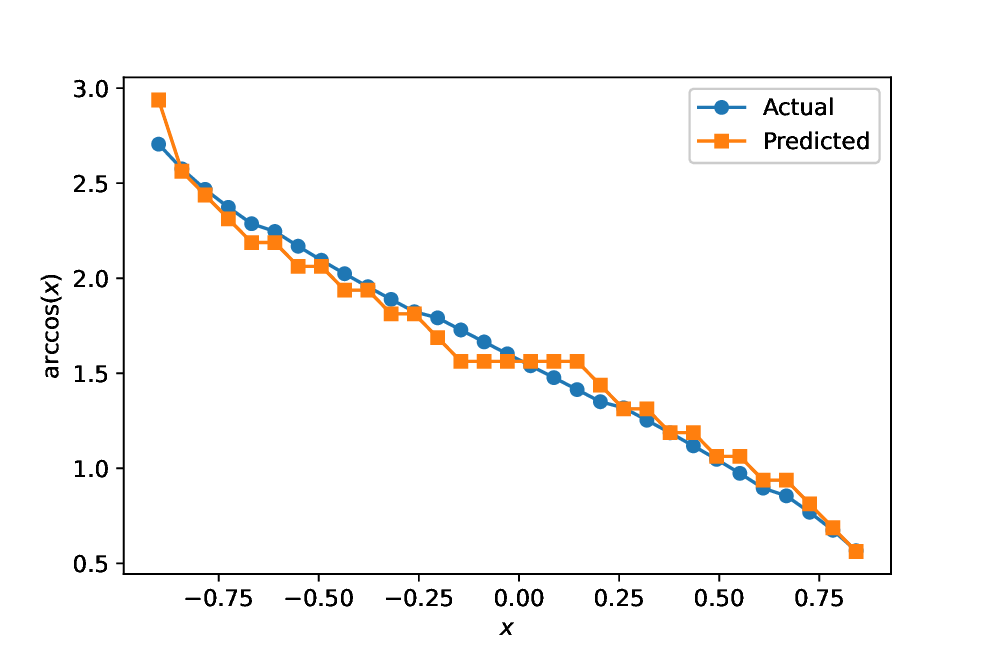}
        \caption{$p=5$}
    \end{subfigure}
    \caption{
    Comparison of actual and predicted values of $\arccos(x)$
    for odd $p$ values ($p = 1, 3, 5$).
    The corresponding mean absolute errors are $0.5195115053$, $0.2251069658$, and $0.0645826060$, respectively.
    }
    \label{fig:odd-p-figs}
\end{figure}
These results indicate a hardware bottleneck that restricts the algorithm’s simulation to small parameters. Consequently, a fully quantum implementation is not feasible with current hardware. 
Therefore, we adopt a hybrid classical–quantum approach, performing arithmetic operations classically while enabling quantum simulation for small parameters.

\begin{figure}[t]
    \centering

    \begin{subfigure}{0.30\textwidth}
        \centering
        \includegraphics[width=\textwidth]{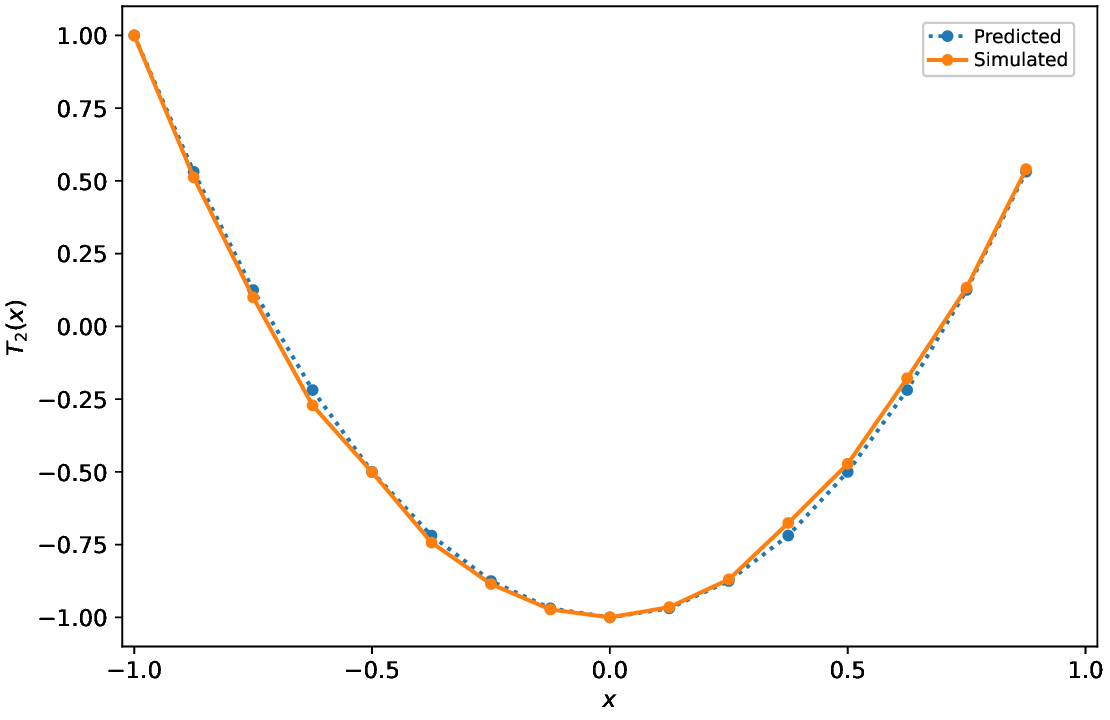}
        \caption{$n=4, p=5, r=2$}
    \end{subfigure}
    \hfill
    \begin{subfigure}{0.30\textwidth}
        \centering
        \includegraphics[width=\textwidth]{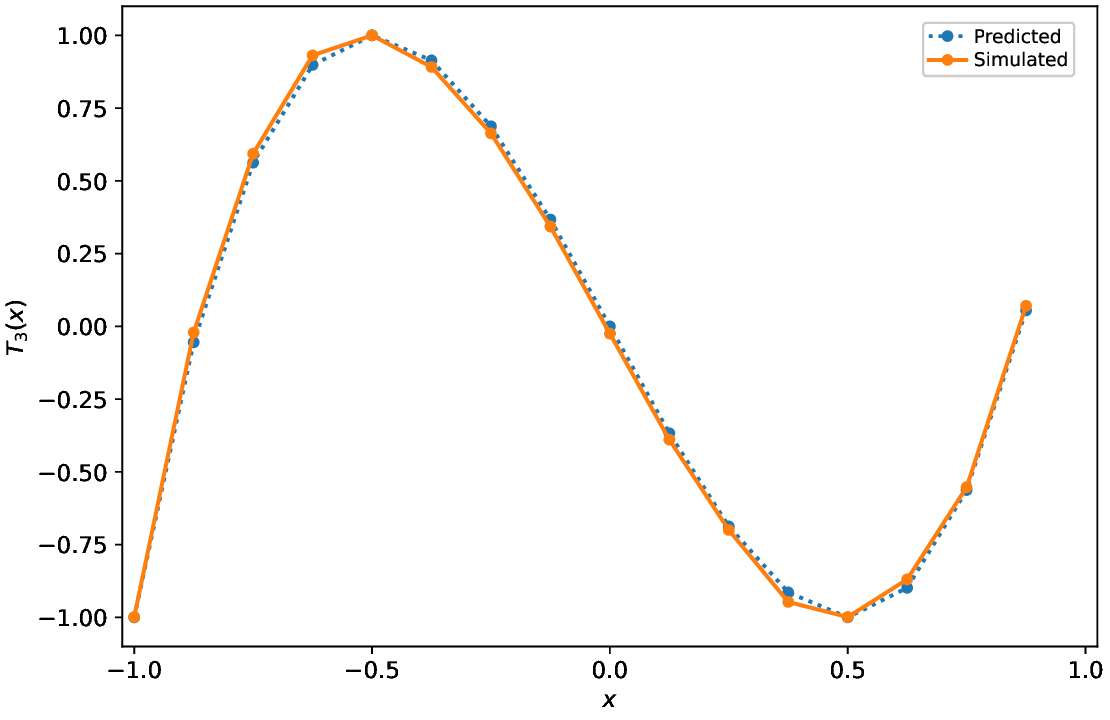}
        \caption{$n=4, p=5, r=3$}
    \end{subfigure}
    \hfill
    \begin{subfigure}{0.30\textwidth}
        \centering
        \includegraphics[width=\textwidth]{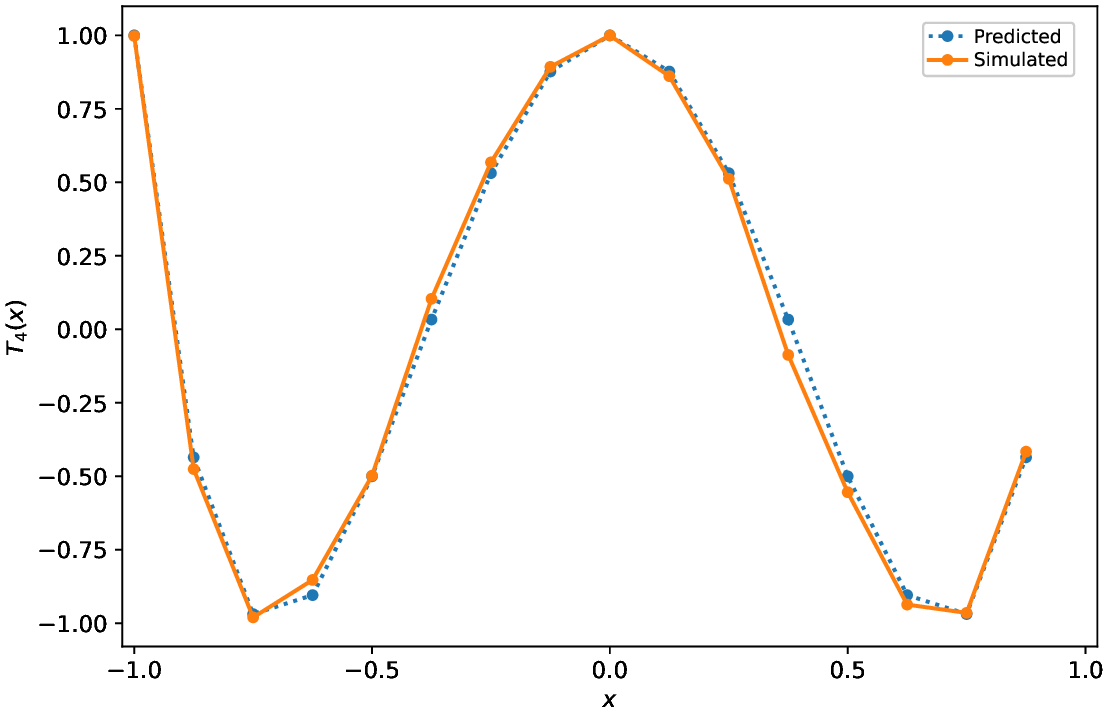}
        \caption{$n=4, p=5, r=4$}
    \end{subfigure}

    \vspace{0.4cm}

    \begin{subfigure}{0.30\textwidth}
        \centering
        \includegraphics[width=\textwidth]{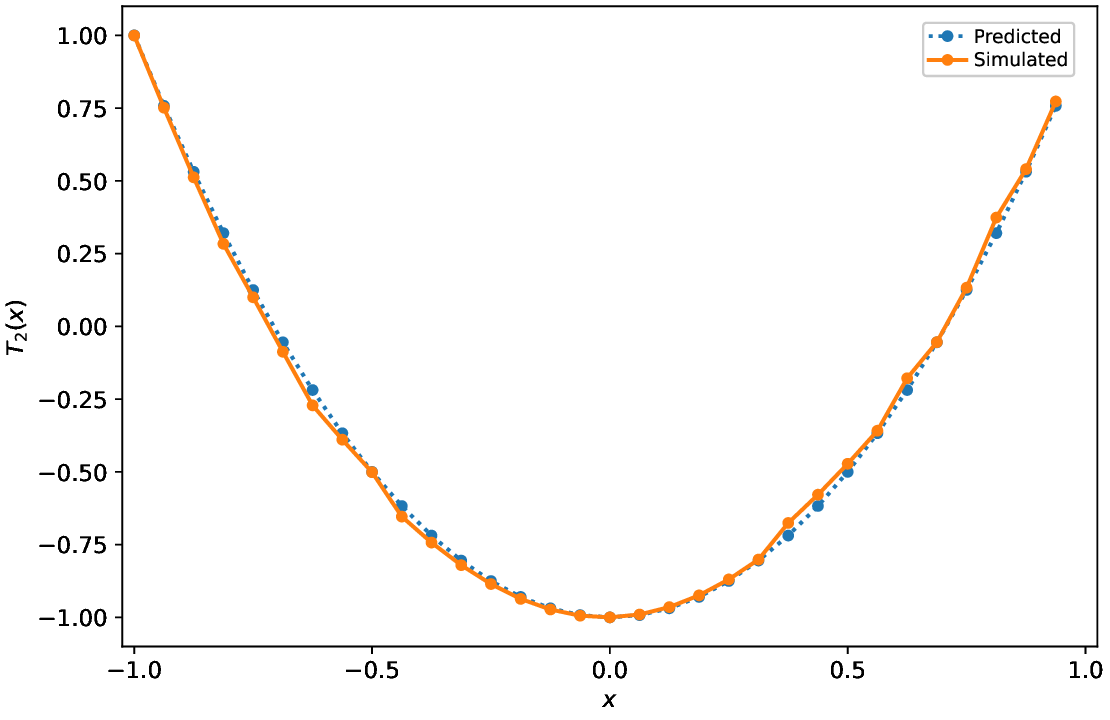}
        \caption{$n=5, p=5, r=2$}
    \end{subfigure}
    \hfill
    \begin{subfigure}{0.30\textwidth}
        \centering
        \includegraphics[width=\textwidth]{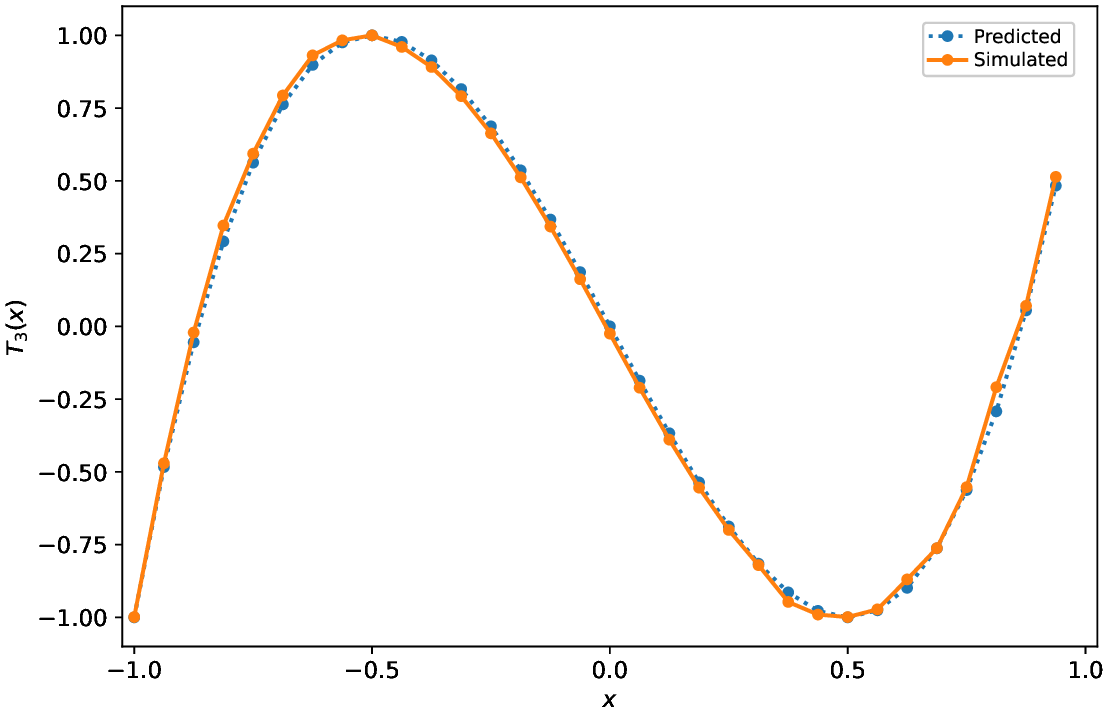}
        \caption{$n=5, p=5, r=3$}
    \end{subfigure}
    \hfill
    \begin{subfigure}{0.30\textwidth}
        \centering
        \includegraphics[width=\textwidth]{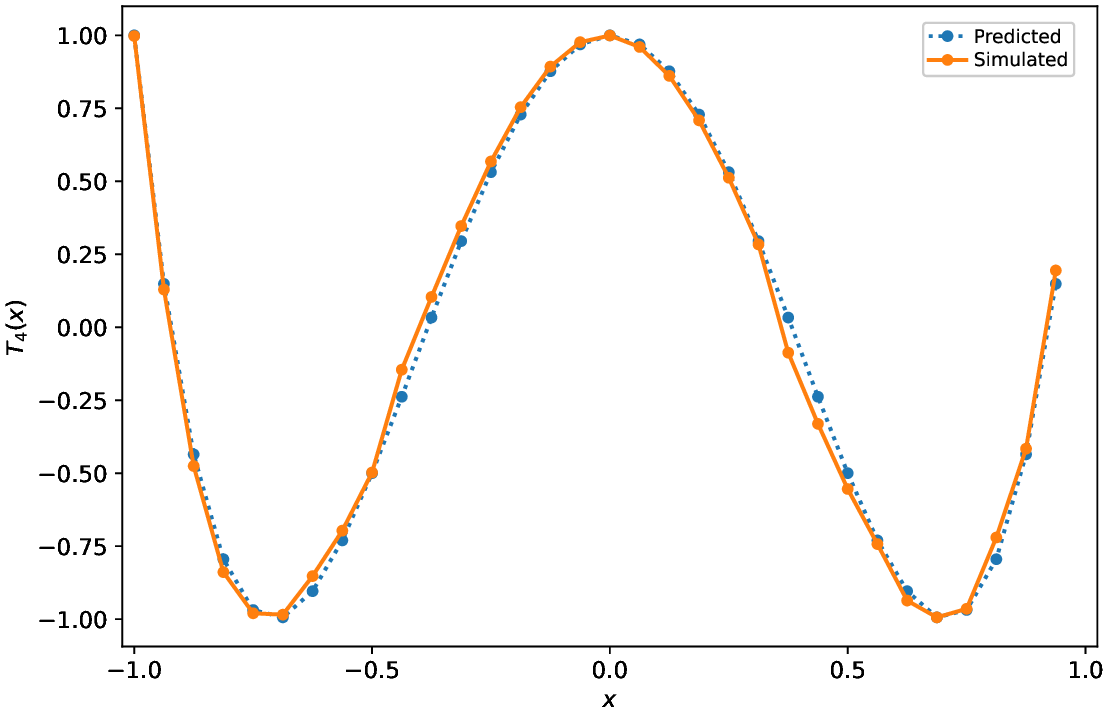}
        \caption{$n=5, p=5, r=4$}
    \end{subfigure}

    \vspace{0.4cm}

    \begin{subfigure}{0.30\textwidth}
        \centering
        \includegraphics[width=\textwidth]{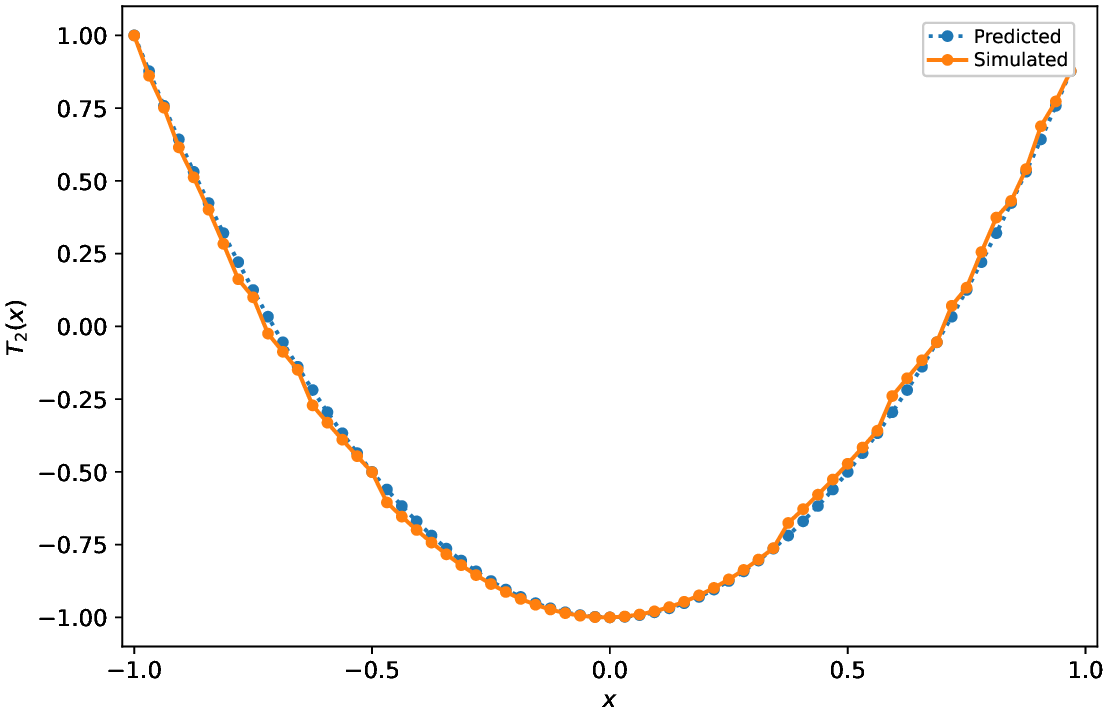}
        \caption{$n=6, p=5, r=2$}
    \end{subfigure}
    \hfill
    \begin{subfigure}{0.30\textwidth}
        \centering
        \includegraphics[width=\textwidth]{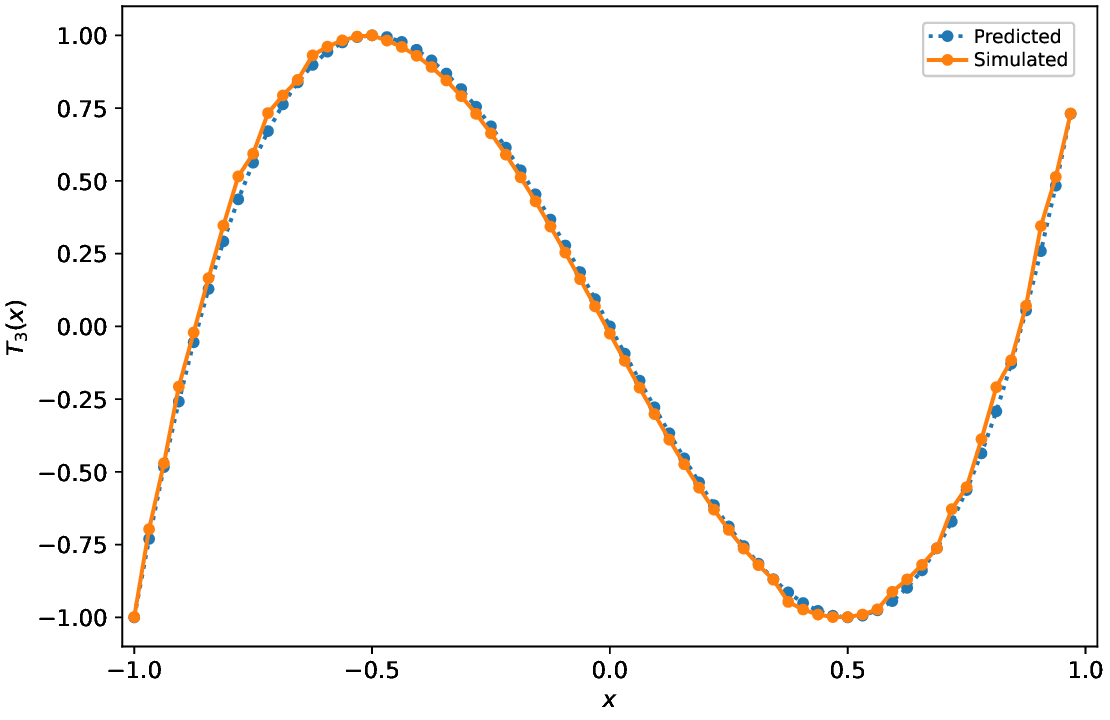}
        \caption{$n=6, p=5, r=3$}
    \end{subfigure}
    \hfill
    \begin{subfigure}{0.30\textwidth}
        \centering
        \includegraphics[width=\textwidth]{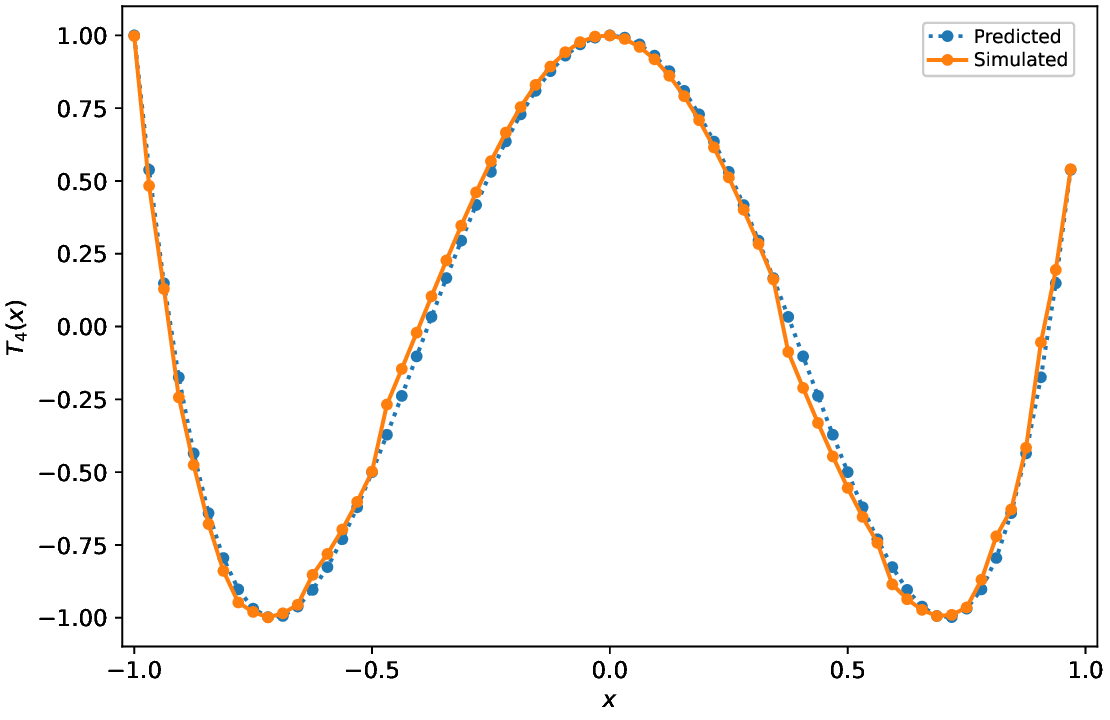}
        \caption{$n=6, p=5, r=4$}
    \end{subfigure}
    \caption{Comparison of computed and predicted outputs for different parameter of the quantum circuit $U_{v_{\vec r}}$ for $n \in \{4,5,6\}$ and $r \in \{2,3,4\}$ with fixed $p=5$.}
    \label{fig:vr-plots}
\end{figure}
\subsection{Hybrid Implementation of $U_{v_{\vec r}}$}
We validate the algorithm in \cref{4.1} implementing the unitary operator $U_{v_{\vec r}}$. The $\arccos$ computation is performed classically due to current hardware and simulation constraints and is included in the pre-processing stage. The resulting parameters are then loaded into a $p$-qubit register.
To verify correctness, the circuit is simulated for small problem sizes where full statevector simulation is feasible. The simulated amplitudes are renormalized by multiplying by $\sqrt{N}$ and compared with the classical values $T_r(x_j)$ defining $v_{\vec r}$, computed via direct Chebyshev evaluation using \texttt{NumPy}.
For various degrees $r$ and qubit counts $n$, the simulated amplitudes are compared with the corresponding Chebyshev polynomial values. The results in \cref{fig:vr-plots} show that $U_{v_{\vec r}}$ accurately reproduces the vector $v_{\vec r}$.
\cref{tab:vr-p-error-ratio} reports the statevector error for fixed $n=4$ and $r=4$ as $p$ increases. The ratios of consecutive errors confirm that the error decreases by approximately a factor of two when $p$ increases by one.
\begin{table}[h]
    \renewcommand{\arraystretch}{1.25}
    \centering
    \scalebox{0.80}{
    \begin{tabular}{c|cccc|c}\hline\hline
        \# & $n$ & $p$ & $r$ & Error & Ratio \\\hline
        1 & 4 & 1 & 4 & 0.9880609500 & 1.6475 \\
        2 & 4 & 2 & 4 & 0.5997170775 & 2.2256 \\
        3 & 4 & 3 & 4 & 0.2694581809 & 2.1944 \\
        4 & 4 & 4 & 4 & 0.1227710437 & 2.0449 \\
        5 & 4 & 5 & 4 & 0.0600312609 & 1.8409 \\
        6 & 4 & 6 & 4 & 0.0326167771 & 2.2405 \\
        7 & 4 & 7 & 4 & 0.0145581686 & 1.9220 \\
        8 & 4 & 8 & 4 & 0.0075722616 & -- \\\hline\hline
    \end{tabular}}
    \caption{Statevector error for fixed $n=4$ and $r=4$ as $p$ increases.}
    \label{tab:vr-p-error-ratio}
\end{table}

\subsection{Hybrid Implementation of $U_{u_{\vec r}}$}
We next validate the algorithm in \cref{4.2} implementing the unitary operator $U_{u_{\vec r}}$.  
As in the previous analysis, all arithmetic computations are carried out classically, and the resulting data are subsequently loaded into the quantum registers.
Let $\vec{x} = \vec{t} - \vec{s}/N$ and 
\begin{equation}
u_q(\vec{x}) = \sqrt{N} e^{-i \pi N \vec{x}}\, T_q(2N \vec{x})
\end{equation}
denote the normalized amplitude of the $\ket{0}$ computational basis state of the circuit $U_q$.
For a range of values of $q$ and transform sizes $N$, 
we plot the real and imaginary parts of the simulated amplitudes and the predicted components in~\cref{fig:ur-plots}.
The results in~\cref{fig:ur-plots} provide numerical evidence that
$U_{q}$ accurately reproduces 
$u_{qr}(\vec{x})/\sqrt{N}$. 
\cref{tab:ur-p-m-error-ratio} reports the statevector error for fixed $n=4$ and $q=2$ as $p,m$ increases. 
\begin{figure}[h]
    \centering
    \begin{subfigure}{0.30\textwidth}
        \centering
        \includegraphics[width=\textwidth]{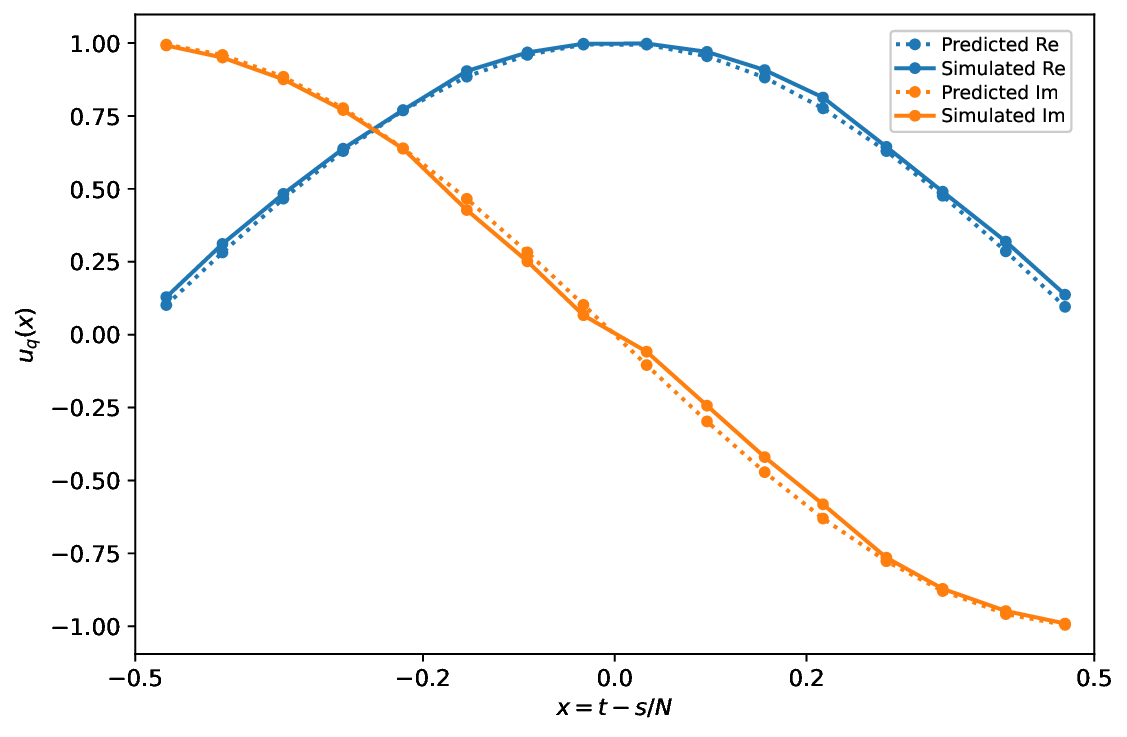}
        \caption{$n=4, p=5, m=4, q=0$}
    \end{subfigure}
    \hfill
    \begin{subfigure}{0.30\textwidth}
        \centering
        \includegraphics[width=\textwidth]{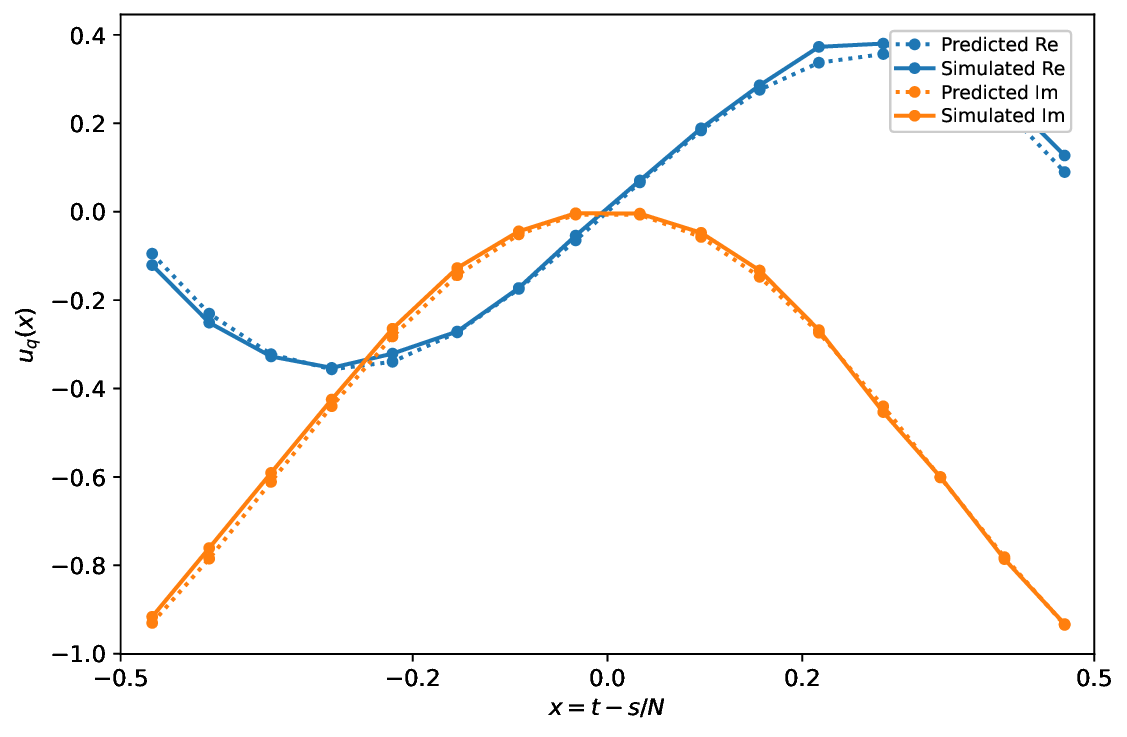}
        \caption{$n=4, p=5, m=4, q=1$}
    \end{subfigure}
    \hfill
    \begin{subfigure}{0.30\textwidth}
        \centering
        \includegraphics[width=\textwidth]{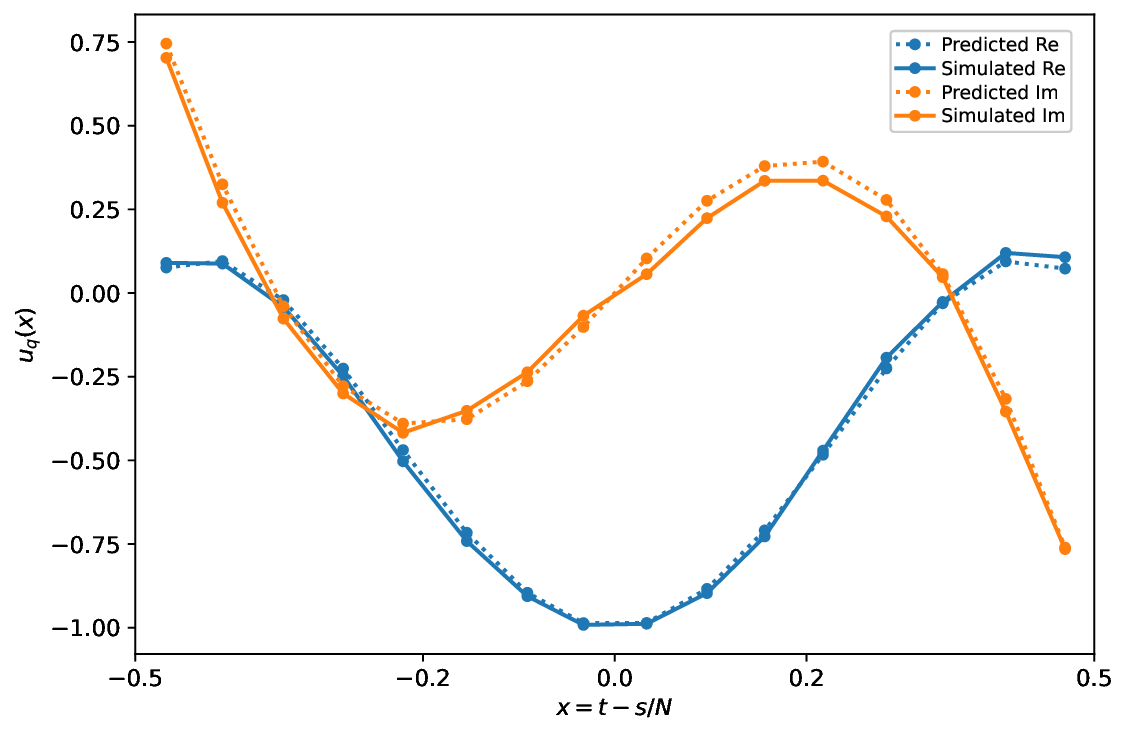}
        \caption{$n=4, p=5, m=4, q=2$}
    \end{subfigure}

    \vspace{0.4cm}
    \begin{subfigure}{0.30\textwidth}
        \centering
        \includegraphics[width=\textwidth]{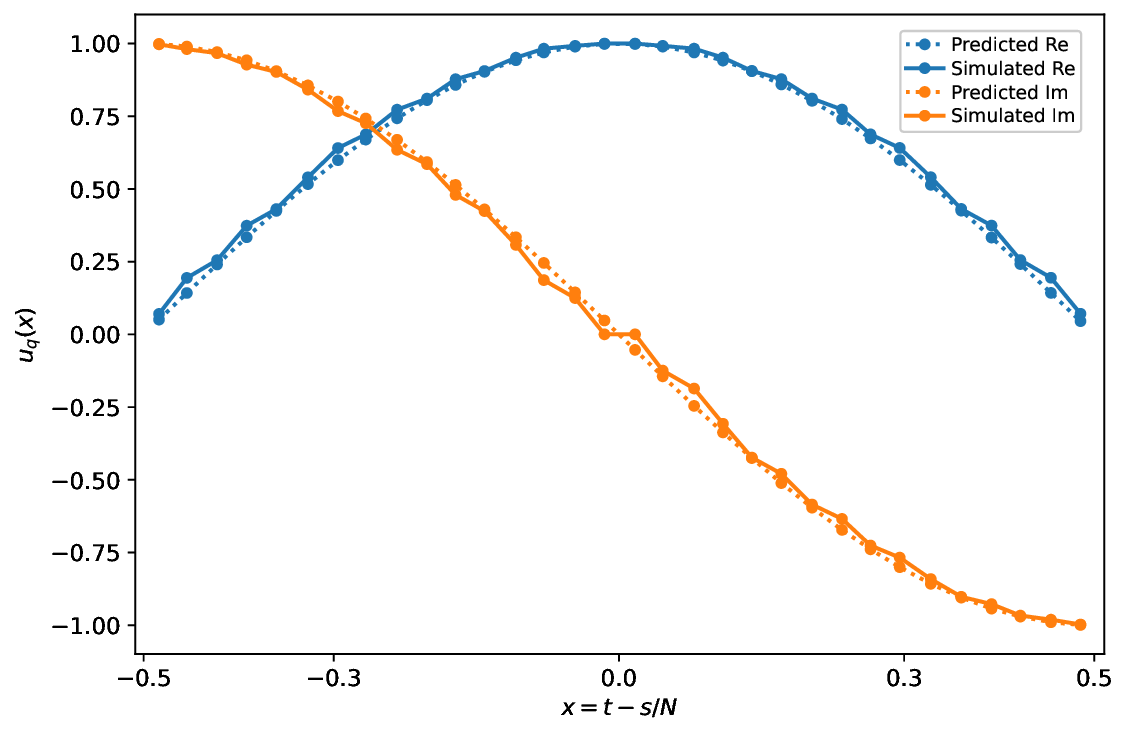}
        \caption{$n=5, p=5, m=4, q=0$}
    \end{subfigure}
    \hfill
    \begin{subfigure}{0.30\textwidth}
        \centering
        \includegraphics[width=\textwidth]{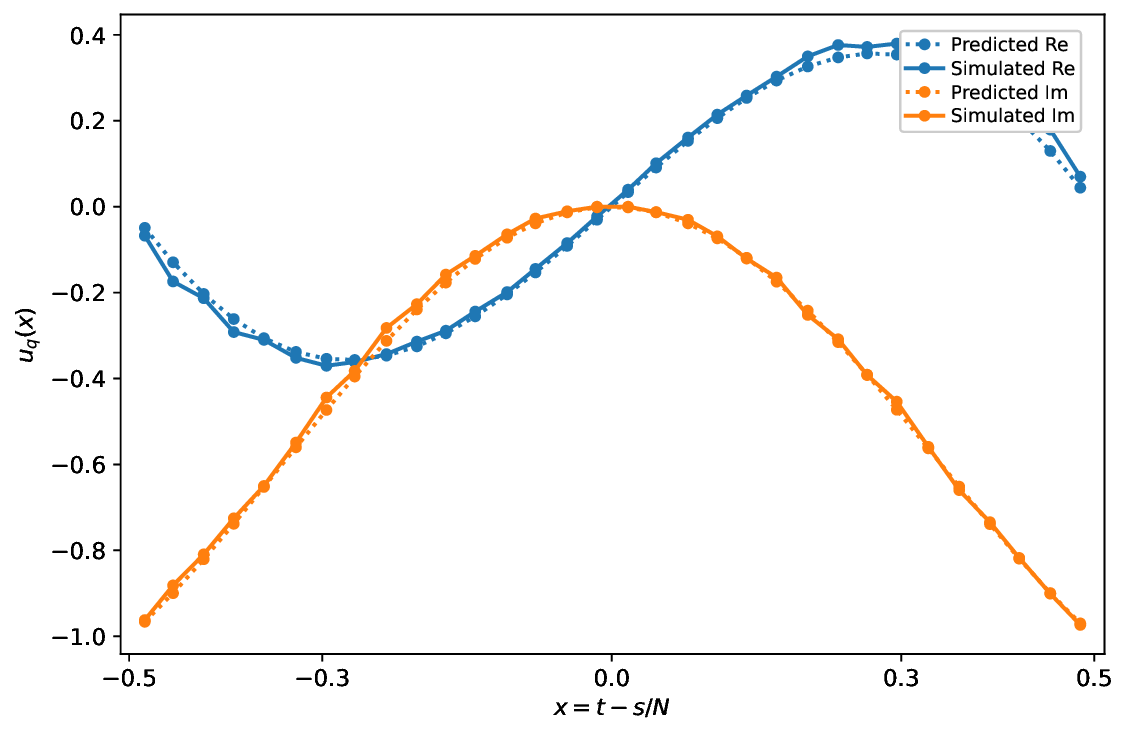}
        \caption{$n=5, p=5, m=4, q=1$}
    \end{subfigure}
    \hfill
    \begin{subfigure}{0.30\textwidth}
        \centering
        \includegraphics[width=\textwidth]{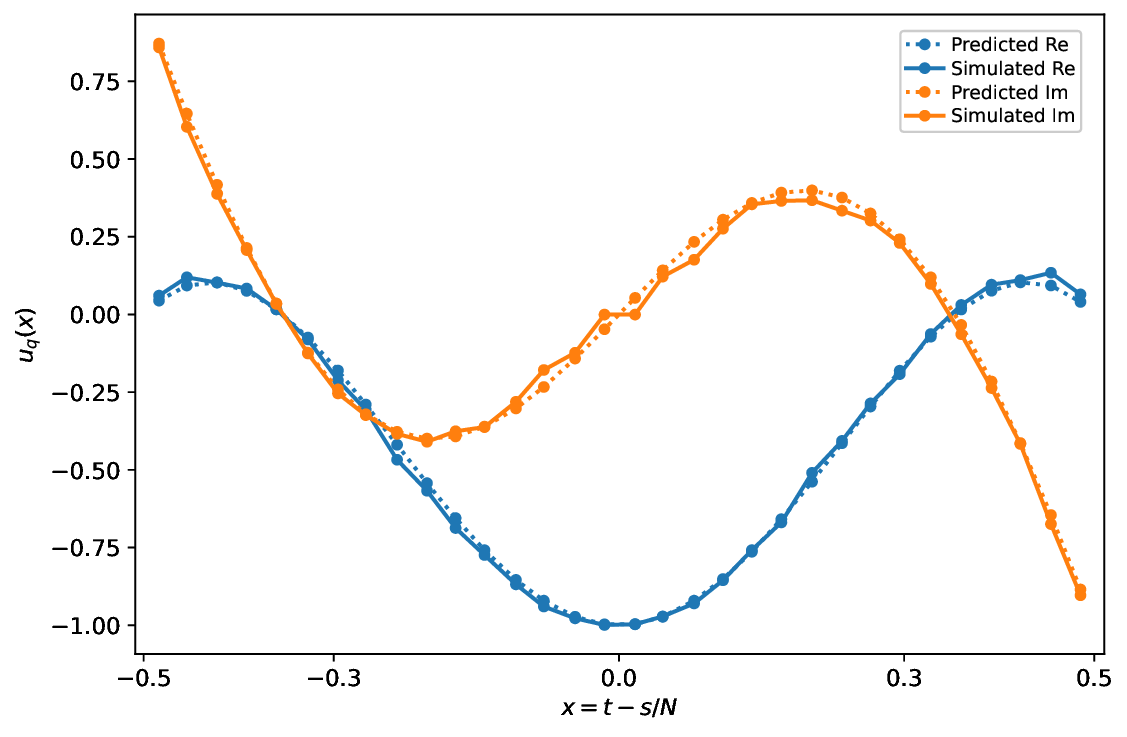}
        \caption{$n=5, p=5, m=4, q=2$}
    \end{subfigure}

    \vspace{0.4cm}
    \begin{subfigure}{0.30\textwidth}
        \centering
        \includegraphics[width=\textwidth]{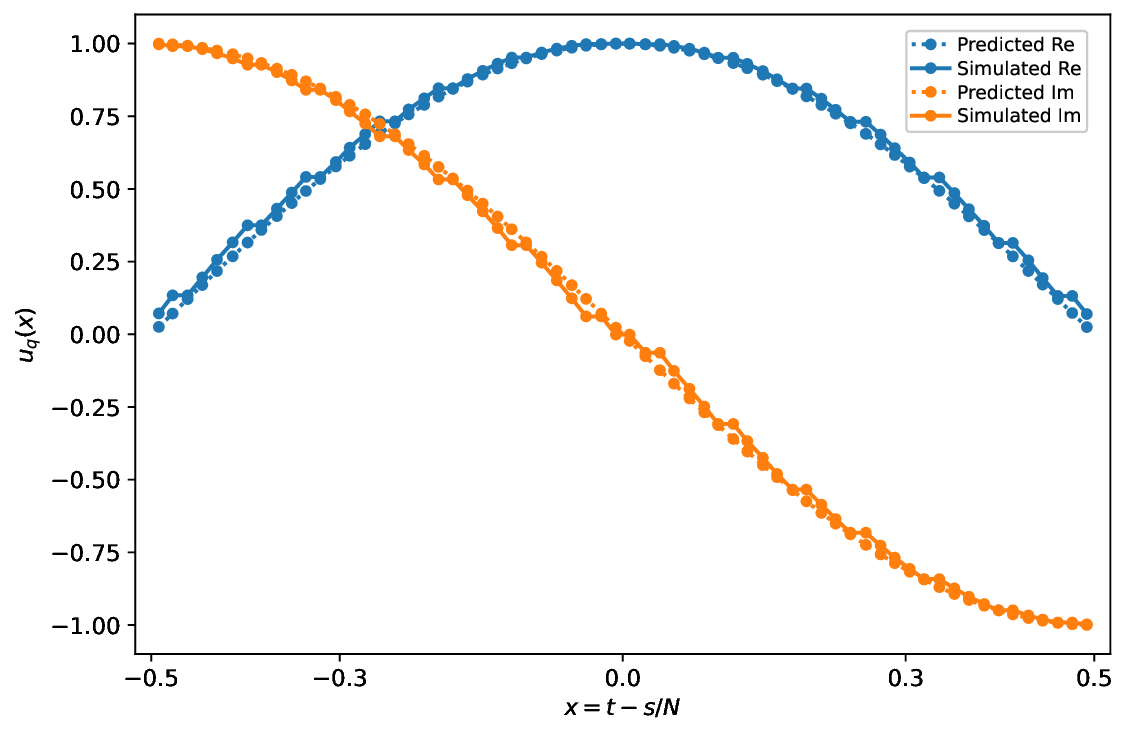}
        \caption{$n=6, p=5, m=4, q=0$}
    \end{subfigure}
    \hfill
    \begin{subfigure}{0.30\textwidth}
        \centering
        \includegraphics[width=\textwidth]{Figures/C4a/uqr_amp_domain_case008_n5_m4_p5_q1.eps}
        \caption{$n=6, p=5, m=4, q=1$}
    \end{subfigure}
    \hfill
    \begin{subfigure}{0.30\textwidth}
        \centering
        \includegraphics[width=\textwidth]{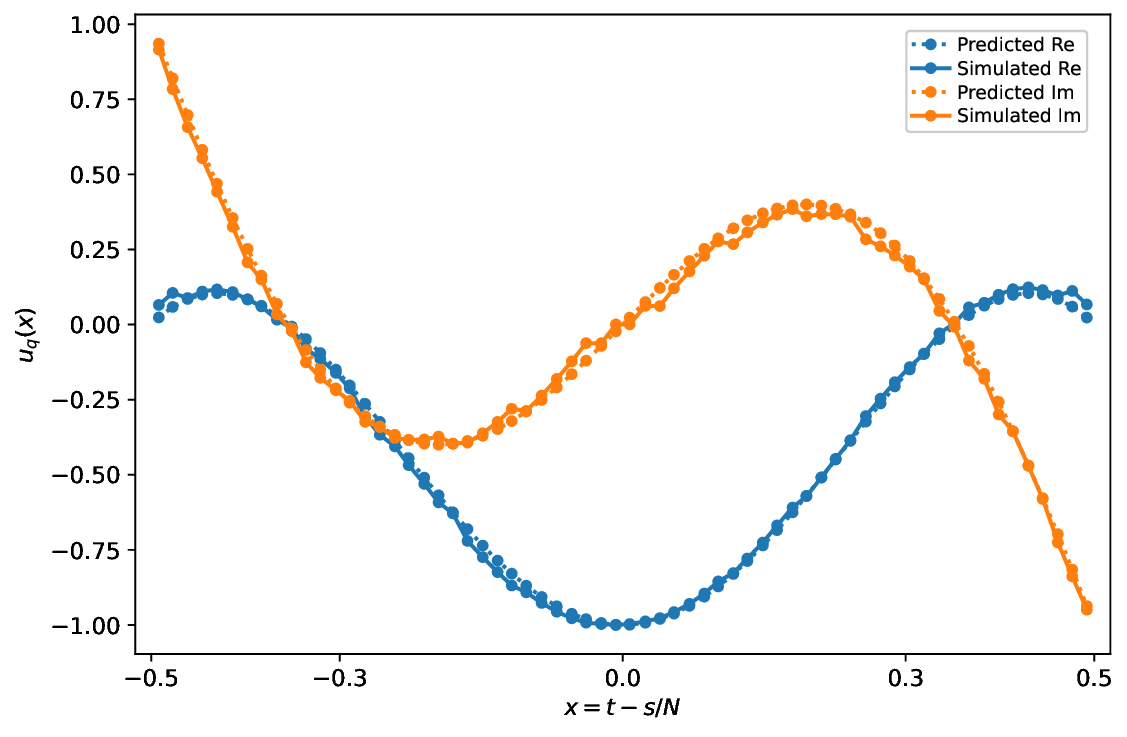}
        \caption{$n=6, p=5, m=4, q=2$}
    \end{subfigure}
    \caption{
    Comparison of computed and predicted outputs for different parameters of the quantum circuit $U_{q}$ for $n \in \{4,5,6\}$ and $q \in \{0,1,2\}$ with $m=4$ and $p=5$.
    }
    \label{fig:ur-plots}
\end{figure}

\begin{table}[h]
    \renewcommand{\arraystretch}{1.25}
    \centering
    \scalebox{0.80}
    {
    \begin{tabular}{c@{\hspace{1.5cm}}c}
    \begin{tabular}{c|ccccc|c}\hline\hline
    \# & $n$ & $m$ & $p$ & $q$ & $\text{Error}$ & $\text{Ratio}$ \\\hline
    1 & 4 & 2 & 1 & 2 & 1.6138906535 & 2.00799 \\
    2 & 4 & 2 & 2 & 2 & 0.8036549819 & 1.30723 \\
    3 & 4 & 2 & 3 & 2 & 0.6147763375 & 1.28833 \\
    4 & 4 & 2 & 4 & 2 & 0.4771891201 & 1.06065 \\
    5 & 4 & 2 & 5 & 2 & 0.4499019429 & 1.03347 \\
    6 & 4 & 2 & 6 & 2 & 0.4353344292 & 1.00562 \\
    7 & 4 & 2 & 7 & 2 & 0.4329028110 & -- \\
    \hline\hline
    \end{tabular}
    &
    \begin{tabular}{c|ccccc|c}\hline\hline
    \# & $n$ & $m$ & $p$ & $q$ & $\text{Error}$ & $\text{Ratio}$ \\\hline
    1 & 4 & 1 & 2 & 2 & 1.1343486555 & 1.41149 \\
    2 & 4 & 2 & 2 & 2 & 0.8036549819 & 1.12381 \\
    3 & 4 & 3 & 2 & 2 & 0.7151168695 & 1.04547 \\
    4 & 4 & 4 & 2 & 2 & 0.6840150547 & 1.01223 \\
    5 & 4 & 5 & 2 & 2 & 0.6757494378 & 1.00142 \\
    6 & 4 & 6 & 2 & 2 & 0.6747895628 & 1.00031 \\
    7 & 4 & 7 & 2 & 2 & 0.6745814549 & -- \\
    \hline\hline
    \end{tabular}
    \end{tabular}
    }
    \caption{State-vector error for fixed $n=4$ and $q=2$. 
    The left table varies $p$ with $m=2$ fixed, while the right table varies $m$ with $p=2$ fixed. 
    The last column shows the ratio of consecutive rows.}
    \label{tab:ur-p-m-error-ratio}
\end{table}

We next validate $U_{u_{\vec r}} = \sum_{q=0}^{K-1} \alpha'_{qr} U_q$ using the LCU algorithm. Let 
\begin{equation}
u_r(\vec{x}) = \sqrt{N} \| \vec \alpha_{qr} \|_1 \sum_{q=0}^{K-1} \alpha'_{qr} e^{-i \pi N \vec{x}}\, T_q(2N \vec{x})
\end{equation}
denote the normalized amplitude of the $\ket{0}$ basis state of the circuit $U_{u_{\vec r}}$. For $K=2$, we report the results in \cref{fig:ur-lcu}, which provide evidence that
$U_{u_{\vec r}}$ accurately reproduces
$u_{r}(\vec x)/\| \alpha_{qr} \|_1 \sqrt{N}$.
\begin{figure}[h]
    \centering

    \begin{subfigure}{0.30\textwidth}
        \centering
        \includegraphics[width=\textwidth]{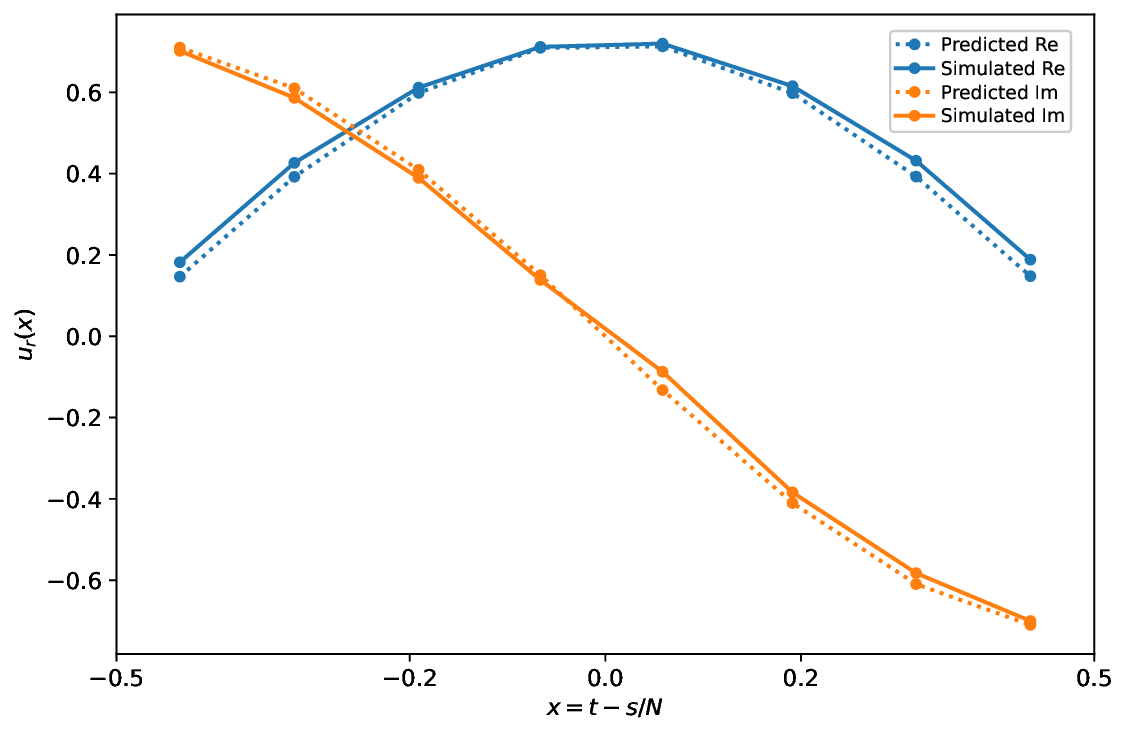}
        \caption{$n=3, p=5, m=4, r=0$}
    \end{subfigure}
    \hfill
    \begin{subfigure}{0.30\textwidth}
        \centering
        \includegraphics[width=\textwidth]{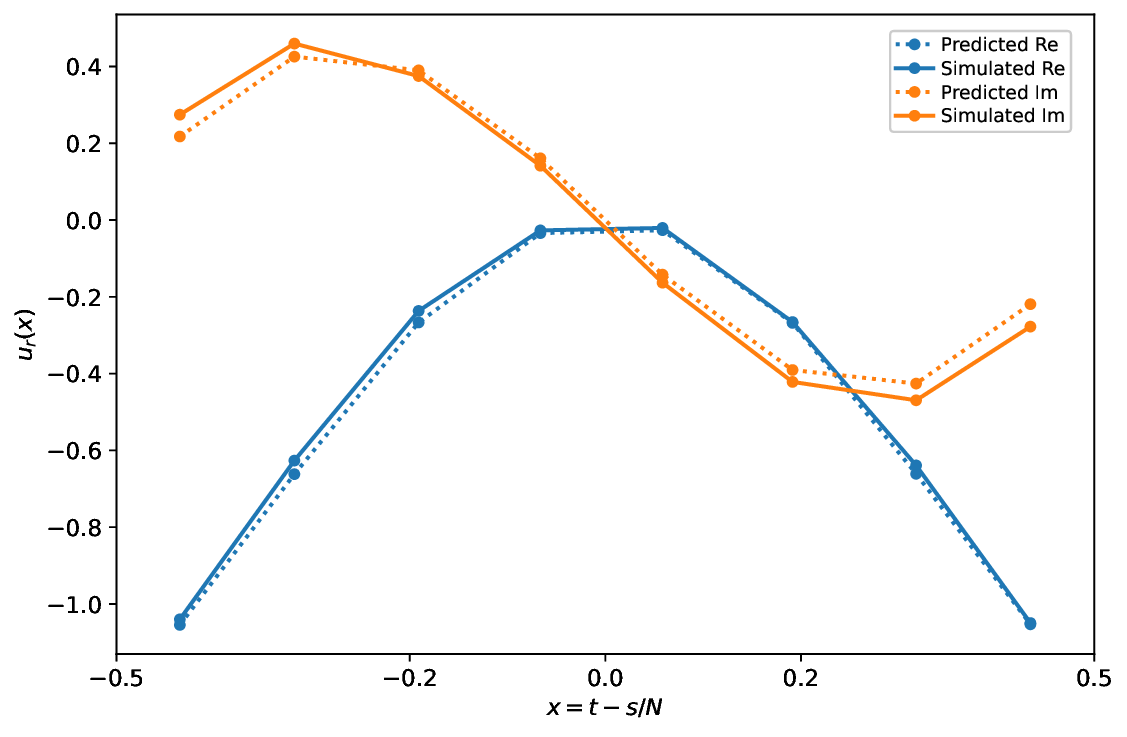}
        \caption{$n=3, p=5, m=4, r=1$}
    \end{subfigure}
    \hfill
    \begin{subfigure}{0.30\textwidth}
        \centering
        \includegraphics[width=\textwidth]{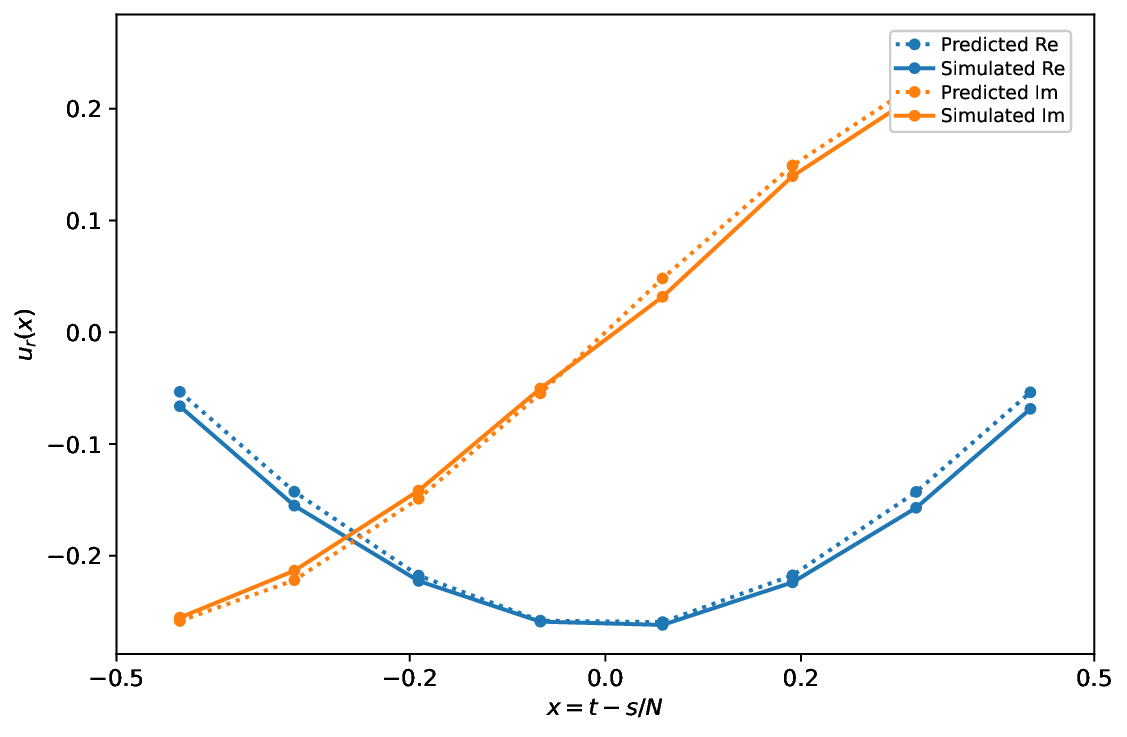}
        \caption{$n=3, p=5, m=4, r=2$}
    \end{subfigure}

    \vspace{0.4cm}

    \begin{subfigure}{0.30\textwidth}
        \centering
        \includegraphics[width=\textwidth]{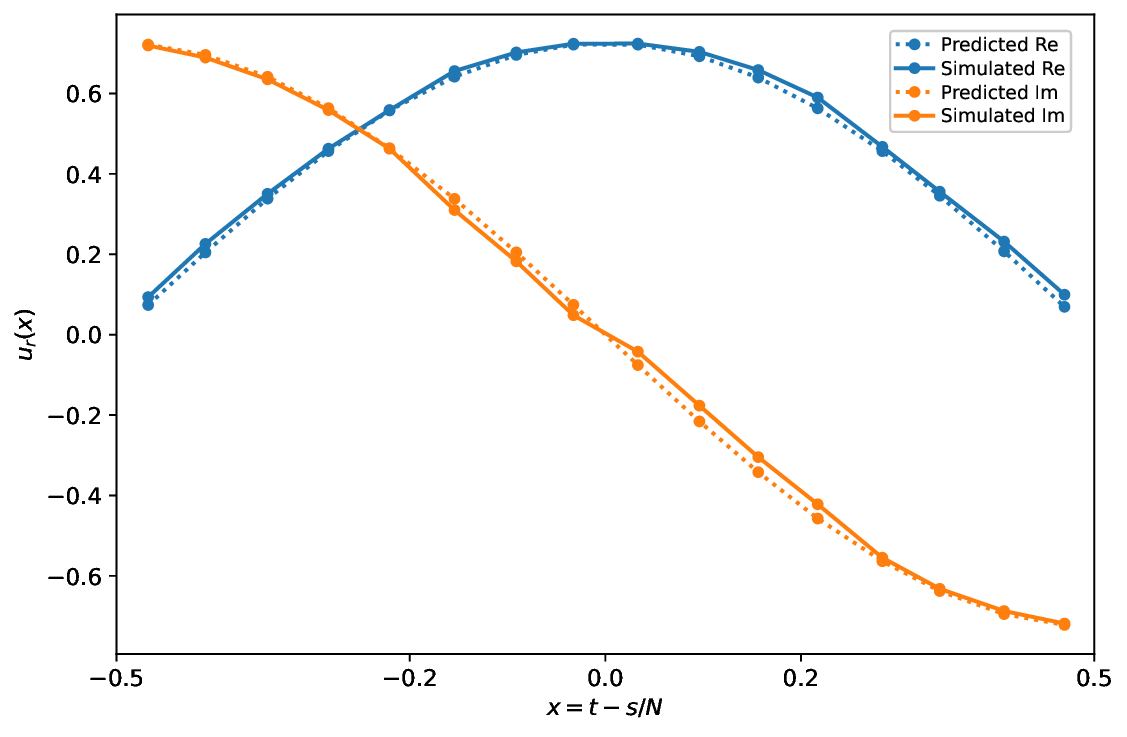}
        \caption{$n=4, p=5, m=4, r=0$}
    \end{subfigure}
    \hfill
    \begin{subfigure}{0.30\textwidth}
        \centering
        \includegraphics[width=\textwidth]{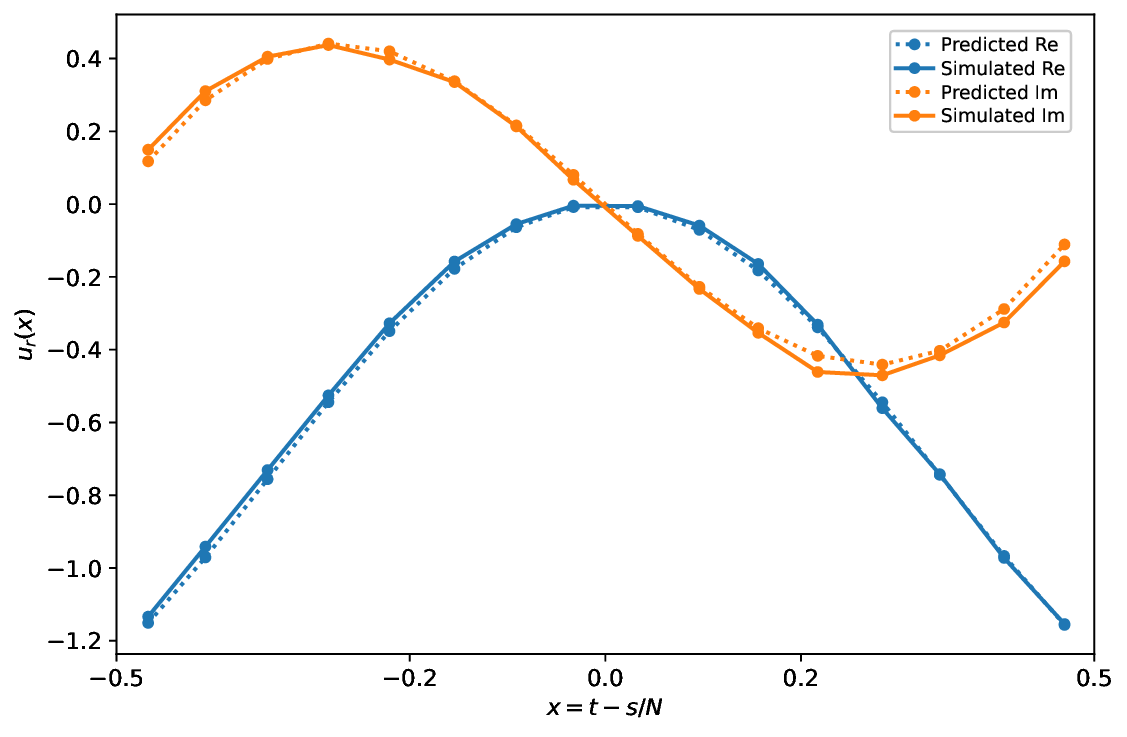}
        \caption{$n=4, p=5, m=4, r=1$}
    \end{subfigure}
    \hfill
    \begin{subfigure}{0.30\textwidth}
        \centering
        \includegraphics[width=\textwidth]{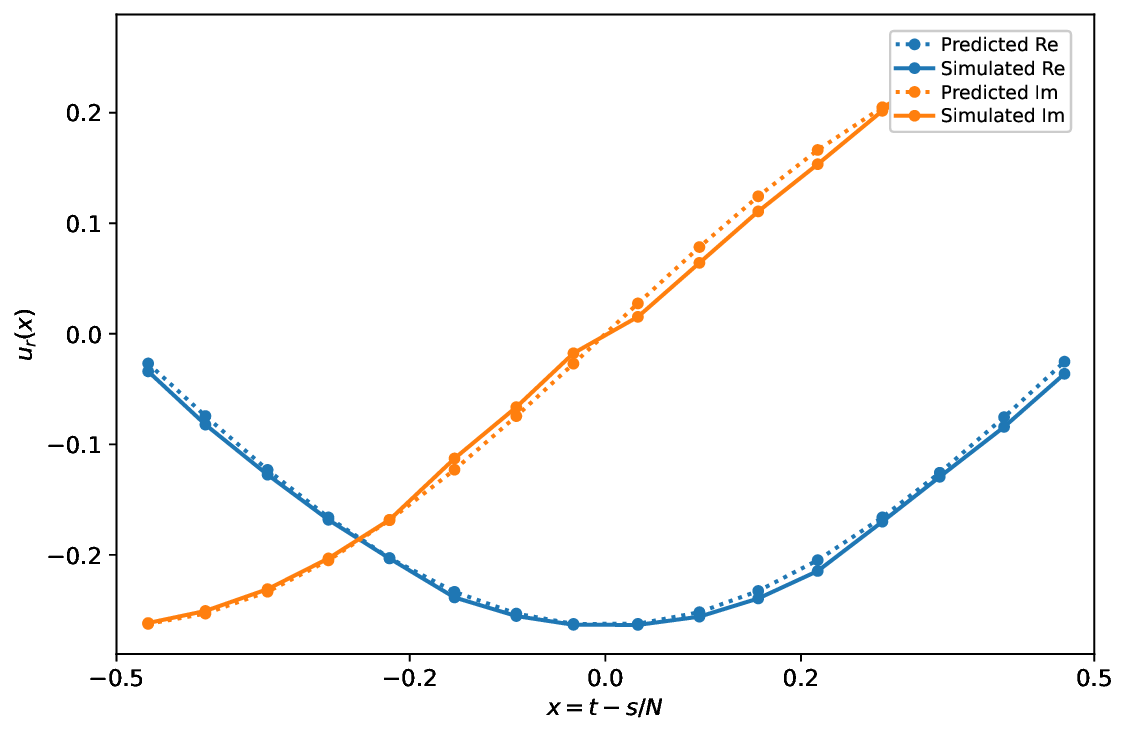}
        \caption{$n=4, p=5, m=4, r=2$}
    \end{subfigure}

    \vspace{0.4cm}

    \begin{subfigure}{0.30\textwidth}
        \centering
        \includegraphics[width=\textwidth]{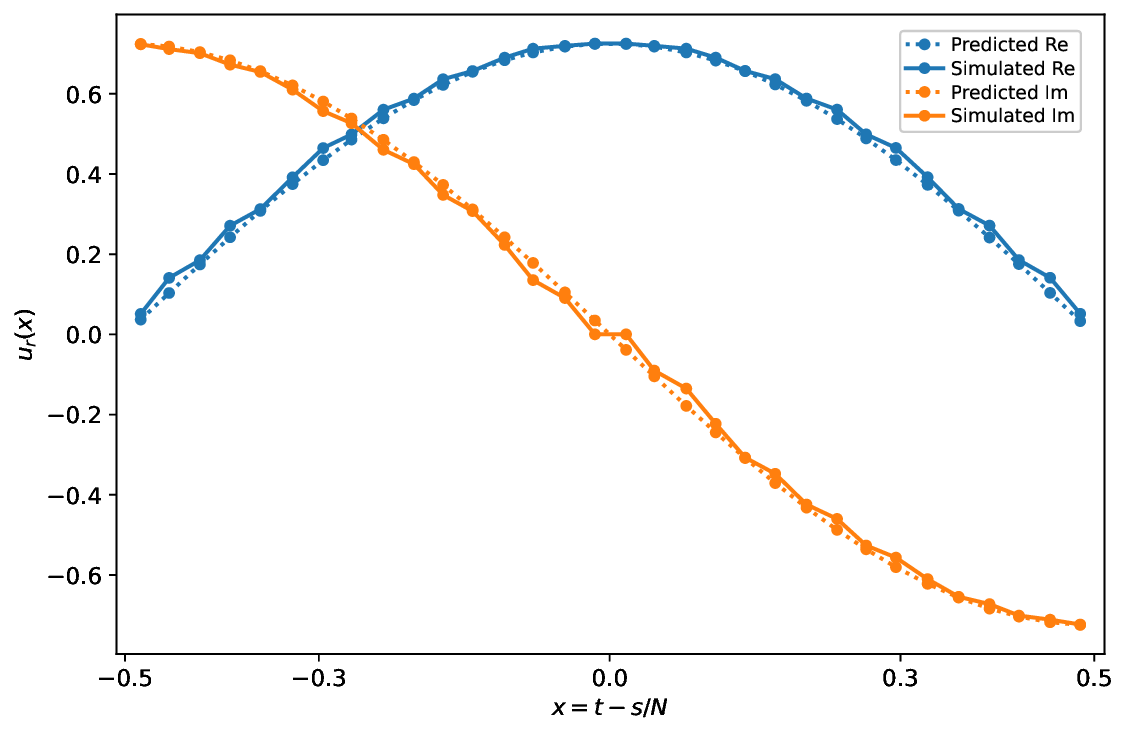}
        \caption{$n=5, p=5, m=4, r=0$}
    \end{subfigure}
    \hfill
    \begin{subfigure}{0.30\textwidth}
        \centering
        \includegraphics[width=\textwidth]{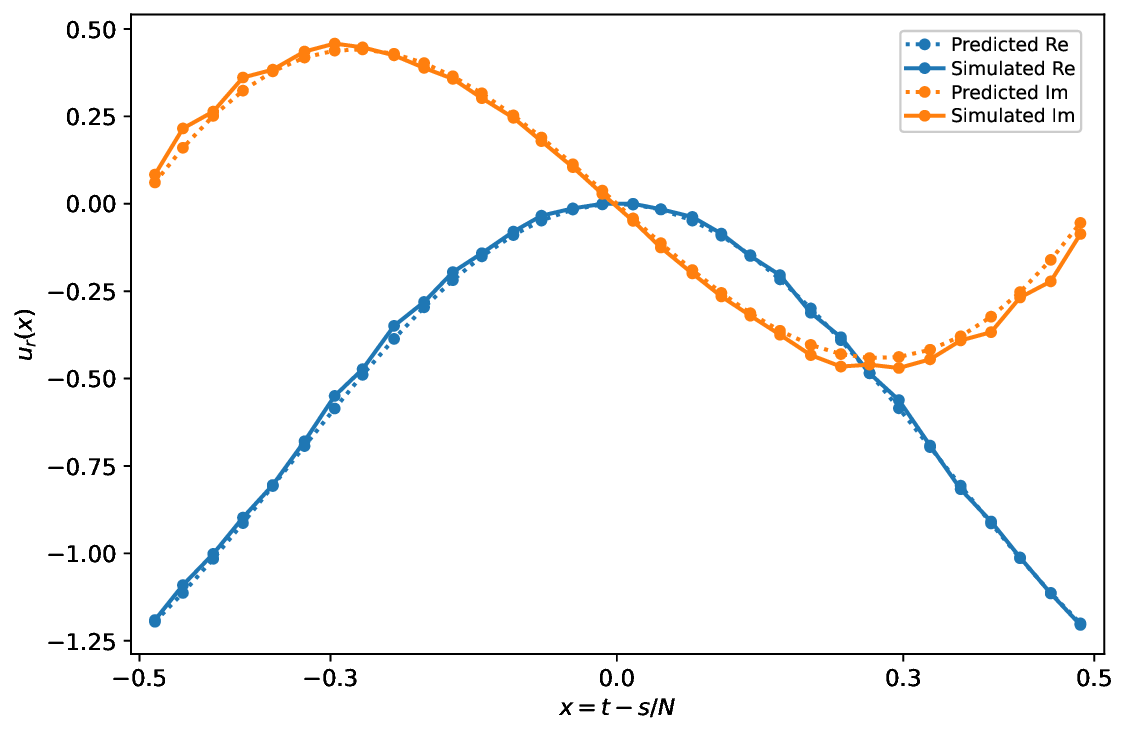}
        \caption{$n=5, p=5, m=4, r=1$}
    \end{subfigure}
    \hfill
    \begin{subfigure}{0.30\textwidth}
        \centering
        \includegraphics[width=\textwidth]{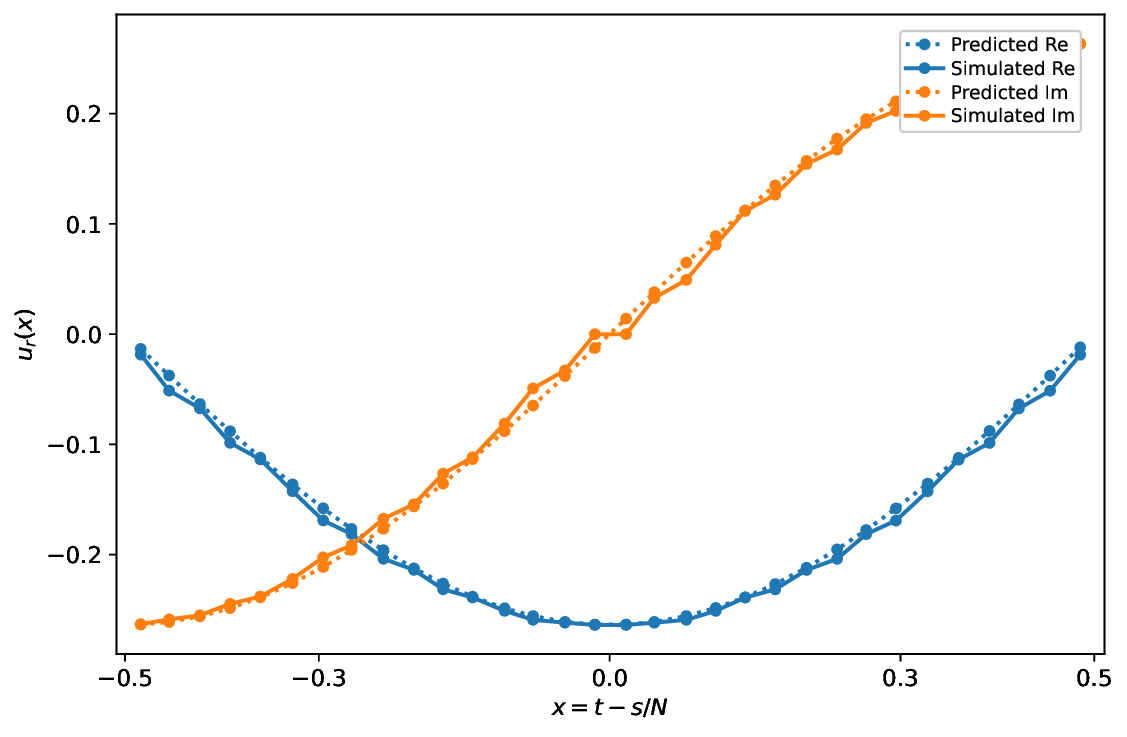}
        \caption{$n=5, p=5, m=4, r=2$}
    \end{subfigure}
    \caption{Comparison of the computed and predicted for different parameters of the quantum circuit $U_{u_{\vec r}}$ for $n \in \{4,5,6\}$ and $r \in \{0,1,2\}$ with $m=4$ and $p=5$.}
    \label{fig:ur-lcu}
\end{figure}
Note that the LCU algorithm dominates the computational cost of implementing the method. This is because the LCU procedure requires multi-qubit controlled unitaries which, when decomposed into elementary gates, reduce to multi-qubit controlled-$X$ gates, since the arithmetic loader ultimately consists of such operations. Given current simulation capabilities, this approach can only be implemented efficiently for small values of $K$ and $n$. In particular, we restrict to the case $K=2$.

\subsection{Numerical Validation of the $M_s$ Block Encoding}
\label{sec:ms-block-encoding-validation}
Lastly, we validate the block encoding constructed for the matrix $M_{\vec s}$. The implemented circuit exploits the structure of $M_{\vec s}$. Since each column has exactly one nonzero entry, the block encoding is realized via a simplified right--isometry. Although inspired by the sparse-access block-encoding lemma, the implementation does not explicitly realize the oracle triplet $(O_c, O_r, O_A)$. Instead, it uses the column structure of $M_{\vec s}$ to build a direct circuit. The resulting unitary $U_{M_{\vec s}}$ uses $n+1$ ancilla qubits and satisfies
\begin{equation}\label{to-very}
(\langle 0|^{n+1}  \otimes I) U_{M_{\vec s}} (|0\rangle^{n+1}  \otimes I)
= \frac{M_s}{\sqrt{d_r}}, 
\quad \quad d_r = \| \vec c_{s} \|_\infty.  
\end{equation}

We verify \cref{to-very} by simulating $U_{M_{\vec s}}$, extracting the encoded matrix by projecting the ancilla register onto $|0\rangle^{n+1}$, and comparing it with the normalized target matrix $M_{\vec s}/\sqrt{d_r}$. We quantify the error using the maximum entrywise norm and the spectral norm. The construction is tested on representative vectors $\vec s$ with different sparsity patterns:

\begin{enumerate}
\item \textbf{Identity:} $s_j=j$, yielding the identity matrix ($d_r=1$).
\item \textbf{Bit-Reversal:} $s_j$ is the bit-reversal of $j$ (permutation, $d_r=1$).
\item \textbf{Constant:} all entries of $s$ are identical, giving maximal multiplicity $d_r=N$.
\item \textbf{Two-Cluster:} half the columns map to one row and half to another ($d_r=N/2$).
\item \textbf{Heavy Row:} most columns map to one row, with the remainder distributed.
\item \textbf{Random:} $s_j$ sampled uniformly from $\{0,\dots,N-1\}$.
\end{enumerate}

We report representative results for system size $n=5$ in \cref{tab:ms-block-encoding-validation}. The extracted encoded block agrees with the expected matrix $M_s/\sqrt{d_r}$ to machine precision for all tested configurations.
\begin{table}[h]
    \renewcommand{\arraystretch}{1.25}
    \centering
    \scalebox{0.85}{
    \begin{tabular}{c|c|c|c|c|c}\hline\hline
        \# & Family & $d_r$ & $\alpha=\sqrt{d_r}$ & Max error & Spectral error \\\hline
        1 & Identity       & 1  & 1.000000 & $0.0$           & $0.0$ \\
        2 & Bit reversal   & 1  & 1.000000 & $0.0$           & $0.0$ \\
        3 & Constant ($s_j=0$)   & 32 & 5.656854 & $3.9\times10^{-16}$ & $1.7\times10^{-15}$ \\
        4 & Constant ($s_j=N-1$) & 32 & 5.656854 & $3.9\times10^{-16}$ & $1.7\times10^{-15}$ \\
        5 & Two cluster    & 16 & 4.000000 & $7.3\times10^{-16}$ & $2.0\times10^{-15}$ \\
        6 & Heavy row      & 17 & 4.123106 & $6.9\times10^{-16}$ & $1.9\times10^{-15}$ \\
        7 & Random         & 3  & 1.732051 & $1.6\times10^{-14}$ & $2.4\times10^{-14}$ \\\hline\hline
    \end{tabular}}
    \caption{Numerical validation of the $M_{\vec s}$ block encoding for $n=5$.}
    \label{tab:ms-block-encoding-validation}
\end{table}

\begin{remark}
The block encoding QFT is straightforward. 
A complete block encoding of the NUQFT matrix can be obtained by composing the components simulated above. We defer the full construction to future work, as the simulation becomes computationally expensive for large $K$.
\end{remark}

\section{Conclusion}\label{conclusion}
We have developed a quantum algorithm for implementing the Type-II non-uniform discrete Fourier transform (NUDFT), which can be readily extended to the Type-I and Type-III variants. We have also presented a rigorous error analysis that identifies the dominant sources of approximation error and derives explicit upper bounds on their respective contributions. The analysis shows that the required computational resources scale logarithmically with the inverse of the target precision and polynomially with the number of qubits.

\subsection{Future Directions}
We discuss below several future research directions:

\begin{enumerate}
\item 
The proposed NUQFT algorithm assumes oracle access to the non-uniform sampling points (cf.~\cref{assumption-t}) and to the corresponding inverse mapping (cf.~\cref{assumption-Or}). An important direction for future work is to identify concrete computational scenarios in which such oracles can be implemented efficiently, either through explicit circuit constructions or through data-generation models where the required access arises naturally.

\item
The present framework relies on Chebyshev polynomial approximations implemented via quantum signal processing. Alternative approximation strategies, including rational approximations, Fourier--Bessel expansions, or minimax constructions tailored to restricted domains, may offer improved constants or reduced ancilla usage. Exploring how such schemes can be incorporated into QSP or the more general QSVT framework may lead to more resource-efficient NUQFT implementations.


\item
The analysis is carried out in a fault-tolerant setting with idealized primitives, including exact QFT implementations and perfect oracle access. Future work should address explicit compilation into standard fault-tolerant gate sets, with attention to optimizing \(T\)-count and \(T\)-depth, as well as adapting the construction to hardware constraints such as limited connectivity. Extending the framework to NISQ devices is another important direction. This includes incorporating approximate implementations of the Quantum Fourier Transform (QFT) and approximate rotations (see.~\cref{rotation-approximate}) by eliminating small-angle rotations to reduce circuit depth.

\item  The numerical experiments presented here are limited by the qubit counts and memory capacity of current quantum simulators and hardware. Future work may explore more scalable implementations. In particular, more resource-efficient realizations of key arithmetic primitives---especially the $\arccos$ subroutine and the LCU stage---could substantially reduce qubit and gate overhead. Possible directions include improved polynomial approximations, QSVT-based constructions, circuit-level optimizations, or hybrid LCU implementations. Such improvements could enable larger-scale realizations and provide a clearer evaluation of the practical performance of the proposed approach.

\item
Many applications of non-uniform Fourier analysis arise in higher dimensions. Extending the NUQFT framework to multidimensional NUDFTs, for example via tensorized or separable low-rank decompositions, is therefore a natural next step. More broadly, the techniques developed here may generalize to other non-uniform integral transforms, including cosine and Laplace transforms, discrete Legendre transforms, and Fourier transforms defined on manifolds.

\item
Finally, the NUQFT is intended to function as a subroutine within larger quantum algorithms rather than as an isolated primitive. Studying its integration into complete quantum workflows—including state preparation, intermediate processing, and measurement—remains an important open problem, particularly in identifying application domains where such integration yields a provable advantage.
\end{enumerate}

\printbibliography

\end{document}